\documentclass[12pt]{article}

\usepackage{natbib}
\usepackage{xcolor}
\usepackage{url}
\usepackage{amsmath,amssymb,amsfonts,amsthm}
\usepackage{setspace}
\usepackage{multirow}
\usepackage{threeparttable}
\usepackage{graphicx}
\usepackage{enumerate}
\usepackage{color}
\usepackage{verbatim}
\usepackage{comment}
\usepackage{bbm}
\usepackage{booktabs}
\usepackage[colorlinks,citecolor=blue,urlcolor=blue]{hyperref}

\newcommand{\blind}{1}

\newtheorem{theorem}{Theorem}
\newtheorem{lemma}{Lemma}

\newtheorem{prop}{Proposition}
\newtheorem{assumption}{Assumption}

\newtheorem{remark}{Remark}
\newtheorem{corollary}{Corollary}

\def\ps{p_{[k]}}

\def\sumn{\sum_{i=1}^n}
\def\sumk{\sum_{k=1}^{K}}
\def\sumis{\sum_{i:S_i = k}}
\def\pins{\pi_{n[k]}}
\def\pns{p_{n[k]}}
\def\ps{p_{[k]}}
\def\ns{n_{[k]}}
\def\nst{n_{[k]1}}
\def\nsc{n_{[k]0}}
\def\dim{\textnormal{tdim}}
\def\str{\textnormal{str}}
\def\Var{\textnormal{Var}}
\def\ncone{n_{0(1)}}
\def\nscone{n_{[k]0(1)}}
\def\cd{\xrightarrow{d}}
\def\cp{\xrightarrow{p}}
\def\op{o_{P}(1)}

\setcitestyle{authoryear, round, comma, aysep={;}, yysep={,}, notesep={, }}

\begin{document}
	
	\def\spacingset#1{\renewcommand{\baselinestretch}
		{#1}\small\normalsize} \spacingset{1}

	\if1\blind
	{
		\title{\bf Treatment effect estimation under covariate-adaptive randomization with heavy-tailed outcomes}
		\author{Hongzi Li\\
			Center for Statistical Science, Department of Industrial Engineering, \\
   Tsinghua University, Beijing, 100084, China\vspace{.1cm}\\
   Wei Ma\\
   Institute of Statistics and Big Data,  Renmin University of China, \\
   Beijing, 100872, China\vspace{.1cm}\\
   Yingying Ma\\
   School of Economics and Management, Beihang University,\\ Beijing, 100191, China\vspace{.1cm}\\
			Hanzhong Liu\thanks{Dr. Liu was supported by the National Natural Science Foundation of China (12071242).}\\
			Center for Statistical Science, Department of Industrial Engineering, \\
			Tsinghua University, Beijing, 100084, China\vspace{.1cm}\\
			}
		\date{}
		\maketitle
	}\fi

	\if0\blind
	{
		\title{\bf Treatment effect estimation under covariate-adaptive randomization with heavy-tailed outcomes}
  \date{}
  \maketitle
	}\fi
	
 \vspace{-0.5cm}
	
	\begin{abstract}

Randomized experiments are the gold standard for investigating causal relationships, with comparisons of potential outcomes under different treatment groups used to estimate treatment effects. However, outcomes with heavy-tailed distributions pose significant challenges to traditional statistical approaches. While recent studies have explored these issues under simple randomization, their application in more complex randomization designs, such as stratified randomization or covariate-adaptive randomization, has not been adequately addressed. To fill the gap, this paper examines the properties of the estimated influence function-based M-estimator under covariate-adaptive randomization with heavy-tailed outcomes, demonstrating its consistency and asymptotic normality. Yet, the existing variance estimator tends to overestimate the asymptotic variance, especially under more balanced designs, and lacks universal applicability across randomization methods. To remedy this, we introduce a novel stratified transformed difference-in-means estimator to enhance efficiency and propose a universally applicable variance estimator to facilitate valid inferences. Additionally, we establish the consistency of kernel-based density estimation in the context of covariate-adaptive randomization.  Numerical results demonstrate the effectiveness of the proposed methods in finite samples.

	\end{abstract}
	
	
	\newpage
	\spacingset{1.9} 	
	
	\section{Introduction}

Randomization is the gold standard for drawing causal inference, often utilizing comparisons of potential outcomes across different treatment groups to estimate treatment effects \citep{Neyman:1923, fisher1935design, rosenbaum2002observational, Imbens2015}. However, classical methods based on large sample theory for valid inferences assume well-behaved potential outcomes, a condition often violated by heavy-tailed distributions common in fields like economics, social sciences, and clinical trials. For example, factors such as payment amounts or customers' spending abilities in social science experiments \citep{athey2021semiparametric}, and data like CD4 counts in HIV studies \citep{henry1998randomized}, often exhibit heavy-tailed distributions, posing significant challenges to traditional statistical methods.

While heavy-tailed distributions have been extensively studied in association studies \citep[see, e.g.,][]{pickands1975statistical, hill1975simple,embrechts2013modelling,coles2001introduction}, recent years have witnessed a growing interest in their impact on causal inference. Seminal works by \cite{athey2021semiparametric} and \cite{ghosh2021efficiency} proposed novel methods for estimating treatment effects in randomized experiments with heavy-tailed outcomes. However, these studies primarily focus on simple randomization, neglecting the complexities introduced by stratified or more general covariate-adaptive randomization.

Covariate-adaptive randomization methods are widely used in experimental designs for their ability to balance baseline covariates, such as gender and age, across treatment groups, thereby enhancing estimation and inference efficiencies \citep{Bugni2018,Bugni2019}. For example, the stratified biased-coin design partitions samples into different strata based on covariates and then independently assigns treatments within each stratum using a biased coin to ensure balance within strata \citep{Efron1971}. Stratified permuted block randomization divides each stratum into blocks of pre-specified sizes and allocates treatments within each block \citep{Zelen1974}. Pocock and Simon's minimization method focuses on minimizing imbalances between the number of participants in each treatment group across several prognostic factors \citep{Pocock1975}. According to  \cite{ciolino2019ideal}, nearly $80\%$ of clinical trials published in leading medical journals in 2009 and 2014 utilized covariate-adaptive randomization strategies.

Recent years have witnessed significant advances in achieving valid inferences under covariate-adaptive randomization. For instance, \cite{Bugni2018} developed the groundwork by establishing a new law of large numbers and central limit theory tailored for covariate-adaptive randomization. Their study explored the asymptotic properties of the two-sample $t$-test and regression-adjusted estimators that adjust for stratification indicators,  without assuming a correct underlying outcome data-generating model. \cite{Bugni2019} proposed a fully saturated regression estimator through a two-step procedure to enhance efficiency. \cite{Ma2020Regression}, \cite{Ye2020Inference}, \cite{liu2020general}, and \cite{Gu2024} further studied various regression methods to address imbalances caused by additional covariates beyond stratification variables. Their findings demonstrated that these regression-adjusted estimators could achieve improved efficiency even under misspecified linear models. However, it's important to note that all these methods and theories require finite second moments of outcomes and are not applicable when outcomes exhibit heavy tails. To overcome this challenge, researchers have proposed methods for estimating quantile treatment effects, as evidenced by works such as \cite{firpo2007efficient} and \cite{firpo2009unconditional} under simple randomization, as well as \cite{zhang2020quantile} and \cite{jiang2023regression} under covariate-adaptive randomization. Notably, these studies primarily focus on treatment effects at specific quantiles rather than the overall treatment effect.

\cite{athey2021semiparametric} introduced an M-estimation approach to estimate the overall average treatment effect with theoretical guarantees under simple randomization. Under the constant quantile treatment effects assumption, the overall treatment effect is equivalent to the treatment effect at any specific quantile. \cite{athey2021semiparametric} demonstrated that the M-estimator achieves the semiparametric efficiency bound, making it generally more efficient than the estimator of any particular quantile treatment effect. However, it remains unclear whether the point and variance estimators remain valid under covariate-adaptive randomization. The challenge stems from, unlike simple randomization, which independently assigns each unit to the treatment with a fixed probability, covariate-adaptive randomization introduces a complex dependence structure between covariates and treatment assignments. It determines treatment allocation for a unit not only based on its own covariates but also on the covariates and treatments of all previously assigned units. 

To fill the gap, we establish the asymptotic theory of the estimated influence function-based M-estimation approach \citep{athey2021semiparametric} under covariate-adaptive randomization. Given that the point estimator asymptotically equals the difference in the sample means of appropriately transformed outcomes under treatment and control, we term it the transformed difference-in-means estimator. Our analysis demonstrates that this estimator remains consistent and asymptotically normal under covariate-adaptive randomization, broadening its applicability. Moreover, we observe 
a decrease in the asymptotic variance as treatment allocation becomes more balanced, demonstrating the advantages of balanced designs. A related work by \cite{wang2021model} investigated the properties of a general class of M-estimators. However, their theory only applied to simple, stratified, or biased-coin randomization. More importantly, they assumed that the function used in the estimating equation is known. With heavy-tailed outcomes, this function depends on the unknown density function of the potential outcomes \citep{athey2021semiparametric}. To address this challenge, we establish the consistency of kernel-based density estimation and the $L_2$ convergence of the estimated score function under covariate-adaptive randomization. The results for density estimation are of particular interest in many other problems.

Several key issues arise with the estimated influence function-based M-estimation approach. Firstly, the conditions for deriving the asymptotic results preclude the applicability of Pocock and Simon's minimization method. Secondly, the asymptotic variance depends on the randomization methods and thus violates the universal applicability requirement outlined by \cite{Ye2020}, which asserts that the same inference procedure should apply to all commonly used covariate-adaptive randomization methods. Lastly, the variance estimator proposed by \cite{athey2021semiparametric} remains consistent only for simple randomization and tends to overestimate the asymptotic variance for more balanced designs.

To tackle these challenges, we propose a novel stratified transformed difference-in-means estimator and show that it is consistent and asymptotically normal. Notably, its asymptotic variance is independent of the randomization method and is either smaller than or equal to that of the transformed difference-in-means estimator, thereby enhancing the efficiency. Moreover, we provide a consistent nonparametric variance estimator, which is universally applicable across randomization methods, for constructing confidence intervals or performing hypothesis testing. Simulation studies and real data results demonstrate the validity of the proposed methods in finite samples.

The remainder of this paper is organized as follows: Section 2 outlines the covariate-adaptive randomization framework and introduces relevant notations. In Section 3, we introduce the transformed difference-in-means estimator based on the M-estimation approach, propose the stratified transformed difference-in-means estimator, and establish their asymptotic theory under covariate-adaptive randomization. Section 4 presents the simulation results, while Section 5 reports the results of a real data example. Section 6 concludes the paper with discussions of future research. Proofs are relegated to the supplementary material.

\section{Framework and notation}
In a covariate-adaptive randomized experiment, suppose that $n$ units are drawn independently from a super-population, and each unit is assigned to the treatment or control group. Let $A_i$ represent the treatment assignment indicator for unit $i$, with $A_i=1$ indicating assignment to the treatment group and $A_i=0$ to the control group. The potential outcomes under treatment and control are denoted as $Y_i(1)$ and $Y_i(0)$, respectively, and the observed outcome is given by $Y_i=A_iY_i(1)+(1-A_i)Y_i(0)$. The $n$ units are divided into $K$ strata based on important baseline covariates like age, gender, or race. Let $S_i$ denote the stratum indicator, such that $S_i=k$ ($k=1,\ldots,K$) if unit $i$ belongs to stratum $k$. The probability of a unit belonging to stratum $k$ is denoted as $\ps = P(S_i=k)$. We assume that $K$ is fixed, $\ps$ is independent of $n$, and $\ps > 0$ for all $k=1,\ldots,K$. Let $\ns = \sum_{i=1}^n I_{S_i = k} $ denote the number of units in stratum $k$, where $I$ is the indicator function. Additionally,  $n_1 = \sumn A_i  $ and $n_0 = \sumn (1 - A_i )$ denote the numbers of treated and control units, respectively, while $\nst = \sum_{i=1}^n A_i I_{S_i = k}$ and $\nsc = \sum_{i=1}^n (1- A_i ) I_{S_i = k}$ represent the numbers of treated and control units in stratum $k$, respectively. The proportion of treated units in stratum $k$ is denoted as $\pins =\nst/\ns$, and the proportion of stratum sizes is denoted as $\pns = \ns / n$. Regarding the data-generating process and treatment assignment mechanism, we adopt assumptions similar to those in \cite{Bugni2018,Bugni2019}.

	\begin{assumption}
		\label{a1}
		$\left\{Y_i(1),Y_i(0),S_i\right\}_{i=1}^n$ are independent and identically distributed (i.i.d.).
	\end{assumption}

	\begin{assumption}
		\label{a2}
		$[\left\{Y_i(1), Y_i(0)\right\}_{i=1}^n \perp \left\{A_i\right\}_{i=1}^n ] \mid\left\{S_i\right\}_{i=1}^n$.
	\end{assumption}

	\begin{assumption}
		\label{a3}
		$\pins$ converges to $\pi \in (0,1)$ in probability for $k=1,\ldots,K$.
	\end{assumption}

 	\begin{assumption}
		\label{a4}
		$\{n^{-1 / 2} \{D_{\ns}\}_{k=1}^K \mid \left\{S_i\right\}_{i=1}^n\} \stackrel{d}{\rightarrow} N\left(0, \Sigma \right)$ almost surely, where $D_{\ns}=\nst-\pi \ns$, $\Sigma =\operatorname{diag}\left\{\pi_{[k]} q_{[k]}: k=\right.$ $1, \ldots, K\}$ with $0 \leq q_{[k]} \leq \pi(1-\pi)$ for $k=1, \ldots, K$.
	\end{assumption}

Assumption \ref{a2} indicates that the potential outcomes and treatment assignment are conditionally independent given the stratum information. In a covariate-adaptive randomized experiment, if the treatment assignments depend only on the vector of stratum covariates and an exogenous randomization device, then Assumption \ref{a2} is satisfied. For example, the stratified biased-coin design \citep{Efron1971}, stratified permuted block randomization \citep{Zelen1974}, Pocock and Simon's minimization \citep{Pocock1975}, and the randomization methods proposed by \cite{Hu2012} satisfy Assumption \ref{a2}. Assumption \ref{a3} requires that the proportion of treated units in each stratum converges to the same target treated probability $\pi$. Almost all commonly used covariate-adaptive randomization methods satisfy this assumption. Assumption \ref{a4} was proposed by \cite{Bugni2018}. It is stronger than Assumption \ref{a3} because it necessitates the level of imbalance in treatment assignment to be asymptotically independent across strata, as the asymptotic covariance $\Sigma$ forms a diagonal matrix. Simple randomization and stratified permuted block randomization are two typical designs satisfying Assumption \ref{a4}, with their corresponding $q_{[k]}$ values being $\pi(1-\pi)$ and zero, respectively. However, Pocock and Simon's minimization violates this assumption.

To address heavy-tailed outcomes, we refrain from imposing second-order moment conditions on the outcomes, as typically done in studies such as \cite{Bugni2018,Bugni2019,Ye2020Inference,Ye2020,Ma2020Regression,liu2020general}. However, certain issues can become intricate in the absence of specific moment conditions. For instance, defining the average treatment effect $E\{ Y_i(1) - Y_i(0) \}$ may become  challenging if the potential outcomes follow a Cauchy distribution. Even if well-defined, commonly used treatment effect estimators like the difference-in-means \citep{Bugni2018} or the regression-adjusted estimators \citep{Bugni2019,Ye2020Inference,Ye2020,Ma2020Regression,liu2020general} may not exhibit asymptotic normality. To address these challenges, we adopt a methodology outlined in  \cite{athey2021semiparametric} and \cite{ghosh2021efficiency}, focusing on a constant quantile treatment effects assumption. Specifically, let $F_1(y)$ and $F_0(y)$ denote the distribution functions of potential outcomes $Y_i(1)$ and $Y_i(0)$, respectively. Throughout the paper, we assume that $F_1(y) = F_0(y-\tau)$, for $y \in \mathbb R$, where $\tau$ is the target treatment effect. A sufficient but not necessary condition for constant quantile treatment effects is the additivity of the treatment effect, i.e., $Y_i(1) - Y_i(0) = \tau$.

The estimand $\tau$ is equal to the average treatment effect when the first moments of $Y_i(1)$ and $Y_i(0)$ exist. It is also equal to any specific quantile treatment effect  \citep{zhang2020quantile,jiang2023regression}, such as the treatment effect at the median. Given a distribution function $F$, we denote $F^{-1}$ as the quantile function, $f$ as the derivative of $F$, and $f^{\prime}$ as the derivative of $f$. The information function is defined as $I(f)=\int\left(f^{\prime}/f\right)^2(x) f(x) dx$. Even when the constant quantile treatment effects assumption is violated, the estimand $\tau$ retains a causal interpretation \citep{athey2021semiparametric}, representing a weighted average of quantile treatment effects: $\tau = \int_0^1 (F_1^{-1}(u) - F_0^{-1}(u)) dW_f(u),$ where the weight function is 
$
W_f(u) = \int_{0}^u  I(f_0)^{-1} ( - f'_0 / f_0)' (F_0^{-1}(t))  dt$ for $u \in [0,1].
$
Under simple randomization, \cite{athey2021semiparametric} discussed a class of weighted quantile estimators and demonstrated that using the above weight is more efficient compared to using weight at any single quantile.

\section{Treatment effect estimators}
\label{sec:ate-estimators}

\subsection{Transformed difference-in-means estimator}

Under simple randomization, \cite{athey2021semiparametric} proposed an influence function-based M-estimator of $\tau$. Initially, we introduce this estimator and subsequently investigate its asymptotic property under covariate-adaptive randomization, elucidating its limitations.

If we possess knowledge of the potential outcome distribution shape, i.e., we operate within a parametric model: $Y_i(0) \sim F(y-\eta)$, or $Y_i(1) \sim F(y-\eta-\tau)$, where $F$ is known, then the only unknown parameters are the shift parameter $\eta$ and the treatment effect $\tau$. In the context of simple randomization, if the information function corresponding to $f_0$ is bounded away from zero and infinity, $0 < I(f_0) < \infty$, the maximum likelihood estimator of $\tau$, denoted by $\hat \tau$, has the following influence function \citep{athey2021semiparametric}:
	$$
	\psi_{f_0}(A, Y ; \tau)=-\frac{1}{I\left(f_0\right)} \cdot\left\{\frac{A}{\pi} \cdot \frac{f_0^{\prime}}{f_0}(Y-\tau)-\frac{1-A}{1-\pi} \cdot \frac{f_0^{\prime}}{f_0}(Y)\right\},
	$$
where $f_0(\cdot)=f(\cdot-\eta)$ is the probability density function of $Y_i(0)$. Moreover, $\hat{\tau}$ can be approximated near the true $\tau$ through the summation of influence functions, expressed as:
 $$
 \hat{\tau} = {\tau}+\frac{1}{n}\sum_{i=1}^n \psi_{f_0}(A_i, Y_i ; {\tau}) + o_P(n^{-1/2}).
 $$
However, the parameter of interest $\tau$ in the above expansion remains unknown. Consequently, \cite{athey2021semiparametric} suggested using a two-step estimator: first identifying an initial estimator $\tilde{\tau}$, such as the difference between the medians of the outcomes under the treatment and control, followed by an update process
	\begin{equation}
		\begin{aligned}	
			\hat{\tau}_{0} &= \tilde{\tau}+\frac{1}{n}\sum_{i=1}^n \psi_{f_0}(A_i, Y_i ; \tilde{\tau})= \tilde{\tau}-\frac{1}{nI(f_0)}\sum_{i=1}^n \left\{\frac{A_i}{\pi} \cdot \frac{f_1^{\prime}}{f_1}(Y_i-\tilde{\tau}+\tau)-\frac{1-A_i}{1-\pi} \cdot \frac{f_0^{\prime}}{f_0}(Y_i)\right\}, \nonumber
		\end{aligned}
	\end{equation}
	where $f_1(\cdot)=f_0(\cdot-\tau)$ is the probability density function of $Y_i(1)$. The update step is necessary for reducing the asymptotic variance.

Let $Z_i(a) = -I(f_0)^{-1}  (f_a^{\prime}/f_a)(Y_i(a))$, $a=0,1$. Under simple randomization, if $\tilde{\tau}$ is sufficiently close to $\tau$ such that  $\sqrt{n}( \tilde{\tau} - \tau ) = O_P(1)$, then the asymptotic difference between $\hat{\tau}_{0}$ and $\tau$ is $n_1^{-1}\sumn A_i Z_i(1)- n_0^{-1}\sumn (1 - A_i)Z_i(0)$ \citep{athey2021semiparametric}. In this asymptotic context, $\hat{\tau}_{0}$ can be interpreted as the difference-in-means estimator applied to $Z_i(a)$. Intuitively, $Z_i(a)$ serves as a transformation of $Y_i(a)$ and has a thinner tail compared to $Y_i(a)$. This thinning effect stems from the condition $I(f_0)^{-1} =E [Z_i^2(a)] < \infty$, indicating the existence of the second-order moment of $Z_i(a)$. Consequently, the asymptotic normality of the difference-in-means estimator is implied \citep{Bugni2018}.

A practical challenge associated with the estimator $\hat{\tau}_{0}$ arises from the often unknown nature of $f_0$. To address this, one can use the nonparametric method proposed by \cite{bickel1982adaptive} for estimating $f_0$. Specifically, $f_0$ can be effectively estimated using a kernel method, such as
$ \hat{f_0}(y) = (nh)^{-1}\sum_{i=1}^n \phi \{(y-Y_i)/n\},$
where $\phi$ is the probability density function of $N(0,1)$, and $h$ denotes the bandwidth to be determined. Subsequently,  $\partial \hat{f_0} / \partial y$ is utilized as the estimator of $f_0^{\prime}$. To ensure robustness, it becomes necessary to truncate the estimated score function $\hat{f_0}^{\prime} / \hat{f_0}$ appropriately; for a more detailed discussion, refer to \cite{bickel1982adaptive}.

\begin{remark}
We use the Gaussian kernel here simply as an illustrative example. Theoretically, it could be substituted with any bounded kernel function. In practical applications, we recommend the triweight kernel. This preference arises because we aim to minimize the influence of distant samples on local density estimation, particularly in scenarios with heavy-tailed distributions. Therefore, we favor kernels that have bounded support. 
\end{remark}
 
However, using the same samples for estimating $f_0$ and constructing $\hat{\tau}_{0}$ can potentially lead to overfitting \citep{bickel1993efficient, hastie2009elements}. To mitigate this issue, \cite{athey2021semiparametric} suggested using sample splitting and cross-fitting techniques \citep{klaassen1987consistent, chernozhukov2018double}. We slightly modify the sample splitting procedure to take into account the stratification used in the design stage, thereby facilitating the proof of theoretical results. Specifically, we split the samples separately within each stratum for each treatment arm to ensure conditional independence. Let $\mathcal{L}_{[k],a,1}$ denote a random subset of $\{i:S_i = k, \ A_i = a\}$ (the index set for units in stratum $k$ under treatment arm $a$) with cardinality $| \mathcal{L}_{[k],a,1} | =\lfloor n_{[k]a}/2 \rfloor$. Let $\mathcal{L}_1=\bigcup_{k=1}^{K}\bigcup_{a=0}^{1} \mathcal{L}_{[k],a,1}$ denote the combination of $\mathcal{L}_{[k],a,1}$. Additionally, let $\mathcal{L}_2$ represent the complementary set of $\mathcal{L}_1$. 
Then, we modify $\hat{\tau}_{0}$ as follows: 
	\begin{equation}
		\hat{\tau}_{\dim} = \tilde{\tau}+\frac{1}{n}\left\{\sum_{i \in \mathcal{L}_1} \psi_{\hat{f}_{0(2)}}(A_i, Y_i ; \tilde{\tau}) + \sum_{i \in \mathcal{L}_2} \psi_{\hat{f}_{0(1)}}(A_i, Y_i ; \tilde{\tau}) \right\}, \label{E1}
	\end{equation}
	where $\hat{f}_{0(j)}$ denotes the nonparametric estimator  of $f_0$ \citep{bickel1982adaptive} using samples in $\mathcal{L}_j$, $j=1,2$. With a slight abuse of notation, we substitute $I(\hat{f}_{0(j)})$ in the definition of $\psi_{\hat{f}_{0(j)}}$ with $\hat{I}(f_{0})$, representing the estimated Fisher information, which will be formally defined later. For simplicity,  
 we refer to $\hat{\tau}_{\dim}$ as the transformed difference-in-means estimator.

\begin{remark}
 In the context of simple randomization, \cite{athey2021semiparametric} proposed a method involving random splitting of the sample at each treatment arm into two nearly equal parts, denoted as $\mathcal{L}_{a,1}$ and $\mathcal{L}_{a,2}$, followed by the separate combination of samples from different treatment arms into $\mathcal{L}_{1}$ and $\mathcal{L}_{2}$ (ensuring that the cardinality of $\mathcal{L}_1$ under both splitting mechanisms is approximately $n/2$). They defined the treatment effect estimator in the same way as $\hat{\tau}_{\dim}$ and showed that it shares the same asymptotic properties as $\hat{\tau}_{0}$ under the condition of accurate estimation of $f_0$. Under covariate-adaptive randomization, one can directly utilize \cite{athey2021semiparametric}'s estimator without alterations, which exhibits comparable performance to $\hat{\tau}_{\dim}$ in simulation studies. However, deriving its theoretical properties becomes more challenging due to the intricate dependence structure inherent in covariate-adaptive randomization with sample splitting.
 \end{remark}

\begin{remark}
To ensure that the final estimator $\hat{\tau}_{\dim}$ possesses favorable asymptotic properties, the initial estimator $\tilde{\tau}$ must approximate the true value of $\tau$ as closely as possible. Specifically, it is necessary for $\tilde{\tau}$ to be $\sqrt{n}$-consistent, i.e., $\sqrt{n}( \tilde{\tau} - \tau ) = O_P(1)$. A reasonable choice of $\tilde{\tau}$ is the difference in medians of outcomes under the treatment and control groups \citep{athey2021semiparametric}. As demonstrated in the simulation studies and detailed in the supplementary material, the final estimator, $\hat{\tau}_{\dim}$, exhibits lower finite-sample and asymptotic variance compared to the initial difference-in-medians estimator. To take into account stratification under covariate-adaptive randomization, another option of $\tilde{\tau}$ is the difference-in-weighted-medians estimator with weights based on stratum size, or the weighted average of the stratum-specific differences-in-medians estimators within each stratum. In Section~D.3 of the supplementary material, we establish the $\sqrt{n}$-consistency of these estimators. Our simulation studies revealed that these three estimators displayed comparable finite sample performance when used as initial estimators.
\end{remark}

In the following, we examine the asymptotic properties of $\hat{\tau}_{\dim}$ under covariate-adaptive randomization. This analysis is challenging for several reasons: firstly, the dependence among the influence function for different units renders the standard semiparametric and M-estimator theories, which are based on independence assumptions, inapplicable. Although \cite{wang2021model} extended the properties of M-estimators to non-independent settings, their theory required the function $\psi_{f_0}(A_i, Y_i; \tau)$ to be known and applied only to simple, stratified, or biased-coin randomization, excluding minimization. Secondly, existing theories that involve the sample-splitting technique often assume that the two split parts are independent. However, under covariate-adaptive randomization, $\mathcal{L}_1$ and $\mathcal{L}_2$ are correlated due to the dependence among $A_i$ and $Y_i$, necessitating the consideration of this correlation. Finally, the accuracy of the estimator $\hat{f}_0$ needs to be re-evaluated under covariate-adaptive randomization, as the correlation between samples may affect the effectiveness of the original method for estimating $f_0$.

The asymptotic variance of $\hat{\tau}_{\dim}$ depends on the following quantities:
$$
V_{Z}^2 = \frac{\Var\{\tilde Z_i(1)\}}{\pi} + \frac{\Var\{\tilde Z_i(0)\}}{1-\pi}, \quad V_{H}^2 = E\{ \check{Z}_i(1) - \check{Z}_i(0) \}^2, \quad V_{A}^2 = E\Big[ q_{[S_i]} \Big\{ \frac{\check Z_i(1)}{\pi} + \frac{\check Z_i(0)}{1-\pi}  \Big\}^2  \Big],
$$
where $\tilde Z_i(a) = Z_i(a) - E\{Z_i(a) \mid S_i  \}$ and $\check Z_i(a) = E\{Z_i(a) \mid S_i  \} -  E\{Z_i(a)\}$, $a=0,1$. Let $\sigma_{\dim}^2 = V_{Z}^2 + V_{H}^2 + V_{A}^2$. Let $f_{[k]0}(y)$ and $F_{[k]0}(y)$ be the density function and cumulative distribution function of $Y_i(0)$ conditional on $S_i = k$, for $k=1,\ldots,K$.
  
\begin{assumption}
\label{cond:second-moment}
Suppose that the following conditions hold:
\begin{itemize}
\item[(i)] $f_0$ is twice differentiable with $0< I(f_0) < \infty$ and $\int_{-\infty}^{\infty} |f_0^{\prime\prime}(y) | dy < \infty$; 
\item[(ii)] $\max_{k=1,\ldots,K}\Var\{ Z_i(a) \mid S_i = k\} > 0$; 
\item[(iii)] for $j=1,2$, $\sup_{y \in \mathbb R} \{|\hat{f}_{0(j)}^{\prime}/ \hat{f}_{0(j)}|(y) \} = o_P(n^{1/2})$, $\int ( \hat{f}_{0(j)}^{\prime}/ \hat{f}_{0(j)} -f_0^{\prime} / f_0)^2(y) d F_0(y) = o_P(1)$, and for any  $Y_1^{\prime},\ldots,Y_m^{\prime}$ i.i.d. $\sim F_{[k]0}$ and are independent of $\hat{f}_{0(j)}^{\prime}/\hat{f}_{0(j)}$, and ${\delta}_n = O_P(n^{-1 / 2})$, if $m/n$ converges to a positive constant, then as $n \to \infty$, we have $ m^{-1}\sum_{i=1}^m|({\hat{f}_{0(j)}^{\prime}}/{\hat{f}_{0(j)}})^{\prime}|(Y_i^{\prime}+{\delta}_n)= O_P(1)$, $m^{-1}\sum_{i=1}^m({\hat{f}_{0(j)}^{\prime}}/{\hat{f}_{0(j)}})^{\prime}(Y_i^{\prime}+{\delta}_n)= m^{-1}\sum_{i=1}^m({f_0^{\prime}}/{f_0})^{\prime}\left(Y_i^{\prime}\right)+\op$.
\end{itemize}
\end{assumption}

Assumption~\ref{cond:second-moment}(i) ensures the existence of the second-order moment of $Z_i(a)$, which is satisfied by most heavily-tailed distributions, such as the $t$-distribution (including the Cauchy distribution as a special case). The absolute integrability of the second derivative is also a common regularity condition, sufficient to show $E(f_0^{\prime}/f_0)^{\prime}(Y_i(0)) = -I(f_0)$. Assumption~\ref{cond:second-moment}(ii) rules out the case of degeneration of the asymptotic variance. Assumption~\ref{cond:second-moment}(iii) requires  $\hat{f}_{0(j)}^{\prime}/\hat{f}_{0(j)}$ to be an accurate estimator of the true score function. 

\begin{remark}
    Regarding Assumption~\ref{cond:second-moment}(iii), \cite{athey2021semiparametric} does not require $\sup_{y \in \mathbb R} \{|\hat{f}_{0(j)}^{\prime}/ \\ \hat{f}_{0(j)}|(y) \} = o_P(n^{1/2})$ and $ m^{-1}\sum_{i=1}^m|({\hat{f}_{0(j)}^{\prime}}/{\hat{f}_{0(j)}})^{\prime}|(Y_i^{\prime}+{\delta}_n)= O_P(1)$. However, as shown in Theorem \ref{thm:Athey-var} below, the variance estimator proposed by \cite{athey2021semiparametric} is conservative under general covariate-adaptive randomization methods. To derive a consistent variance estimator, we impose these two additional conditions. Notably, these conditions are not necessary if we use the treatment and control group samples to estimate $f_1$ and $f_0$, respectively. However, such an approach may lead to increased computational complexity and less accurate estimates due to smaller sample sizes in one arm (when $\pi \neq 1/2$). In practice, we recommend using the larger sample size arm for density estimation.
\end{remark}

Theorem~\ref{l0} below demonstrates that Assumption~\ref{cond:second-moment}(iii) holds under mild conditions on $f_0$, which are satisfied by many common heavily-tailed distributions (for details, see Section~A in the supplementary material).

\begin{theorem}
    \label{l0}
    Under Assumptions \ref{a1}--\ref{a3} and {\color{red} S.1}--{\color{red} S.2} in the supplementary material, if $f_0$ is twice differentiable with  $0< I(f_0) < \infty$ and $\int_{-\infty}^{\infty} \{f_0^{\prime\prime}(y)\}^2/f_0(y)  dy < \infty$, then $\hat{f}_{0(j)}^{\prime}/\hat{f}_{0(j)}$ ($j = 1, 2$) satisfies Assumption \ref{cond:second-moment}(iii). 
\end{theorem}

    Theorem~\ref{l0} illustrates that under mild conditions, the estimated score function displays notable smoothness and continuity. These properties were initially established by \cite{stone1975adaptive} and \cite{bickel1982adaptive} for a class of kernel density estimation methods applied to independent and identically distributed outcomes. We extend this research by confirming that these favorable properties also apply to outcomes from covariate-adaptive randomization. This includes the consistency of $\hat{f}_{0(j)}$ and $\hat{f}_{0(j)}^{\prime}$, as proven in Section~D.1 in the supplementary material. These findings are crucial for studies that involve density estimation under covariate-adaptive randomization.

	\begin{theorem}
		\label{t1}
		Under Assumptions \ref{a1}--\ref{a2} and \ref{a4}--\ref{cond:second-moment}, if the initial estimator $\tilde \tau$ is $\sqrt{n}$-consistent, i.e., $\sqrt{n} (\tilde \tau - \tau) = O_P(1)$, then
  $\sqrt{n}(\hat{\tau}_{\dim}-\tau) \stackrel{d}{\rightarrow} N (0, \sigma_{\dim}^2)$.
		
	\end{theorem}

\begin{remark}
\label{remark:root-n-consistent}

Under the constant quantile treatment effects assumption, estimators for any specific quantile can serve as initial estimators. Since the asymptotic variance of the quantile estimator is roughly inversely proportional to the square of the density function value at that quantile \citep{zhang2020quantile}, we typically choose the median for efficiency considerations. 
\end{remark}

Theorem~\ref{t1} implies that the transformed difference-in-means estimator $\hat \tau_{\dim}$ is consistent and asymptotically normal under covariate-adaptive randomization. However, Assumption~\ref{a4} excludes Pocock and Simon's minimization method, although this assumption is satisfied for many other covariate-adaptive randomization methods, such as simple randomization, stratified permuted block randomization, and stratified biased coin design. Consequently, it remains unclear whether $\hat \tau_{\dim}$ is asymptotically normal or not under minimization. Deriving the asymptotic distribution of $\hat \tau_{\dim}$ under minimization is challenging due to the complex dependence structure of treatment assignment across strata, even asymptotically.

As indicated by Theorem~\ref{t1}, the asymptotic variance depends on the randomization methods used during the design stage via $q_{[k]}$. Consequently, we cannot provide a universal inference method that is applicable to all commonly used covariate-adaptive randomization methods. For simple randomization, $q_{[k]}= \pi (1 - \pi)$ and the asymptotic variance simplifies to $\{\pi(1-\pi)I(f_0)\}^{-1}$; see also \cite{athey2021semiparametric}. For general covariate-adaptive randomization methods, we can verify that $$\frac{1}{\pi(1-\pi)I(f_0)} - \sigma^2_{\dim} = E\Big[ \{ \pi(1-\pi)- q_{[S_i]}\} \Big\{ \frac{\check Z_i(1)}{\pi} + \frac{\check Z_i(0)}{1-\pi}  \Big\}^2 \Big] \geq 0.$$
Thus, the asymptotic variance of $\hat \tau_{\dim}$ is no greater than that under simple randomization. Moreover, as the design becomes more balanced (i.e., $q_{[k]}$ decreases), the asymptotic variance $\sigma^2_{\dim}$ decreases, indicating that a more balanced design leads to a more efficient estimator for the average treatment effect. The asymptotic variance $\sigma^2_{\dim}$ achieves its minimum value if the design is strongly balanced, i.e., $q_{[k]} = 0$ for $k=1,...,K$. In this case, the asymptotic variance has a much simpler expression, as shown in Corollary~\ref{cor1} below.

	\begin{corollary}
		\label{cor1}
		Under Assumptions \ref{a1}--\ref{a2} and \ref{a4}--\ref{cond:second-moment} with $q_{[k]} = 0$ for $k=1,...,K$,  if the initial estimator $\tilde \tau$ is $\sqrt{n}$-consistent, i.e., $\sqrt{n} (\tilde \tau - \tau) = O_P(1)$, then
		$
		\sqrt{n}(\hat{\tau}_{\dim}-\tau) \stackrel{d}{\rightarrow} N(0, \sigma^2_{\str}),
		$ where $\sigma_{str}^2 = V_{Z}^2 + V_{H}^2$.
	\end{corollary}

\cite{athey2021semiparametric} suggested using $\{\pi(1-\pi)\hat{I}(f_0)\}^{-1}$ to estimate the asymptotic variance, where $\hat{I}(f_0)$ is a consistent estimator of ${I}(f_0)$. For example,
$$
\hat{I}(f_0)=\frac{1}{n_0}\bigg[\sum_{i \in \mathcal{L}_1, A_i=0}\bigg\{\frac{\hat{f}_{0(2)}^{\prime}}{\hat{f}_{0(2)}}(Y_i)\bigg\}^2 + \sum_{i \in \mathcal{L}_2, A_i=0}\bigg\{\frac{\hat{f}_{0(1)}^{\prime}}{\hat{f}_{0(1)}}(Y_i)\bigg\}^2\bigg].
$$

\begin{remark}
$\hat{I}(f_0)$ appears in both the point and variance estimators. In theory, we only need $\hat{I}(f_0)$ to be consistent with the true Fisher information $I(f_0)$, i.e., $\hat{I}(f_0) = I(f_0) + o_P(1)$. There are various ways to estimate $I(f_0)$. For example, \cite{athey2021semiparametric} proposed an alternative estimator $\hat I^*({f}_{0}) = \{ I(\hat{f}_{0(1)}) +  I(\hat{f}_{0(2)}) \} /2 $, where
$$I(\hat{f}_{0(j)}) = -|\mathcal{L}_{0,j}|^{-1}\sum_{i \in \mathcal{L}_j, A_i=0} \frac{\hat{f}_{0(j)}(Y_i)\hat{f}_{0(j)}^{\prime \prime}(Y_i) - \hat{f}_{0(j)}^{\prime}(Y_i)^2}{\hat{f}_{0(j)}(Y_i)^2}, \quad j=1,2.$$
The point estimator is modified accordingly as $\hat{\tau} = (\hat{\tau}_{(1)}+\hat{\tau}_{(2)})/2$, where 
$$
\hat {\tau}_{(j)} =  \tilde{\tau}_{(3-j)}-\frac{1}{|\mathcal{L}_j|I(\hat{f}_{0(3-j)})}\sum_{i \in \mathcal{L}_j} \left\{\frac{A_i}{\pi} \cdot \frac{f_{0(3-j)}^{\prime}}{f_{0(3-j)}}(Y_i-\tilde{\tau}_{(3-j)})-\frac{1-A_i}{1-\pi} \cdot \frac{f_{0(3-j)}^{\prime}}{f_{0(3-j)}}(Y_i)\right\}, \  j=1,2.
$$ 
Here, $\tilde{\tau}_{(j)}$ is an initial estimator (they used the difference-in-medians of outcomes under the treatment and control group within $\mathcal{L}_j$). Our simulation results indicate that these two estimating strategies perform similarly.
\end{remark}

Theorem~\ref{thm:Athey-var} below indicates that $\{\pi(1-\pi)\hat{I}(f_0)\}^{-1}$ is generally conservative under covariate-adaptive randomization when $q_{[k]}< \pi (1 - \pi)$.

\begin{theorem}
\label{thm:Athey-var}
Under the Assumptions of Theorem~\ref{t1}, $\{\pi(1-\pi)\hat{I}(f_0)\}^{-1}$ converges in probability to $\{\pi(1-\pi)I(f_0)\}^{-1}$. Moreover, we have $\{\pi(1-\pi)I(f_0)\}^{-1} \geq \sigma^2_{\dim}$, with equality holding if $q_{[k]} = \pi(1-\pi)$ for $k=1,...,K$.
\end{theorem}

Theorem~\ref{thm:Athey-var} shows that $\{\pi(1-\pi)\hat{I}(f_0)\}^{-1}$ is consistent under simple randomization but overestimates the asymptotic variance under a more balanced design with $0\leq q_{[k]}< \pi (1 - \pi)$, resulting in a loss of power. To make valid inferences on the average treatment effect using normal approximation, we need a non-parametric variance estimator. Specifically, we can replace the population mean and variance of the three terms in $\sigma_{\dim}^2$ with their corresponding sample mean and variance to create a plug-in estimator.

Let $\hat{Z}_i$ denote the ``observed" transformed outcome as follows: for $i \in \mathcal{L}_j, \ A_i = 0$, we define $\hat{Z}_i = -\hat{I}(f_0)^{-1} \cdot (\hat{f}_{0(3-j)}^{\prime}/\hat{f}_{0(3-j)})(Y_i),$ $j=1,2$, while for $i \in \mathcal{L}_j,\ A_i = 1$, we define $\hat{Z}_i = -\hat{I}(f_0)^{-1} \cdot (\hat{f}_{0(3-j)}^{\prime}/\hat{f}_{0(3-j)})(Y_i-\tilde{\tau}),$ $j=1,2$. Let $\Bar{\hat{Z}}_{[k]a} = \sum_{i:S_i = k, A_i=a} \hat{Z}_i / n_{[k]a}$ denote the stratum-specific sample mean of $\hat{Z}_i$ under treatment arm $a$ in stratum $k$ and $\Bar{\hat{Z}}_{a} = \sum_{i:A_i=a} \hat{Z}_i / n_a$ denote the sample mean of  $\hat{Z}_i$ under treatment arm $a$, where $a=0,1$, $k=1,\ldots,K$.
We define $\hat \sigma^2_\dim = \hat{V}_{Z}^2  + \hat{V}_{H}^2 + \hat{V}_{A}^2 $ with
$$
  \hat{V}_{Z}^2 = \frac{1}{\pi}\sum_{k=1}^K p_{n[k]}\Big\{\frac{1}{n_{[k]1}}\sumis A_i(\hat{Z}_i-\Bar{\hat{Z}}_{[k]1})^2 \Big\} + \frac{1}{1-\pi}\sum_{k=1}^K p_{n[k]}\Big\{ \frac{1}{n_{[k]0}}\sumis (1-A_i)(\hat{Z}_i-\Bar{\hat{Z}}_{[k]0})^2 \Big\}, 
$$
$$
  \hat{V}_{H}^2 = \sum_{k=1}^K p_{n[k]} \left\{(\Bar{\hat{Z}}_{[k]1}-\Bar{\hat{Z}}_{1})-(\Bar{\hat{Z}}_{[k]0}-\Bar{\hat{Z}}_{0}) \right\}^2, 
$$
$$
  \hat{V}_{A}^2 =\sum_{k=1}^K p_{n[k]}q_{[k]}\Big\{\frac{(\Bar{\hat{Z}}_{[k]1}-\Bar{\hat{Z}}_{1})}{\pi}+\frac{(\Bar{\hat{Z}}_{[k]0}-\Bar{\hat{Z}}_{0})}{1-\pi} \Big\}^2.
$$

\begin{theorem}
\label{t2}
Under Assumptions \ref{a1}--\ref{a2} and \ref{a4}--\ref{cond:second-moment}, if the initial estimator $\tilde \tau$ is $\sqrt{n}$-consistent, i.e., $\sqrt{n} (\tilde \tau - \tau) = O_P(1)$, then $\hat \sigma^2_\dim$ converges in probability to $\sigma^2_\dim$.
\end{theorem}

Theorem~\ref{t2} indicates that the plug-in estimator $\hat \sigma^2_\dim$ is a consistent estimator of the asymptotic variance. Thus, the Wald-type confidence interval,
$$
[\hat \tau_\dim - q_{\alpha/2} \hat \sigma_\dim / \sqrt{n}, \hat \tau_\dim + q_{\alpha/2} \hat \sigma_\dim / \sqrt{n}],
$$
has an asymptotic coverage rate of $ 1-\alpha$, $0<\alpha<1$,
where $q_{\alpha/2}$ is the upper $\alpha/2$ quantile of a standard normal distribution.

\subsection{Stratified transformed difference-in-means estimator}

The asymptotic variance of $\hat \tau_\dim$ depends on the randomization methods, and its plug-in variance estimator is not universally applicable. To address this issue and enhance efficiency, we propose a stratified transformed difference-in-means estimator along with a universally applicable inference method in this section.

Previous studies have demonstrated that the stratified difference-in-means estimator, which aggregates the difference-in-means estimator calculated for each stratum weighted by the stratum size proportion $\pns$, often performs better than the simple difference-in-means estimator \citep{Bugni2019, Ma2020Regression, Ye2020}. This observation motivated us to explore a stratified transformed difference-in-means estimator. As mentioned earlier, by treating $Z_i(a) = -I(f_0)^{-1} \cdot (f_a^{\prime}/f_a)(Y_i(a))$ as the transformed outcomes, $\hat{\tau}_{\dim}$ is asymptotically equivalent to the difference-in-means estimator applied to  $Z_i(a)$. Therefore, the stratified transformed difference-in-means estimator is defined as:
$$
\hat{\tau}_{\str} = \tilde{\tau}+\sum_{k=1}^{K}\pns \cdot \frac{1}{\ns}\bigg(\sum_{i \in \mathcal{L}_1,S_i=k} \psi_{\hat{f}_{0(2)}}^{(k)}(A_i, Y_i ; \tilde{\tau}) + \sum_{i \in \mathcal{L}_2,S_i=k} \psi_{\hat{f}_{0(1)}}^{(k)}(A_i, Y_i ; \tilde{\tau}) \bigg),
$$
where $$
	\psi_{f_0}^{(k)}(A, Y ; \tau)=-\dfrac{1}{I\left(f_0\right)} \cdot\left\{ \dfrac{A}{\pi_{n[k]}} \cdot \dfrac{f_0^{\prime}}{f_0}(Y-\tau)-\dfrac{1-A}{1-\pi_{n[k]}} \cdot \dfrac{f_0^{\prime}}{f_0}(Y)\right \}.
	$$
The asymptotic variance of $\hat{\tau}_{\str}$ can be estimated by $\hat \sigma^2_{\str} = \hat{V}_{Z}^2  + \hat{V}_{H}^2$. Theorem~\ref{t4} below presented the asymptotic properties of $\hat{\tau}_{\str}$ and  $\hat \sigma^2_{\str}$.

	\begin{theorem}
		\label{t4}
		Under Assumptions \ref{a1}--\ref{a3} and \ref{cond:second-moment}, if the initial estimator $\tilde \tau$ is $\sqrt{n}$-consistent, i.e., $\sqrt{n} (\tilde \tau - \tau) = O_P(1)$, then $\hat{\tau}_{\str}$ is consistent and asymptotically normal,
		$
		\sqrt{n}(\hat{\tau}_{\str}-\tau) \stackrel{d}{\rightarrow} N(0, \sigma^2_{\str} ).
		$ 
  Moreover, $\hat \sigma^2_{\str}$ converges in probability to $\sigma^2_{\str}$. 
	\end{theorem}

It is worth noting that the asymptotic normality of $\hat{\tau}_{\str}$ and the consistency of the variance estimator $\hat \sigma^2_{\str}$ require less stringent design requirements as the conclusion holds under the weaker Assumption~\ref{a3} instead of the stronger Assumption~\ref{a4}. Since Assumption~\ref{a3} is quite general, Theorem~\ref{t4} is applicable to almost all covariate-adaptive randomization methods, including Pocock and Simon's minimization. Moreover, the asymptotic variance is identical to that of $\hat{\tau}_{\dim}$ under strongly balanced designs, which is independent of the randomization methods used in the design stage. Therefore,  $\hat{\tau}_{\str}$ and $\hat \sigma^2_{\str}$ are universally applicable  \citep{Ye2020}. For designs that are not well-balanced, such as simple randomization, Theorem \ref{t4} indicates that the asymptotic variance of $\hat{\tau}_{\str}$ is smaller than or equal to that of $\hat{\tau}_{\dim}$, thus enhancing the asymptotic efficiency in comparison to $\hat{\tau}_{\dim}$. These conclusions align with those in \cite{Bugni2019} for outcomes with light-tails.

By Theorem~\ref{t4}, we are able to provide a valid inference for the treatment effect $\tau$. One approach is constructing a Wald-type $1 - \alpha$ ($0< \alpha < 1$) confidence interval for $\tau$: 
$$
[ \hat{\tau}_{\str} - q_{\alpha/2} \hat\sigma_{\str}/\sqrt{n},   \hat{\tau}_{\str} + q_{\alpha/2} \hat\sigma_{\str}/\sqrt{n}]. 
$$
The asymptotic coverage rate of the above confidence interval is $1 - \alpha$.

\section{Simulation study}

In this section, we carry out simulation studies to evaluate the finite sample performance of the proposed estimators $\hat{\tau}_{\dim}$ and $\hat{\tau}_{\str}$ under three covariate-adaptive randomization methods: simple randomization (SR), stratified randomization (STR), and Pocock and Simon's minimization (MIN).

We set the sample size as $n =500, 1000$ and consider three models for generating the potential outcomes, similar to those elucidated in \cite{ghosh2021efficiency}.

	Model 1 (Linear stratum effect):  
$ Y_i(0) = (3/4) x_{i1}+x_{i2}+\varepsilon_i,$
	where $x_{i1} \sim \operatorname{Unif}(-1,1)$, $x_{i2}$ takes values in $\{-1, -1/3, 1/3, 1\}$ with equal probabilities, and it is independent of $x_{i1}$.

	Model 2 (Nonlinear stratum effect): 
	$Y_i(0) =(1/2) \left\{ \exp \left(z_{i}\right)+\exp \left(z_{i} / 2\right)\right\} +\varepsilon_i,$
	with $z_i=(3/4) x_{i1}+x_{i2}$, where $x_{i1} \sim \operatorname{Unif}(-1,1)$ and $x_{i2}$ takes values in $\{-1, -1/3, 1/3, 1\}$ with equal probabilities, independent with $x_{i1}$.

	Model 3 (Nonlinear stratum effect with one covariate transformed to exponential): 
	$Y_i(0)=(1/2) \left(z_{i}+\sqrt{z_{i}}\right)+\varepsilon_i,$
	with $z_i=x_{i1}+x_{i1}x_{i2}$, where $x_{i1}=e^{u_i}, \ u_i \sim \operatorname{Unif}(-1,1)$ and $x_{i2}$ takes values in $\{-1, -1/3, 1/3, 1\}$ with equal probabilities, independent with $u_{i1}$.

The error terms $\epsilon_i$ are independently drawn from the standard normal distribution, the standard double exponential (Laplace) distribution, and the standard Cauchy distribution. The generation of $\epsilon_i$ is independent of $x_{i1}$ and $x_{i2}$. The treatment probability $\pi$ is set to 0.5, and the potential outcome $Y_i(1)$ is generated as $Y_i(1) = Y_i(0) + \tau$ with $\tau$ set at 0 and 1. The results for unequal treatment probability ($\pi=1/3$) are relegated to the supplementary material. Under each model and each tail distribution, $x_{i2}$ is used as the stratification covariate in STR. For minimization, both $x_{i1}$ and $x_{i2}$ are dichotomized as stratification covariates. The biased-coin probability and the weight used in minimization are set to 0.85 and $(0.5, 0.5)$, respectively. Under each model and each tail distribution, the data generation process is repeated 1000 times to evaluate the bias, standard deviation (SD), root mean squared error (RMSE) of the point estimators, and empirical coverage probability (CP) and mean confidence interval length (Length) of $95\%$ confidence intervals. 

We apply the kernel density estimation methods proposed by \cite{bickel1982adaptive} to estimate the influence function for $\hat{\tau}_{\dim}$ and $\hat{\tau}_{\str}$, and use the triweight kernel instead of the original Gaussian kernel to achieve better finite-sample performance. The bandwidth as a critical hyperparameter is selected through a one-step adaptive update, and other necessary truncation parameters are chosen empirically. To highlight the enhanced efficiency of our proposed methods in scenarios with heavy-tailed distributions, we compare their performance with that of the standard difference-in-means estimator and the weighted average difference-in-means estimator, denoted as $\hat{\tau}_{\textnormal{naive-dim}}$ and $\hat{\tau}_{\textnormal{str-dim}}$, respectively. The difference-in-medians estimator and the difference-in-weighted-medians estimator, denoted as $\hat{\tau}_{\textnormal{md}}$ and $\hat{\tau}_{\textnormal{wt-md}}$, are used as the initial estimators for $\hat{\tau}_{\dim}$ and $\hat{\tau}_{\str}$, respectively, and their performance is included in our results. More detailed discussions on the initial estimators are provided in the supplementary material. For variance estimation, we use the method from \cite{Ma2020Regression} for $\hat{\tau}_{\textnormal{naive-dim}}$ and $\hat{\tau}_{\textnormal{str-dim}}$, and use a plug-in method for $\hat{\tau}_{\textnormal{md}}$ and $\hat{\tau}_{\textnormal{wt-md}}$. Since $\hat{\tau}_{\textnormal{naive-dim}}$, $\hat{\tau}_{\textnormal{md}}$ and $\hat{\tau}_{\dim}$ lack asymptotic properties without Assumption \ref{a4}, their inference-related results are not available under minimization.

Tables \ref{tab:n=1000model1}-\ref{tab:n=1000model3} present the results for $n = 1000$. The results for $n = 500$ exhibit a comparable pattern but with increased variability and these results are relegated to the supplementary material. Our observations are as follows:

\begin{table}[htbp]
  \centering
  \caption{Simulation results under Model 1 for $n=1000$ and $\pi=1/2$}
  \begin{threeparttable}
    \resizebox{\textwidth}{!}{
    \begin{tabular}{rrrrrrrrrrrrrr}
    \toprule
          &       & \multicolumn{4}{c}{SR}        & \multicolumn{4}{c}{STR}       & \multicolumn{4}{c}{MIN} \\
\cmidrule{3-14}          &       & \multicolumn{1}{r}{Bias} & \multicolumn{1}{r}{SD} & \multicolumn{1}{r}{SE} & \multicolumn{1}{r}{CP} & \multicolumn{1}{r}{Bias} & \multicolumn{1}{r}{SD} & \multicolumn{1}{r}{SE} & \multicolumn{1}{r}{CP} & \multicolumn{1}{r}{Bias} & \multicolumn{1}{r}{SD} & \multicolumn{1}{r}{SE} & \multicolumn{1}{r}{CP} \\
    \midrule
    \multicolumn{1}{l}{$\tau$ = 0} &       &       &       &       &       &       &       &       &       &       &       &       &  \\
    \multicolumn{1}{l}{Normal tail} & $\hat{\tau}_{\textnormal{naive-dim}}$ & 0.004  & 0.083  & 0.083  & 0.940  & -0.003  & 0.071  & 0.069  & 0.945  & -0.001  & 0.080  & /     & / \\
          & $\hat{\tau}_{\textnormal{str-dim}}$ & 0.003  & 0.071  & 0.069  & 0.941  & -0.003  & 0.071  & 0.069  & 0.945  & -0.002  & 0.069  & 0.069  & 0.950  \\
          & $\hat{\tau}_{\textnormal{md}}$ & 0.006  & 0.106  & 0.118  & 0.963  & -0.004  & 0.094  & 0.106  & 0.962  & -0.002  & 0.104  & /     & / \\
          & $\hat{\tau}_{\textnormal{wt-md}}$ & 0.005  & 0.095  & 0.106  & 0.961  & -0.004  & 0.094  & 0.106  & 0.962  & -0.003  & 0.094  & 0.106  & 0.960  \\
          & $\hat{\tau}_{\dim}$ & 0.005  & 0.090  & 0.091  & 0.941  & -0.005  & 0.078  & 0.081  & 0.963  & -0.002  & 0.089  & /     & / \\
          & $\hat{\tau}_{\str}$ & 0.005  & 0.079  & 0.082  & 0.945  & -0.005  & 0.078  & 0.081  & 0.965  & -0.003  & 0.080  & 0.082  & 0.949  \\
          &       &       &       &       &       &       &       &       &       &       &       &       &  \\
    \multicolumn{1}{l}{Laplace tail} & $\hat{\tau}_{\textnormal{naive-dim}}$ & 0.006  & 0.105  & 0.105  & 0.947  & -0.004  & 0.098  & 0.093  & 0.928  & 0.000  & 0.103  & /     & / \\
          & $\hat{\tau}_{\textnormal{str-dim}}$ & 0.005  & 0.094  & 0.093  & 0.945  & -0.004  & 0.098  & 0.093  & 0.927  & -0.001  & 0.094  & 0.093  & 0.953  \\
          & $\hat{\tau}_{\textnormal{md}}$ & 0.008  & 0.117  & 0.119  & 0.953  & -0.003  & 0.107  & 0.109  & 0.958  & 0.000  & 0.118  & /     & / \\
          & $\hat{\tau}_{\textnormal{wt-md}}$ & 0.007  & 0.107  & 0.109  & 0.948  & -0.003  & 0.107  & 0.109  & 0.959  & 0.000  & 0.108  & 0.109  & 0.945  \\
          & $\hat{\tau}_{\dim}$ & 0.006  & 0.107  & 0.103  & 0.932  & -0.006  & 0.092  & 0.092  & 0.943  & 0.002  & 0.101  & /     & / \\
          & $\hat{\tau}_{\str}$ & 0.006  & 0.095  & 0.092  & 0.942  & -0.006  & 0.092  & 0.092  & 0.943  & 0.001  & 0.091  & 0.092  & 0.958  \\
          &       &       &       &       &       &       &       &       &       &       &       &       &  \\
    \multicolumn{1}{l}{Cauchy tail} & $\hat{\tau}_{\textnormal{naive-dim}}$ & 0.367  & 29.541  & 8.246  & 0.977  & 0.109  & 30.100  & 8.310  & 0.979  & 0.927  & 28.668  & /     & / \\
          & $\hat{\tau}_{\textnormal{str-dim}}$ & 0.324  & 28.599  & 8.233  & 0.978  & 0.110  & 30.105  & 8.310  & 0.979  & 0.970  & 28.581  & 8.316  & 0.977  \\
          & $\hat{\tau}_{\textnormal{md}}$ & 0.000  & 0.147  & 0.147  & 0.955  & 0.004  & 0.135  & 0.138  & 0.946  & 0.006  & 0.137  & /     & / \\
          & $\hat{\tau}_{\textnormal{wt-md}}$ & -0.001  & 0.138  & 0.138  & 0.942  & 0.005  & 0.135  & 0.138  & 0.946  & 0.004  & 0.130  & 0.138  & 0.965  \\
          & $\hat{\tau}_{\dim}$ & 0.002  & 0.128  & 0.128  & 0.940  & 0.003  & 0.118  & 0.118  & 0.946  & 0.006  & 0.124  & /     & / \\
          & $\hat{\tau}_{\str}$ & 0.001  & 0.118  & 0.118  & 0.949  & 0.003  & 0.118  & 0.118  & 0.945  & 0.005  & 0.116  & 0.118  & 0.956  \\
    \multicolumn{1}{l}{$\tau$ = 1} &       &       &       &       &       &       &       &       &       &       &       &       &  \\
    \multicolumn{1}{l}{Normal tail} & $\hat{\tau}_{\textnormal{naive-dim}}$ & 0.003  & 0.083  & 0.083  & 0.951  & 0.007  & 0.068  & 0.069  & 0.955  & -0.004  & 0.077  & /     & / \\
          & $\hat{\tau}_{\textnormal{str-dim}}$ & 0.003  & 0.068  & 0.069  & 0.951  & 0.007  & 0.068  & 0.069  & 0.955  & -0.004  & 0.067  & 0.069  & 0.956  \\
          & $\hat{\tau}_{\textnormal{md}}$ & 0.001  & 0.107  & 0.118  & 0.967  & 0.006  & 0.092  & 0.106  & 0.977  & 0.000  & 0.100  & /     & / \\
          & $\hat{\tau}_{\textnormal{wt-md}}$ & 0.001  & 0.095  & 0.106  & 0.968  & 0.006  & 0.092  & 0.106  & 0.977  & 0.001  & 0.091  & 0.106  & 0.974  \\
          & $\hat{\tau}_{\dim}$ & 0.003  & 0.091  & 0.091  & 0.948  & 0.007  & 0.075  & 0.081  & 0.961  & -0.002  & 0.084  & /     & / \\
          & $\hat{\tau}_{\str}$ & 0.003  & 0.078  & 0.082  & 0.968  & 0.007  & 0.075  & 0.081  & 0.960  & -0.002  & 0.075  & 0.082  & 0.967  \\
          &       &       &       &       &       &       &       &       &       &       &       &       &  \\
    \multicolumn{1}{l}{Laplace tail} & $\hat{\tau}_{\textnormal{naive-dim}}$ & 0.001  & 0.106  & 0.105  & 0.948  & 0.002  & 0.096  & 0.093  & 0.944  & 0.000  & 0.100  & /     & / \\
          & $\hat{\tau}_{\textnormal{str-dim}}$ & 0.001  & 0.094  & 0.093  & 0.951  & 0.002  & 0.096  & 0.093  & 0.945  & 0.000  & 0.090  & 0.093  & 0.953  \\
          & $\hat{\tau}_{\textnormal{md}}$ & 0.003  & 0.114  & 0.120  & 0.952  & 0.004  & 0.108  & 0.109  & 0.945  & -0.001  & 0.109  & /     & / \\
          & $\hat{\tau}_{\textnormal{wt-md}}$ & 0.003  & 0.100  & 0.109  & 0.972  & 0.004  & 0.108  & 0.109  & 0.946  & -0.002  & 0.100  & 0.109  & 0.965  \\
          & $\hat{\tau}_{\dim}$ & 0.002  & 0.105  & 0.103  & 0.937  & 0.004  & 0.092  & 0.092  & 0.945  & -0.001  & 0.097  & /     & / \\
          & $\hat{\tau}_{\str}$ & 0.002  & 0.091  & 0.092  & 0.944  & 0.004  & 0.092  & 0.092  & 0.945  & -0.001  & 0.088  & 0.092  & 0.965  \\
          &       &       &       &       &       &       &       &       &       &       &       &       &  \\
    \multicolumn{1}{l}{Cauchy tail} & $\hat{\tau}_{\textnormal{naive-dim}}$ & -34.037  & 1037.167  & 44.967  & 0.978  & -33.248  & 1099.566  & 45.306  & 0.987  & 35.288  & 1088.502  & /     & / \\
          & $\hat{\tau}_{\textnormal{str-dim}}$ & -35.219  & 1076.533  & 44.902  & 0.980  & -33.307  & 1101.616  & 45.306  & 0.987  & 34.421  & 1060.345  & 44.601  & 0.978  \\
          & $\hat{\tau}_{\textnormal{md}}$ & 0.001  & 0.148  & 0.147  & 0.948  & -0.004  & 0.144  & 0.138  & 0.938  & -0.007  & 0.144  & /     & / \\
          & $\hat{\tau}_{\textnormal{wt-md}}$ & 0.002  & 0.138  & 0.138  & 0.948  & -0.004  & 0.143  & 0.138  & 0.938  & -0.007  & 0.138  & 0.138  & 0.948  \\
          & $\hat{\tau}_{\dim}$ & -0.002  & 0.129  & 0.128  & 0.946  & -0.007  & 0.122  & 0.118  & 0.933  & -0.006  & 0.126  & /     & / \\
          & $\hat{\tau}_{\str}$ & -0.002  & 0.119  & 0.118  & 0.944  & -0.007  & 0.122  & 0.118  & 0.931  & -0.006  & 0.120  & 0.118  & 0.942  \\
    \bottomrule
    \end{tabular}%
    }
    \end{threeparttable}
  \label{tab:n=1000model1}
  \begin{tablenotes}
				\footnotesize
				\item[] Note: $\hat{\tau}_{\textnormal{naive-dim}}$, standard difference-in-means estimator; $\hat{\tau}_{\textnormal{str-dim}}$, stratified difference-in-means esti-\\mator; $\hat{\tau}_{\textnormal{md}}$, difference-in-medians estimator; $\hat{\tau}_{\textnormal{wt-md}}$, difference-in-weighted-medians estimator; $\hat{\tau}_{\dim}$,\\ transformed difference-in-means estimator; $\hat{\tau}_{\str}$, stratified transformed difference-in-means estimator;\\ SR, simple randomization; STR, stratified randomization; MIN, Pocock and Simon's minimization;\\ SD, standard deviation; SE, standard error; CP, coverage probability.
			\end{tablenotes}
\end{table}%

\begin{table}[htbp]
  \centering
  \caption{Simulation results under Model 2 for $n=1000$ and $\pi=1/2$}
  \begin{threeparttable}
    \resizebox{\textwidth}{!}{
    \begin{tabular}{rlrrrrrrrrrrrr}
    \toprule
          &       & \multicolumn{4}{c}{SR}        & \multicolumn{4}{c}{STR}       & \multicolumn{4}{c}{MIN} \\
\cmidrule{3-14}          &       & \multicolumn{1}{r}{Bias} & \multicolumn{1}{r}{SD} & \multicolumn{1}{r}{SE} & \multicolumn{1}{r}{CP} & \multicolumn{1}{r}{Bias} & \multicolumn{1}{r}{SD} & \multicolumn{1}{r}{SE} & \multicolumn{1}{r}{CP} & \multicolumn{1}{r}{Bias} & \multicolumn{1}{r}{SD} & \multicolumn{1}{r}{SE} & \multicolumn{1}{r}{CP} \\
    \midrule
    \multicolumn{1}{l}{$\tau$ = 0} &       &       &       &       &       &       &       &       &       &       &       &       &  \\
    \multicolumn{1}{l}{Normal tail} & $\hat{\tau}_{\textnormal{naive-dim}}$ & -0.002  & 0.082  & 0.083  & 0.957  & -0.002  & 0.071  & 0.070  & 0.945  & 0.001  & 0.080  & /     & / \\
          & $\hat{\tau}_{\textnormal{str-dim}}$ & 0.000  & 0.069  & 0.070  & 0.959  & -0.003  & 0.071  & 0.070  & 0.943  & 0.000  & 0.069  & 0.070  & 0.947  \\
          & $\hat{\tau}_{\textnormal{md}}$ & 0.002  & 0.102  & 0.109  & 0.959  & -0.002  & 0.090  & 0.100  & 0.970  & -0.002  & 0.097  & /     & / \\
          & $\hat{\tau}_{\textnormal{wt-md}}$ & 0.004  & 0.094  & 0.100  & 0.965  & -0.002  & 0.090  & 0.100  & 0.969  & -0.003  & 0.088  & 0.100  & 0.961  \\
          & $\hat{\tau}_{\dim}$ & 0.000  & 0.088  & 0.085  & 0.939  & -0.003  & 0.077  & 0.078  & 0.954  & 0.000  & 0.084  & /     & / \\
          & $\hat{\tau}_{\str}$ & 0.001  & 0.079  & 0.079  & 0.943  & -0.003  & 0.077  & 0.078  & 0.954  & 0.000  & 0.078  & 0.079  & 0.945  \\
          &       &       &       &       &       &       &       &       &       &       &       &       &  \\
    \multicolumn{1}{l}{Laplace tail} & $\hat{\tau}_{\textnormal{naive-dim}}$ & 0.000  & 0.103  & 0.104  & 0.949  & -0.002  & 0.097  & 0.094  & 0.937  & 0.002  & 0.098  & /     & / \\
          & $\hat{\tau}_{\textnormal{str-dim}}$ & 0.003  & 0.094  & 0.094  & 0.947  & -0.002  & 0.097  & 0.094  & 0.938  & 0.002  & 0.092  & 0.094  & 0.957  \\
          & $\hat{\tau}_{\textnormal{md}}$ & 0.000  & 0.104  & 0.106  & 0.955  & 0.000  & 0.101  & 0.098  & 0.940  & 0.003  & 0.102  & /     & / \\
          & $\hat{\tau}_{\textnormal{wt-md}}$ & 0.003  & 0.097  & 0.098  & 0.951  & 0.000  & 0.101  & 0.098  & 0.939  & 0.002  & 0.096  & 0.098  & 0.944  \\
          & $\hat{\tau}_{\dim}$ & 0.002  & 0.098  & 0.096  & 0.948  & 0.001  & 0.092  & 0.090  & 0.949  & 0.003  & 0.093  & /     & / \\
          & $\hat{\tau}_{\str}$ & 0.004  & 0.090  & 0.090  & 0.951  & 0.001  & 0.092  & 0.090  & 0.949  & 0.002  & 0.088  & 0.090  & 0.951  \\
          &       &       &       &       &       &       &       &       &       &       &       &       &  \\
    \multicolumn{1}{l}{Cauchy tail} & $\hat{\tau}_{\textnormal{naive-dim}}$ & 0.434  & 29.979  & 7.450  & 0.969  & -0.036  & 29.854  & 7.453  & 0.978  & -1.446  & 29.947  & /     & / \\
          & $\hat{\tau}_{\textnormal{str-dim}}$ & 0.379  & 28.675  & 7.438  & 0.968  & -0.037  & 29.826  & 7.453  & 0.978  & -1.453  & 30.317  & 7.479  & 0.979  \\
          & $\hat{\tau}_{\textnormal{md}}$ & 0.001  & 0.137  & 0.138  & 0.946  & 0.000  & 0.127  & 0.132  & 0.955  & -0.001  & 0.137  & /     & / \\
          & $\hat{\tau}_{\textnormal{wt-md}}$ & 0.003  & 0.129  & 0.131  & 0.946  & 0.000  & 0.127  & 0.132  & 0.956  & -0.001  & 0.130  & 0.132  & 0.954  \\
          & $\hat{\tau}_{\dim}$ & 0.003  & 0.125  & 0.125  & 0.946  & 0.000  & 0.116  & 0.118  & 0.954  & 0.000  & 0.125  & /     & / \\
          & $\hat{\tau}_{\str}$ & 0.005  & 0.116  & 0.117  & 0.946  & 0.000  & 0.115  & 0.118  & 0.955  & 0.000  & 0.118  & 0.118  & 0.949  \\
    \multicolumn{1}{l}{$\tau$ = 1} &       &       &       &       &       &       &       &       &       &       &       &       &  \\
    \multicolumn{1}{l}{Normal tail} & $\hat{\tau}_{\textnormal{naive-dim}}$ & -0.002  & 0.084  & 0.083  & 0.948  & -0.001  & 0.072  & 0.070  & 0.948  & 0.001  & 0.078  & /     & / \\
          & $\hat{\tau}_{\textnormal{str-dim}}$ & -0.002  & 0.071  & 0.070  & 0.942  & -0.001  & 0.072  & 0.070  & 0.947  & 0.002  & 0.067  & 0.070  & 0.961  \\
          & $\hat{\tau}_{\textnormal{md}}$ & -0.003  & 0.101  & 0.109  & 0.964  & -0.004  & 0.092  & 0.101  & 0.971  & 0.001  & 0.097  & /     & / \\
          & $\hat{\tau}_{\textnormal{wt-md}}$ & -0.003  & 0.091  & 0.101  & 0.963  & -0.004  & 0.092  & 0.101  & 0.969  & 0.003  & 0.089  & 0.101  & 0.973  \\
          & $\hat{\tau}_{\dim}$ & -0.002  & 0.089  & 0.085  & 0.935  & -0.001  & 0.079  & 0.079  & 0.930  & 0.001  & 0.084  & /     & / \\
          & $\hat{\tau}_{\str}$ & -0.002  & 0.079  & 0.079  & 0.945  & -0.001  & 0.079  & 0.079  & 0.929  & 0.002  & 0.078  & 0.079  & 0.957  \\
          &       &       &       &       &       &       &       &       &       &       &       &       &  \\
    \multicolumn{1}{l}{Laplace tail} & $\hat{\tau}_{\textnormal{naive-dim}}$ & -0.001  & 0.108  & 0.104  & 0.949  & -0.003  & 0.096  & 0.094  & 0.952  & -0.003  & 0.102  & /     & / \\
          & $\hat{\tau}_{\textnormal{str-dim}}$ & 0.000  & 0.097  & 0.094  & 0.940  & -0.003  & 0.096  & 0.094  & 0.950  & -0.002  & 0.095  & 0.094  & 0.939  \\
          & $\hat{\tau}_{\textnormal{md}}$ & 0.000  & 0.109  & 0.106  & 0.936  & -0.006  & 0.099  & 0.098  & 0.945  & 0.000  & 0.103  & /     & / \\
          & $\hat{\tau}_{\textnormal{wt-md}}$ & 0.000  & 0.098  & 0.098  & 0.938  & -0.006  & 0.099  & 0.098  & 0.945  & 0.000  & 0.099  & 0.098  & 0.936  \\
          & $\hat{\tau}_{\dim}$ & 0.000  & 0.099  & 0.096  & 0.939  & -0.003  & 0.089  & 0.089  & 0.955  & -0.001  & 0.094  & /     & / \\
          & $\hat{\tau}_{\str}$ & 0.000  & 0.090  & 0.089  & 0.947  & -0.003  & 0.089  & 0.089  & 0.955  & 0.000  & 0.090  & 0.089  & 0.953  \\
          &       &       &       &       &       &       &       &       &       &       &       &       &  \\
    \multicolumn{1}{l}{Cauchy tail} & $\hat{\tau}_{\textnormal{naive-dim}}$ & -159.708  & 4998.819  & 170.268  & 0.976  & -157.734  & 5119.119  & 171.269  & 0.984  & -160.063  & 5108.970  & /     & / \\
          & $\hat{\tau}_{\textnormal{str-dim}}$ & -160.720  & 5028.971  & 170.017  & 0.979  & -157.407  & 5108.884  & 171.269  & 0.984  & -172.505  & 5501.904  & 177.367  & 0.982  \\
          & $\hat{\tau}_{\textnormal{md}}$ & 0.003  & 0.139  & 0.138  & 0.943  & 0.001  & 0.127  & 0.131  & 0.953  & -0.003  & 0.136  & /     & / \\
          & $\hat{\tau}_{\textnormal{wt-md}}$ & 0.004  & 0.133  & 0.131  & 0.935  & 0.002  & 0.127  & 0.131  & 0.953  & -0.001  & 0.131  & 0.131  & 0.952  \\
          & $\hat{\tau}_{\dim}$ & 0.001  & 0.127  & 0.125  & 0.936  & 0.000  & 0.115  & 0.118  & 0.947  & -0.004  & 0.126  & /     & / \\
          & $\hat{\tau}_{\str}$ & 0.001  & 0.118  & 0.117  & 0.939  & 0.000  & 0.115  & 0.118  & 0.947  & -0.003  & 0.122  & 0.118  & 0.945  \\
    \bottomrule
    \end{tabular}%
    }
    \end{threeparttable}
  \label{tab:n=1000model2}
  \begin{tablenotes}
				\footnotesize
				\item[] Note: $\hat{\tau}_{\textnormal{naive-dim}}$, standard difference-in-means estimator; $\hat{\tau}_{\textnormal{str-dim}}$, stratified difference-in-means esti-\\mator; $\hat{\tau}_{\textnormal{md}}$, difference-in-medians estimator; $\hat{\tau}_{\textnormal{wt-md}}$, difference-in-weighted-medians estimator; $\hat{\tau}_{\dim}$,\\ transformed difference-in-means estimator; $\hat{\tau}_{\str}$, stratified transformed difference-in-means estimator;\\ SR, simple randomization; STR, stratified randomization; MIN, Pocock and Simon's minimization;\\ SD, standard deviation; SE, standard error; CP, coverage probability.
			\end{tablenotes}
\end{table}%

\begin{table}[htbp]
  \centering
  \caption{Simulation results under Model 3 for $n=1000$ and $\pi=1/2$}
  \begin{threeparttable}
    \resizebox{\textwidth}{!}{
    \begin{tabular}{rlrrrrrrrrrrrr}
    \toprule
          &       & \multicolumn{4}{c}{SR}        & \multicolumn{4}{c}{STR}       & \multicolumn{4}{c}{MIN} \\
\cmidrule{3-14}          &       & \multicolumn{1}{r}{Bias} & \multicolumn{1}{r}{SD} & \multicolumn{1}{r}{SE} & \multicolumn{1}{r}{CP} & \multicolumn{1}{r}{Bias} & \multicolumn{1}{r}{SD} & \multicolumn{1}{r}{SE} & \multicolumn{1}{r}{CP} & \multicolumn{1}{r}{Bias} & \multicolumn{1}{r}{SD} & \multicolumn{1}{r}{SE} & \multicolumn{1}{r}{CP} \\
    \midrule
    \multicolumn{1}{l}{$\tau$ = 0} &       &       &       &       &       &       &       &       &       &       &       &       &  \\
    \multicolumn{1}{l}{Normal tail} & $\hat{\tau}_{\textnormal{naive-dim}}$ & -0.002  & 0.087  & 0.085  & 0.947  & 0.001  & 0.074  & 0.072  & 0.950  & -0.001  & 0.077  & /     & / \\
          & $\hat{\tau}_{\textnormal{str-dim}}$ & -0.002  & 0.075  & 0.072  & 0.935  & 0.001  & 0.074  & 0.072  & 0.950  & -0.001  & 0.068  & 0.072  & 0.961  \\
          & $\hat{\tau}_{\textnormal{md}}$ & 0.000  & 0.108  & 0.115  & 0.968  & 0.003  & 0.095  & 0.105  & 0.968  & 0.002  & 0.096  & /     & / \\
          & $\hat{\tau}_{\textnormal{wt-md}}$ & -0.002  & 0.097  & 0.105  & 0.970  & 0.003  & 0.095  & 0.105  & 0.968  & 0.001  & 0.090  & 0.105  & 0.980  \\
          & $\hat{\tau}_{\dim}$ & 0.000  & 0.092  & 0.090  & 0.944  & 0.002  & 0.079  & 0.081  & 0.958  & -0.001  & 0.083  & /     & / \\
          & $\hat{\tau}_{\str}$ & 0.000  & 0.081  & 0.082  & 0.949  & 0.002  & 0.079  & 0.081  & 0.958  & -0.002  & 0.077  & 0.082  & 0.951  \\
          &       &       &       &       &       &       &       &       &       &       &       &       &  \\
    \multicolumn{1}{l}{Laplace tail} & $\hat{\tau}_{\textnormal{naive-dim}}$ & 0.001  & 0.106  & 0.106  & 0.941  & 0.001  & 0.099  & 0.096  & 0.937  & -0.002  & 0.104  & /     & / \\
          & $\hat{\tau}_{\textnormal{str-dim}}$ & 0.001  & 0.094  & 0.096  & 0.949  & 0.001  & 0.099  & 0.096  & 0.936  & -0.003  & 0.096  & 0.096  & 0.953  \\
          & $\hat{\tau}_{\textnormal{md}}$ & 0.001  & 0.116  & 0.113  & 0.945  & 0.002  & 0.103  & 0.104  & 0.956  & -0.001  & 0.109  & /     & / \\
          & $\hat{\tau}_{\textnormal{wt-md}}$ & 0.000  & 0.105  & 0.104  & 0.945  & 0.002  & 0.103  & 0.104  & 0.956  & -0.001  & 0.102  & 0.104  & 0.953  \\
          & $\hat{\tau}_{\dim}$ & 0.001  & 0.105  & 0.101  & 0.937  & 0.003  & 0.092  & 0.093  & 0.949  & 0.000  & 0.098  & /     & / \\
          & $\hat{\tau}_{\str}$ & 0.000  & 0.093  & 0.093  & 0.948  & 0.003  & 0.092  & 0.093  & 0.951  & -0.001  & 0.091  & 0.093  & 0.957  \\
          &       &       &       &       &       &       &       &       &       &       &       &       &  \\
    \multicolumn{1}{l}{Cauchy tail} & $\hat{\tau}_{\textnormal{naive-dim}}$ & -2.556  & 67.555  & 9.639  & 0.979  & -1.792  & 65.256  & 9.539  & 0.980  & -0.426  & 64.678  & /     & / \\
          & $\hat{\tau}_{\textnormal{str-dim}}$ & -2.735  & 70.760  & 9.623  & 0.982  & -1.789  & 65.172  & 9.539  & 0.980  & -0.342  & 64.733  & 9.522  & 0.973  \\
          & $\hat{\tau}_{\textnormal{md}}$ & 0.006  & 0.145  & 0.145  & 0.950  & -0.004  & 0.133  & 0.138  & 0.950  & 0.000  & 0.136  & /     & / \\
          & $\hat{\tau}_{\textnormal{wt-md}}$ & 0.005  & 0.137  & 0.137  & 0.948  & -0.004  & 0.133  & 0.138  & 0.950  & -0.001  & 0.131  & 0.138  & 0.958  \\
          & $\hat{\tau}_{\dim}$ & 0.003  & 0.134  & 0.128  & 0.937  & -0.002  & 0.116  & 0.120  & 0.962  & -0.002  & 0.122  & /     & / \\
          & $\hat{\tau}_{\str}$ & 0.003  & 0.126  & 0.120  & 0.929  & -0.002  & 0.116  & 0.120  & 0.961  & -0.002  & 0.117  & 0.120  & 0.950  \\
    \multicolumn{1}{l}{$\tau$ = 1} &       &       &       &       &       &       &       &       &       &       &       &       &  \\
    \multicolumn{1}{l}{Normal tail} & $\hat{\tau}_{\textnormal{naive-dim}}$ & -0.001  & 0.086  & 0.085  & 0.935  & 0.004  & 0.075  & 0.072  & 0.946  & 0.004  & 0.079  & /     & / \\
          & $\hat{\tau}_{\textnormal{str-dim}}$ & 0.000  & 0.073  & 0.072  & 0.934  & 0.004  & 0.075  & 0.072  & 0.949  & 0.003  & 0.069  & 0.072  & 0.963  \\
          & $\hat{\tau}_{\textnormal{md}}$ & -0.002  & 0.103  & 0.115  & 0.963  & 0.004  & 0.099  & 0.105  & 0.953  & 0.008  & 0.101  & /     & / \\
          & $\hat{\tau}_{\textnormal{wt-md}}$ & -0.001  & 0.094  & 0.105  & 0.961  & 0.003  & 0.099  & 0.105  & 0.951  & 0.007  & 0.095  & 0.105  & 0.964  \\
          & $\hat{\tau}_{\dim}$ & 0.000  & 0.090  & 0.089  & 0.946  & 0.004  & 0.082  & 0.081  & 0.945  & 0.005  & 0.084  & /     & / \\
          & $\hat{\tau}_{\str}$ & 0.000  & 0.079  & 0.081  & 0.948  & 0.004  & 0.082  & 0.081  & 0.945  & 0.005  & 0.077  & 0.082  & 0.957  \\
          &       &       &       &       &       &       &       &       &       &       &       &       &  \\
    \multicolumn{1}{l}{Laplace tail} & $\hat{\tau}_{\textnormal{naive-dim}}$ & 0.000  & 0.108  & 0.106  & 0.947  & -0.004  & 0.097  & 0.096  & 0.945  & 0.001  & 0.102  & /     & / \\
          & $\hat{\tau}_{\textnormal{str-dim}}$ & 0.000  & 0.096  & 0.096  & 0.956  & -0.004  & 0.097  & 0.096  & 0.945  & 0.000  & 0.094  & 0.096  & 0.952  \\
          & $\hat{\tau}_{\textnormal{md}}$ & 0.002  & 0.114  & 0.113  & 0.933  & -0.002  & 0.103  & 0.104  & 0.947  & 0.002  & 0.111  & /     & / \\
          & $\hat{\tau}_{\textnormal{wt-md}}$ & 0.002  & 0.101  & 0.104  & 0.950  & -0.002  & 0.103  & 0.104  & 0.948  & 0.001  & 0.103  & 0.104  & 0.949  \\
          & $\hat{\tau}_{\dim}$ & 0.002  & 0.103  & 0.101  & 0.943  & -0.003  & 0.089  & 0.093  & 0.948  & 0.002  & 0.100  & /     & / \\
          & $\hat{\tau}_{\str}$ & 0.002  & 0.092  & 0.093  & 0.954  & -0.003  & 0.089  & 0.093  & 0.949  & 0.001  & 0.092  & 0.093  & 0.950  \\
          &       &       &       &       &       &       &       &       &       &       &       &       &  \\
    \multicolumn{1}{l}{Cauchy tail} & $\hat{\tau}_{\textnormal{naive-dim}}$ & 0.163  & 46.122  & 9.876  & 0.985  & 1.694  & 46.874  & 9.881  & 0.980  & 2.482  & 47.248  & /     & / \\
          & $\hat{\tau}_{\textnormal{str-dim}}$ & 0.144  & 45.653  & 9.861  & 0.985  & 1.697  & 46.929  & 9.881  & 0.980  & 2.489  & 45.433  & 9.826  & 0.975  \\
          & $\hat{\tau}_{\textnormal{md}}$ & -0.004  & 0.138  & 0.144  & 0.956  & 0.009  & 0.134  & 0.136  & 0.956  & 0.000  & 0.143  & /     & / \\
          & $\hat{\tau}_{\textnormal{wt-md}}$ & -0.004  & 0.134  & 0.137  & 0.951  & 0.009  & 0.134  & 0.136  & 0.956  & 0.000  & 0.137  & 0.137  & 0.949  \\
          & $\hat{\tau}_{\dim}$ & -0.003  & 0.126  & 0.127  & 0.954  & 0.006  & 0.119  & 0.120  & 0.953  & 0.002  & 0.132  & /     & / \\
          & $\hat{\tau}_{\str}$ & -0.003  & 0.119  & 0.119  & 0.954  & 0.006  & 0.119  & 0.120  & 0.952  & 0.001  & 0.126  & 0.119  & 0.933  \\
    \bottomrule
    \end{tabular}%
    }
    \end{threeparttable}
  \label{tab:n=1000model3}
  \begin{tablenotes}
				\footnotesize
				\item[] Note: $\hat{\tau}_{\textnormal{naive-dim}}$, standard difference-in-means estimator; $\hat{\tau}_{\textnormal{str-dim}}$, stratified difference-in-means esti-\\mator; $\hat{\tau}_{\textnormal{md}}$, difference-in-medians estimator; $\hat{\tau}_{\textnormal{wt-md}}$, difference-in-weighted-medians estimator; $\hat{\tau}_{\dim}$,\\ transformed difference-in-means estimator; $\hat{\tau}_{\str}$, stratified transformed difference-in-means estimator;\\ SR, simple randomization; STR, stratified randomization; MIN, Pocock and Simon's minimization;\\ SD, standard deviation; SE, standard error; CP, coverage probability.
			\end{tablenotes}
\end{table}%

(1) The biases of both $\hat{\tau}_{\dim}$ and $\hat{\tau}_{\str}$ are negligible in comparison to the standard deviations, even when dealing with the heaviest-tailed distribution (the Cauchy distribution).

(2) With the same model and randomization approach, the standard deviations of $\hat{\tau}_{\dim}$ and $\hat{\tau}_{\str}$ increase as the potential outcomes exhibit heavier tails, yet they remain well behaved. Furthermore, under the STR method, $\hat{\tau}_{\dim}$ shows lower standard deviations compared to SR. This observation aligns with the conclusion in Theorem \ref{t1} that a more balanced design results in a more precise transformed difference-in-means estimator. In contrast, the standard deviations of $\hat{\tau}_{\str}$ are independent of the randomization methods, as predicted by Theorem \ref{t4}. Additionally, these standard deviations are nearly identical to those of $\hat{\tau}_{\dim}$ under STR, highlighting the effectiveness of stratification both in the design and analysis stages.

(3) Compared to the initial estimators $\hat{\tau}_{\textnormal{md}}$ and $\hat{\tau}_{\textnormal{wt-md}}$, both $\hat{\tau}_{\dim}$ and $\hat{\tau}_{\str}$ exhibit smaller standard deviations across all cases, underscoring the need for a one-step update. While $\hat{\tau}_{\dim}$ and $\hat{\tau}_{\str}$ also demonstrate reduced variability compared to $\hat{\tau}_{\textnormal{naive-dim}}$ and $\hat{\tau}_{\textnormal{str-dim}}$, this advantage does not extend to scenarios with normal tails. In fact, when the distribution of the outcomes is perfectly normal, both $\hat{\tau}_{\textnormal{naive-dim}}$ and $\hat{\tau}_{\dim}$ asymptotically achieve the efficiency bounds under simple randomization \citep{athey2021semiparametric}. However, in finite samples,  $\hat{\tau}_{\dim}$ exhibits a larger standard deviation, attributable to the imprecision in the density estimation and the initial estimator. 

(4) The variance estimators are reliable, ensuring that the empirical coverage probabilities are close to $95\%$. The average lengths of confidence intervals for $\hat{\tau}_{\str}$ are either smaller than (approximately $10\%$ smaller under SR) or equal to those of $\hat{\tau}_{\dim}$. Therefore, $\hat{\tau}_{\str}$ improves, or at least does not hurt the inferential efficiency in comparison to $\hat{\tau}_{\dim}$.

Based on these simulation results, we recommend $\hat{\tau}_{\str}$ and its associated variance estimators for inferring the treatment effect under covariate-adaptive randomization with heavy-tailed potential outcomes.

\section{Real data example}
The AIDS Clinical Trial Group conducted a randomized controlled double-blind study to evaluate the effects of dual or triple combinations of HIV-1 reverse transcriptase inhibitors \citep{henry1998randomized}. The study involved 1313 HIV-infected patients from 42 adult AIDS Clinical Trials Group sites and 7 National Hemophilia Foundation centers, spanning the period from June 1993 to June 1996. The CD4 counts were used to determine patient enrollment and were of interest as outcomes post-intervention. All patients were followed up to 40 weeks, with CD4 counts measured every 8 weeks. Patients were randomly assigned to one of four treatment arms, which included: 600mg zidovudine alternating monthly with 400mg didanosine (arm 1), 600mg zidovudine plus 2.25mg of zalcitabine (arm 2), 600mg zidovudine plus 400mg of didanosine (arm 3), and 600mg zidovudine plus 400mg of didanosine plus 400mg of nevirapine (arm 4). The data is available at {\href{https://content.sph.harvard.edu/fitzmaur/ala/cd4.txt}{\it https://content.sph.harvard.edu/fitzmaur/ala/cd4.txt}}.

The original CD4 count data exhibits severe non-normality and heavy-tailed behavior; refer to Figures~\ref{fig:qq} and \ref{fig:hist} (for enhanced visual clarity, the data is scaled by dividing it by $10^5$). We highlight the notable differences in quantiles between the normal and CD4 count distributions at week 0 and week 8. The Shapiro--Wilk test yields a $p$-value below $10^{-6}$, further affirming the non-normality. Both the kurtosis values at week 0 and week 8 significantly exceed 3, indicating the CD4 count distributions possess heavy tails. Although prior studies commonly applied a logarithmic transformation to mitigate skewness in CD4 counts, our approach leverages methods tailored for heavy-tailed data, facilitating the direct analysis of the CD4 counts.

\begin{figure}
    \centering
    \resizebox{\linewidth}{!}{
    \includegraphics{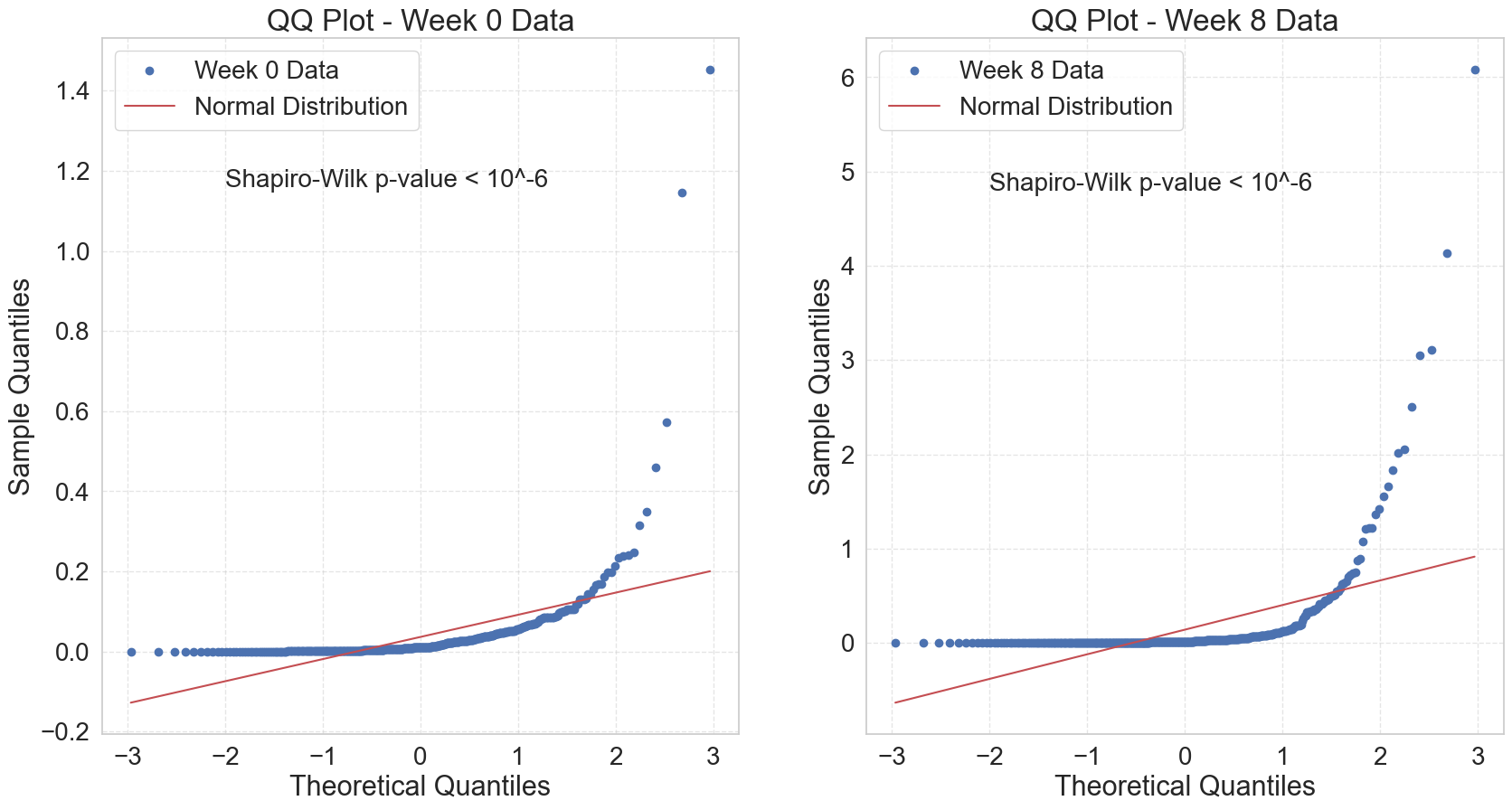}
    }
    \caption{Q-Q plot for CD4 counts ($\times 10^5$) in the first 8 weeks}
    \label{fig:qq}
\end{figure}

\begin{figure}
    \centering
    \resizebox{\linewidth}{!}{
    \includegraphics{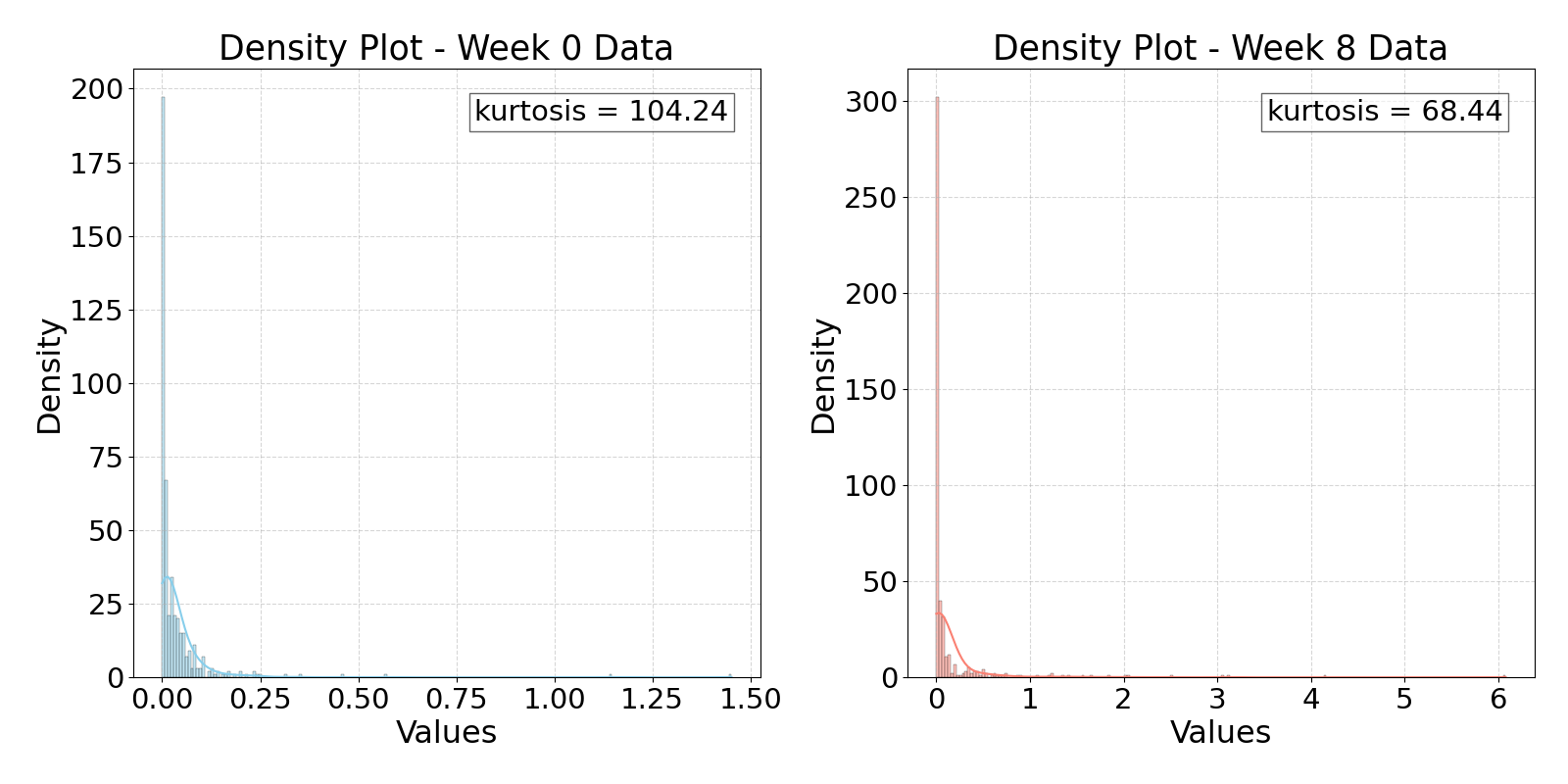}
    }
    \caption{Density plot for CD4 counts ($\times 10^5$) in the first 8 weeks}
    \label{fig:hist}
\end{figure}

We focus on assessing the impact of triple therapy, consisting of 600 mg of zidovudine, 400 mg of didanosine, and 400 mg of nevirapine, which we consider the treatment arm. The other three dual therapies serve as control arms, and we compare their effects pairwisely. We use the difference between the CD4 counts at the 8th, 16th, and 24th weeks and the initial counts as the outcomes. For each time period, we focus on patients without missing data, stratifying them into four strata based on gender and dichotomized age, and consider the estimated propensity score $n_1/n$ as the true treatment assignment probability. Table~\ref{tab:real} shows the point estimators, $95\%$ confidence intervals, and the lengths of the confidence intervals. Our results indicate that the triple therapy leads to significantly more CD4 counts compared to any other therapy, regardless of the period considered (8th, 16th, or 24th week), suggesting that the combination of triple reverse transcriptase inhibitors better delays disease progression in AIDS patients. This result is consistent with the findings of \cite{henry1998randomized}. As a comparison, we can see that the classical difference-in-means estimator and stratified difference-in-means estimator exhibit extremely large variances in such heavy-tailed situations, thereby being unable to provide effective inference.

\begin{table}[htbp]
  \centering
  \caption{Results for CD4 counts: triple therapy vs. three dual therapies}
  \begin{threeparttable}
    \resizebox{\textwidth}{!}{
    \begin{tabular}{lccccccccccccc}
    \toprule
          & \multicolumn{3}{c}{Week 8}    & \multicolumn{3}{c}{Week 16}   & \multicolumn{3}{c}{Week 24} \\
\cmidrule{2-13}          & \multicolumn{1}{l}{Estimator} & \multicolumn{1}{c}{$95\%$ CI} & \multicolumn{1}{l}{Length} & \multicolumn{1}{l}{Estimator} & \multicolumn{1}{c}{$95\%$ CI} & \multicolumn{1}{l}{Length} & \multicolumn{1}{l}{Estimator} & \multicolumn{1}{c}{$95\%$ CI} & \multicolumn{1}{l}{Length} \\
    \midrule
    Arm 4 vs. 1 &       &       &       &       &       &       &       &       &       &       &       &  \\
    $\hat{\tau}_{\textnormal{naive-dim}}$ & 14265 & (6546, 21984) & 15438 & 33485 & (11208, 55762) & 44553 & 30231 & (4451, 56011) & 51560 \\
    $\hat{\tau}_{\textnormal{str-dim}}$ & 14533 & (6839, 22228) & 15389 & 33866 & (11630, 56103) & 44473 & 30906 & (5176, 56636) & 51460 \\
    $\hat{\tau}_{\dim}$ & 673   & (217, 1129) & 912   & 989   & (434, 1543) & 1109  & 632   & (55, 1209) & 1154 \\
    $\hat{\tau}_{\str}$ & 682   & (227, 1138) & 911   & 995   & (441, 1548) & 1107  & 649   & (73, 1225) & 1152 \\
    Arm 4 vs. 2 &       &       &       &       &       &       &       &       &       &       &       &  \\
    $\hat{\tau}_{\textnormal{naive-dim}}$ & 13379 & (5687, 21070) & 15383 & 33345 & (11213, 55478) & 44264 & 29315 & (4245, 54384) & 50139 \\
    $\hat{\tau}_{\textnormal{str-dim}}$ & 13356 & (5684, 21028) & 15344 & 33601 & (11511, 55690) & 44178 & 29361 & (4356, 54366) & 50010 \\
    $\hat{\tau}_{\dim}$ & 1137  & (370, 1904) & 1534  & 1286  & (546, 2026) & 1480  & 1122  & (138, 2106) & 1968 \\
    $\hat{\tau}_{\str}$ & 1140  & (374, 1905) & 1532  & 1299  & (560, 2038) & 1478  & 1125  & (143, 2108) & 1965 \\
    Arm 4 vs. 3 &       &       &       &       &       &       &       &       &       &       &       &  \\
    $\hat{\tau}_{\textnormal{naive-dim}}$ & 6150  & (-2797, 15097) & 17894 & 27367 & (5273, 49460) & 44187 & 23900 & (-1785, 49585) & 51370 \\
    $\hat{\tau}_{\textnormal{str-dim}}$ & 6102  & (-2822, 15025) & 17847 & 26892 & (4875, 48908) & 44033 & 23829 & (-1743, 49401) & 51144 \\
    $\hat{\tau}_{\dim}$ & 3164  & (390, 5938) & 5548  & 2825  & (971, 4679) & 3708  & 2001  & (268, 3734) & 3466 \\
    $\hat{\tau}_{\str}$ & 3375  & (648, 6102) & 5454  & 2811  & (958, 4664) & 3706  & 2071  & (348, 3794) & 3445 \\
    \bottomrule
    \end{tabular}%
    }
    \end{threeparttable}
    \begin{tablenotes}
				\footnotesize
				\item[] Note: $\hat{\tau}_{\textnormal{naive-dim}}$, standard difference-in-means estimator; $\hat{\tau}_{\textnormal{str-dim}}$, stratified difference-in-means esti-\\mator; $\hat{\tau}_{\dim}$, transformed difference-in-means estimator; $\hat{\tau}_{\str}$, stratified transformed difference-in-\\means estimator; $95\%$ CI, 95\% 
    confidence interval; length,  length of the 
    confidence interval.
			\end{tablenotes}
  \label{tab:real}%
\end{table}%

To validate the reliability of our method, we carry out additional simulations using synthetic data. Our focus remains on comparing the triple therapy to other dual therapies, specifically observing changes in CD4 counts at the 8th week. Taking the comparison between the triple therapy and the alternating therapy (600mg zidovudine alternating monthly with 400mg didanosine) as an example, after excluding missing data, we analyze the first eight weeks of data from 223 participants in the control arm and 235 in the treatment arm. We impute the counterfactual outcomes using the estimated treatment effect, $\hat{\tau}_{\str}$, from the original data as the assumed true value. We then resampled the outcomes of 458 patients with replacement, along with their gender and age as stratification variables. We simulate the reallocation of treatment using simple randomization, stratified randomization, and Pocock and Simon's minimization, repeating the process 2000 times to evaluate the bias, standard deviation (SD), standard error (SE), and empirical coverage probability (CP) of the $95\%$ confidence intervals. The results are displayed in Table~\ref{tab:real-simu} alongside two other dual therapies. The findings align with those in the simulation section, showing that both estimators, $\hat{\tau}_{\dim}$ and $\hat{\tau}_{\str}$, exhibit negligible biases, which are significantly smaller than their respective SDs. The true variance term associated with treatment allocation $V_{A}^2$ defined in Section \ref{sec:ate-estimators}, is nearly zero, which explains why the SDs of $\hat{\tau}_{\dim}$ and $\hat{\tau}_{\str}$ are similar under each design. Moreover, the empirical coverage probabilities are approximately $95\%$, confirming the validity of our interval estimators.

\begin{table}[htbp]
  \centering
  \caption{Simulation results for synthetic CD4 count data}
  \begin{threeparttable}
    \resizebox{\textwidth}{!}{
    \begin{tabular}{lrrrrrrrrrrrr}
    \toprule
          & \multicolumn{4}{c}{SR}        & \multicolumn{4}{c}{STR}       & \multicolumn{4}{c}{MIN} \\
\cmidrule{2-13}          & \multicolumn{1}{l}{Bias} & \multicolumn{1}{l}{SD} & \multicolumn{1}{l}{SE} & \multicolumn{1}{l}{CP} & \multicolumn{1}{l}{Bias} & \multicolumn{1}{l}{SD} & \multicolumn{1}{l}{SE} & \multicolumn{1}{l}{CP} & \multicolumn{1}{l}{Bias} & \multicolumn{1}{l}{SD} & \multicolumn{1}{l}{SE} & \multicolumn{1}{l}{CP} \\
    \midrule
    Arm 4 vs. 1 &       &       &       &       &       &       &       &       &       &       &       &  \\
    $\hat{\tau}_{\textnormal{naive-dim}}$ & -103  & 4018  & 3931  & 0.958  & -84   & 4037  & 3911  & 0.959  & 160   & 4059  & /     & / \\
    $\hat{\tau}_{\textnormal{str-dim}}$ & -99   & 4017  & 3908  & 0.957  & -85   & 4037  & 3911  & 0.958  & 149   & 4057  & 3909  & 0.948 \\
    $\hat{\tau}_{\dim}$ & 5     & 293   & 278   & 0.965  & -4    & 302   & 278   & 0.955  & 7     & 289   & /     & / \\
    $\hat{\tau}_{\str}$ & 4     & 294   & 277   & 0.962  & -1    & 303   & 278   & 0.956  & 7     & 289   & 278   & 0.962 \\
    Arm 4 vs. 2 &       &       &       &       &       &       &       &       &       &       &       &  \\
    $\hat{\tau}_{\textnormal{naive-dim}}$ & 39    & 4077  & 3984  & 0.955  & -95   & 4135  & 3963  & 0.954  & 71    & 4055  & /     & / \\
    $\hat{\tau}_{\textnormal{str-dim}}$ & 35    & 4080  & 3959  & 0.952  & -73   & 4137  & 3963  & 0.954  & 77    & 4057  & 3959  & 0.953 \\
    $\hat{\tau}_{\dim}$ & 2     & 330   & 329   & 0.961  & -14   & 319   & 326   & 0.967  & 11    & 320   & /     & / \\
    $\hat{\tau}_{\str}$ & 4     & 334   & 326   & 0.958  & -9    & 318   & 326   & 0.967  & 11    & 319   & 322   & 0.969 \\
    Arm 4 vs. 3 &       &       &       &       &       &       &       &       &       &       &       &  \\
    $\hat{\tau}_{\textnormal{naive-dim}}$ & 198   & 4594  & 4432  & 0.945  & -110  & 4480  & 4411  & 0.946  & -157  & 4533  & /     & / \\
    $\hat{\tau}_{\textnormal{str-dim}}$ & 200   & 4604  & 4408  & 0.951  & -71   & 4478  & 4411  & 0.945  & -153  & 4531  & 4412  & 0.947 \\
    $\hat{\tau}_{\dim}$ & 23    & 580   & 589   & 0.961  & 0     & 584   & 592   & 0.964  & 36    & 581   & /     & / \\
    $\hat{\tau}_{\str}$ & 26    & 582   & 584   & 0.963  & 5     & 584   & 592   & 0.966  & 38    & 581   & 593   & 0.963 \\
    \bottomrule
    \end{tabular}%
    }
    \end{threeparttable}
  \label{tab:real-simu}
  \begin{tablenotes}
				\footnotesize
				\item[] Note: $\hat{\tau}_{\textnormal{naive-dim}}$, standard difference-in-means estimator; $\hat{\tau}_{\textnormal{str-dim}}$, stratified difference-in-means esti-\\mator; $\hat{\tau}_{\dim}$, transformed difference-in-means estimator; $\hat{\tau}_{\str}$, stratified transformed difference-in-\\means estimator; SR, simple randomization; STR, stratified randomization; MIN, Pocock and\\ Simon's minimization; SD, standard deviation; SE, standard error; CP, coverage probability.
			\end{tablenotes}
\end{table}%

\section{Discussion}

Covariate-adaptive randomization methods are extensively used in clinical trials and economic studies to assess the impact of treatments or interventions. In many of these experiments, outcomes often exhibit heavy-tailed behavior, posing challenges in accurate estimation and efficient inference of treatment effects. In this paper, we proposed methods to effectively infer treatment effects under covariate-adaptive randomization with heavy-tailed outcomes. Specifically, we examined the asymptotic properties of the semi-parametric influence function-based M-estimation approach introduced by \cite{athey2021semiparametric} under covariate-adaptive randomization. Our analysis revealed that Athey et al.'s point estimator remains consistent and asymptotically normal under covariate-adaptive randomization, albeit with the asymptotic variance contingent on the randomization method employed. While Athey et al.'s variance estimator is consistent under simple randomization, it tends to be conservative under more balanced covariate-adaptive randomization strategies like stratified permuted block randomization and minimization. To tackle this issue, we proposed a consistent non-parametric variance estimator to enable valid inferences. Furthermore, to enhance efficiency and broaden the applicability of the inference method, we introduced a stratified transformed difference-in-means estimator and established its asymptotic properties. Through simulation studies and a real data example, we demonstrated the validity and efficiency gains achieved through the proposed methods.

This paper primarily focuses on leveraging stratum information to enhance universality and efficiency. In covariate-adaptive randomized experiments, beyond stratum covariates, additional baseline covariates such as age, gender, and clinic location may also be available. Past studies with light-tailed outcomes have demonstrated that accounting for imbalances in both stratum covariates and additional covariates can yield greater efficiency gains compared to solely adjusting for stratum covariates \citep{Ma2020Regression, wang2021model, liu2020general, Ye2020}. It is worth exploring methodologies for integrating additional high-dimensional covariate information to enhance statistical efficiency, particularly in the presence of heavy-tailed outcomes. Moreover, various strategies exist for handling heavy-tailed data; for instance, \cite{ghosh2021efficiency} introduced a ranking-based approach under simple randomization. Investigating the validity of this estimator under covariate-adaptive randomization presents an intriguing avenue for further research.

\section*{Supplementary material}

The supplementary material provides sufficient conditions for Assumption~\ref{cond:second-moment}(iii), $\sqrt{n}$-consistency of three initial estimators, additional simulation results, and proofs.

\bibliographystyle{apalike} 
\bibliography{causal} 

\newpage
	
\appendix

Section~\ref{sec:sufficient-cond} provides sufficient conditions for Assumption~\ref{cond:second-moment}(iii). Section~\ref{sec:initial} establishes the $\sqrt{n}$-consistency of three initial estimators: the difference-in-medians estimator, the difference-in-weighted-medians estimator, and the weighted average difference-in-medians estimator. Section~\ref{sec:add-simu} provides additional simulation results. Section~\ref{sec:proof-main} presents the proofs of the main results. Section~\ref{sec:lemma} provides the proofs of the lemmas.

\section{Sufficient conditions for Assumption~\ref{cond:second-moment}(iii)}
\label{sec:sufficient-cond}

This section provides sufficient conditions on the density function and its estimates to ensure that Assumption~\ref{cond:second-moment}(iii) holds. These conditions, utilized in \cite{athey2021semiparametric} to ensure that the estimated score function closely approximates the true one, were originally established as properties of a class of density estimation methods in \cite{stone1975adaptive} and \cite{bickel1982adaptive}. We extend these results to accommodate covariate-adaptive randomization, allowing for dependent treatment assignments.

In this and subsequent sections, we assume that $\hat{f}_{0(j)}$ is obtained by the method proposed by \cite{bickel1982adaptive}. Specifically, let $\phi_{\sigma}$ be the probability density function of $N(0,\sigma^2)$. We define the convolution of the empirical density and $\phi_{\sigma}$ as
    $$\hat{f}_{\sigma}(y) = \ncone^{-1}\sum_{i \in \mathcal{L}_1, A_i = 0} \phi_{\sigma}(y-Y_i),$$
    where $n_{a(1)}$ denotes the number of units in $\mathcal{L}_1$ with treatment $A_i = a$, for $a=0,1$.
    
For given $\sigma_n, b_n, c_n, d_n, e_n > 0$, define 
    $$
        g_n(y)=\left\{\begin{array}{ll}
			\frac{\hat{f}_{\sigma_n}^{\prime}}{\hat{f}_{\sigma_n}}(y) & \text { if }  \hat{f}_{\sigma_n}(y) \geq d_n, \ |y| \leq e_n, \ |\hat{f}_{\sigma_n}^{\prime}(y)| \leq c_n \hat{f}_{\sigma_n}(y) \text { and } |\hat{f}_{\sigma_n}^{\prime\prime}(y)| \leq b_n \hat{f}_{\sigma_n}(y),\\
			0 & \text { otherwise.}
		\end{array}\right.
    $$
    We use $g_n(y)$ as our estimated score function, denoted as $\hat{f}_{0(1)}^{\prime}/\hat{f}_{0(1)}$. Similarly, we can obtain the estimated score function $\hat{f}_{0(2)}^{\prime}/\hat{f}_{0(2)}$ based on sample $\{Y_i\}_{i \in \mathcal{L}_2, A_i = 0}$. 
    
Let $n_{[k]0(1)}$ denote the number of units in stratum $k$ within $\mathcal{L}_1$ that receive treatment assignment $A_i=0$. Define $\Delta_n = \max_{k=1,\ldots,K} | n_{[k]0(1)} / \ncone - \ps | $. Note that for any $k$,
$$\frac{\nscone}{\ncone} = \frac{\nscone}{\nsc} \cdot \frac{n_0}{\ncone} \cdot \frac{\nsc}{\ns} \cdot \frac{\ns}{n} \cdot \frac{n}{n_0} \cp \frac{1}{2} \cdot 2 \cdot (1-\pi) \cdot \ps \cdot \frac{1}{1-\pi}= \ps,$$
thus $\Delta_n = o_P(1)$.

\begin{assumption}
\label{cond:density-bickel}
Suppose that $c_n \rightarrow \infty,$ $e_n \rightarrow \infty,$ $ \sigma_n \rightarrow 0$ and $d_n \rightarrow 0$ in such a way that
$\sigma_n c_n \to 0$, $e_n\sigma_n^{-5} = o_P(\min\{n, \Delta_n^{-2}\})$.
\end{assumption}

\begin{remark}
    Assumption \ref{cond:density-bickel} imposes requirements on the hyperparameters of the density estimation. In \cite{bickel1982adaptive}, similar restrictions were proposed with $e_n\sigma_n^{-5} = o_P(\min\{n, \Delta_n^{-2}\})$ replaced by $e_n\sigma_n^{-3} = o(n)$. We introduce an additional term $\Delta_n^{-2}$ due to the more general treatment assignment mechanism used here, which differs from independent and identically distributed samples, resulting in greater variability. However, if the treatment allocation satisfies Assumption \ref{a4}, we can show that $\Delta_n= O_P(n^{-1/2})$, and consequently, this extra term is no longer necessary.
\end{remark}

\begin{assumption}
    \label{cond:partial-score}
    Suppose that for any deterministic sequence $\{\tilde{\delta}_n\}_{n=1,2,...}$, as $\tilde{\delta}_n \to 0$, $$\sup_{y \in \mathbb R,\ k \in \{1,...,K\}} \bigg|\frac{f_{[k]0}(y+\tilde{\delta}_n)}{f_{[k]0}(y)} - 1\bigg| \to 0.$$
\end{assumption}

Assumption \ref{cond:partial-score} characterizes the continuity of the stratum-specific density functions, requiring that the change caused by a small perturbation is a uniformly higher-order term relative to the function value itself. This assumption is satisfied by most common distributions, such as the normal, Laplace, and heavy-tailed Cauchy distributions. Given that the actual distribution of outcomes may not be regular, we can weaken this assumption as follows: for any stratum $k$ and $\tilde{\delta}_n >0$ with $\tilde{\delta}_n = O(n^{-1/2})$, as $n \to \infty$, 
$$\sup_{|y| \in \mathcal{Y}_n} \bigg|\frac{f_{[k]0}^{\prime}(y)}{f_{[k]0}(y)}\bigg| = o(n^{1/2}) \quad \text{ and } \quad \sup_{|y| \in \mathcal{Y}_n, \ |\Bar{\delta}| < \delta_n} \bigg|\frac{f_{[k]0}^{\prime\prime}(y+\Bar{\delta})}{f_{[k]0}(y)}\bigg| = o(n),$$
where $\mathcal{Y}_n = \{y: |f_{[k]0}(y^{\prime})|<(b_n+c_n^2)^{-1}, \ \forall \ y-\delta_n \leq y^{\prime} \leq y+\delta_n\} \cap \{y: |y| < e_n\}$ with $b_n, c_n, e_n > 0$ satisfying Assumption \ref{cond:density-bickel}. Noting that the hyperparameters can be reasonably chosen to ensure that $\mathcal{Y}_n$ does not contain singular points, this weaker assumption can be easily satisfied.

\section{Results on the initial estimator}
\label{sec:initial}
Recall that $f_{[k]0}(y)$ and $F_{[k]0}(y)$ are the density function and cumulative distribution function of $Y_i(0)$, conditional on $S_i = k$. Let $\hat{\tau}_{\textnormal{md}}$ denote the difference between the medians of the outcomes in the treatment and control groups, and $\hat{\tau}_{\textnormal{wt-md}}$ denote the difference between the weighted medians of the outcomes in the treatment and control groups with weights $\pi_{n[S_i]}^{-1}$ and $(1-\pi_{n[S_i]})^{-1}$, respectively. Let $q_a(\tau)$ be the $\tau$-th quantile of $Y(a)$, $a=0,1$. Let $E$ be any compact subset of $(0,1)$.

\begin{assumption}
    \label{assumption:median}
    Suppose that (i) $f_0\left(q_0(\tau)\right)$ and $f_{[k]0}\left(q_0(\tau)\right)$ are bounded away from zero and infinity uniformly over $\tau \in E$ and $k=1,...,K$, and 
    (ii) $f_0(\cdot)$ and $f_{[k]0}(\cdot)$ are Lipschitz continuous over $\left\{q_a(\tau): \tau \in E\right\}$.
\end{assumption}

Assumption \ref{assumption:median}, proposed by \cite{zhang2020quantile}, is satisfied by many commonly used distributions, including those with heavy tails (e.g., the Cauchy distribution).

\begin{prop}
\label{prop:root-n-consistent}
(i) Under Assumption \ref{a1}, \ref{a2}, \ref{a4}, and \ref{assumption:median},
$\sqrt{n} (\hat{\tau}_{\textnormal{md}} - \tau) = O_P(1)$; (ii) Under Assumption \ref{a1}--\ref{a3} and \ref{assumption:median},
$\sqrt{n} (\hat{\tau}_{\textnormal{wt-md}} - \tau) = O_P(1)$. 
\end{prop}

Proposition~\ref{prop:root-n-consistent} is a direct result of \citet[][Theorem 3.1 and Theorem 3.2]{zhang2020quantile}, which indicates that the initial estimators $\hat{\tau}_{\textnormal{md}}$ and $\hat{\tau}_{\textnormal{wt-md}}$ are $\sqrt{n}$-consistent.

Let $\hat{\tau}_{\textnormal{str-md}}$ be the weighted average difference-in-medians estimator, computed as the weighted summation of the difference-in-medians estimator within each stratum, with weights being the proportions of the stratum sizes. In order to establish the $\sqrt{n}$-consistency of $\hat{\tau}_{\textnormal{str-md}}$, we require a milder Condition~\ref{a5} as outlined below.

\begin{assumption}
		\label{a5}
  For $(g,G)=(f_a,F_a)$ and $(g,G)= (f_{[k]a},F_{[k]a})$, where $a=0,1$ and $k=1,\ldots,K$, suppose that $g$ is twice continuously differentiable and
there exist constants $C_0>0$ and $\epsilon \in\left(0, 1/2 \right)$ such that 
\begin{itemize}
\item[(i)] $G(x)(1-G(x)) \left|g^{\prime}(x)\right|/ g^2(x) \leq C_0$ if $G(x) \in(0, \epsilon)$ or $G(x) \in(1-\epsilon, 1)$, and $C_0$ can be replaced by another constant if $G(x) \in[\epsilon, 1-\epsilon]$; 
\item[(ii)] $g^{\prime}(x) \geq 0$ for $x<G^{-1}(\epsilon)$ and $g^{\prime}(x) \leq 0$ for $x>G^{-1}(1-\epsilon)$;
\item[(iii)]
$g(G^{-1}(1/2)) > 0$.
\end{itemize}
\end{assumption}

Assumption \ref{a5} serves as a constraint on the outcome distribution. Assumption \ref{a5}(i)--(ii) were introduced by \cite{csorgo1978strong} to ensure the stationarity of the distribution as a quantile process. Assumption \ref{a5}(iii) mandates the positivity of the density at the population median, implying finite second moments for the sample median estimator. In practice, these conditions are widely applicable and can be met by many common distributions, such as the normal distribution and the $t$ distribution, including the heavy-tailed Cauchy distribution. To extend the results to covariate-adaptive randomization and mitigate local errors arising from strata, we strengthen the requirement for each stratum to hold (i.e., $g=f_{[k]0}$ and $G=F_{[k]0}$).

Let $\tau_{[k]}= F_{[k]1}^{-1}(1/2) - F_{[k]0}^{-1}(1/2)$ denote the median treatment effect in stratum $k$.

\begin{prop}
\label{prop:root-n-consistent-str}
Under Assumptions \ref{a1}--\ref{a3}, and \ref{a5}, if $\sumk \ps \tau_{[k]} = \tau$, then $\sqrt{n} (\hat{\tau}_{\textnormal{str-md}} - \tau) = O_P(1)$. 
\end{prop}

Proposition~\ref{prop:root-n-consistent-str} implies that the weighted average difference-in-medians estimator $\hat{\tau}_{\textnormal{str-md}}$ is also $\sqrt{n}$-consistent. It is interesting to observe that Proposition~\ref{prop:root-n-consistent-str} requires the weaker Assumption~\ref{a3} rather than the stronger Assumption~\ref{a4} needed for Proposition~\ref{prop:root-n-consistent}(i). Consequently, $\hat{\tau}_{\textnormal{str-md}}$ has wider applicability compared to  $\hat{\tau}_{\textnormal{md}}$. The proof of Proposition~\ref{prop:root-n-consistent-str} will be provided in Section~\ref{sec:proof-initial}.

\section{Additional simulation results}
\label{sec:add-simu}

\subsection{Unequal treatment probability}

In this section, we evaluate the finite sample performance of the proposed estimators with a treatment probability of $\pi=1/3$. The other settings remain similar to those in the main text, except that the biased-coin probability for Pocock and Simon's minimization is adjusted to 0.75. Tables \ref{tab:unequalpi-model1}--\ref{tab:unequalpi-model3} present the results. The conclusions are consistent with those in the case of equal treatment probability ($\pi=1/2$): we observe negligible bias and a significant reduction in standard deviation for $\hat{\tau}_{\dim}$ and $\hat{\tau}_{\str}$ compared to their initial estimators. Additionally, for the same model and tail, $\hat{\tau}_{\str}$ exhibits comparable standard deviations across all designs. These standard deviations are similar to those of $\hat{\tau}_{\dim}$ under stratified randomization and are smaller than those of $\hat{\tau}_{\dim}$ under simple randomization. A noteworthy phenomenon is that increasing the sample size used for density estimation improves the accuracy of both $\hat{\tau}_{\dim}$ and $\hat{\tau}_{\str}$. For example, compared to the equal treatment probability case in the main text, the standard deviations increase by 4.5\% for $\hat{\tau}_{\dim}$ under simple randomization across all models and tails. Theoretically, this increase should be  6.1\%, indicating that having more control group samples reduces the error in density estimation.

\begin{table}[htbp]
  \centering
  \caption{Simulation results under Model 1 for $n=1000$ and $\pi=1/3$}
  \begin{threeparttable}
    \resizebox{\textwidth}{!}{
    \begin{tabular}{rlrrrrrrrrrrrr}
    \toprule
          &       & \multicolumn{4}{c}{SR}        & \multicolumn{4}{c}{STR}       & \multicolumn{4}{c}{MIN} \\
\cmidrule{3-14}          &       & \multicolumn{1}{l}{Bias} & \multicolumn{1}{l}{SD} & \multicolumn{1}{l}{SE} & \multicolumn{1}{l}{CP} & \multicolumn{1}{l}{Bias} & \multicolumn{1}{l}{SD} & \multicolumn{1}{l}{SE} & \multicolumn{1}{l}{CP} & \multicolumn{1}{l}{Bias} & \multicolumn{1}{l}{SD} & \multicolumn{1}{l}{SE} & \multicolumn{1}{l}{CP} \\
    \midrule
    \multicolumn{1}{l}{$\tau$ = 0} &       &       &       &       &       &       &       &       &       &       &       &       &  \\
    \multicolumn{1}{l}{Normal tail} & $\hat{\tau}_{\textnormal{naive-dim}}$ & 0.002  & 0.087  & 0.088  & 0.947  & -0.003  & 0.070  & 0.073  & 0.955  & 0.004  & 0.083  & /     & / \\
          & $\hat{\tau}_{\textnormal{str-dim}}$ & 0.002  & 0.071  & 0.073  & 0.953  & -0.003  & 0.070  & 0.073  & 0.956  & 0.002  & 0.073  & 0.073  & 0.947  \\
          & $\hat{\tau}_{\textnormal{md}}$ & 0.003  & 0.113  & 0.126  & 0.973  & 0.000  & 0.099  & 0.113  & 0.972  & 0.007  & 0.108  & /     & / \\
          & $\hat{\tau}_{\textnormal{wt-md}}$ & 0.002  & 0.098  & 0.113  & 0.980  & 0.000  & 0.098  & 0.113  & 0.970  & 0.005  & 0.097  & 0.113  & 0.976  \\
          & $\hat{\tau}_{\dim}$ & 0.002  & 0.095  & 0.097  & 0.957  & -0.002  & 0.079  & 0.087  & 0.957  & 0.003  & 0.092  & /     & / \\
          & $\hat{\tau}_{\str}$ & 0.002  & 0.080  & 0.087  & 0.963  & -0.002  & 0.079  & 0.087  & 0.955  & 0.002  & 0.083  & 0.087  & 0.960  \\
          &       &       &       &       &       &       &       &       &       &       &       &       &  \\
    \multicolumn{1}{l}{Laplace tail} & $\hat{\tau}_{\textnormal{naive-dim}}$ & -0.004  & 0.113  & 0.111  & 0.951  & -0.004  & 0.098  & 0.099  & 0.956  & -0.001  & 0.106  & /     & / \\
          & $\hat{\tau}_{\textnormal{str-dim}}$ & -0.005  & 0.101  & 0.099  & 0.949  & -0.004  & 0.098  & 0.099  & 0.955  & -0.003  & 0.098  & 0.099  & 0.948  \\
          & $\hat{\tau}_{\textnormal{md}}$ & -0.005  & 0.125  & 0.128  & 0.953  & 0.000  & 0.108  & 0.117  & 0.966  & -0.001  & 0.121  & /     & / \\
          & $\hat{\tau}_{\textnormal{wt-md}}$ & -0.006  & 0.112  & 0.117  & 0.955  & 0.000  & 0.108  & 0.117  & 0.969  & -0.003  & 0.113  & 0.116  & 0.960  \\
          & $\hat{\tau}_{\dim}$ & -0.006  & 0.109  & 0.109  & 0.950  & -0.002  & 0.093  & 0.098  & 0.963  & -0.001  & 0.106  & /     & / \\
          & $\hat{\tau}_{\str}$ & -0.006  & 0.095  & 0.098  & 0.955  & -0.002  & 0.093  & 0.098  & 0.963  & -0.003  & 0.098  & 0.098  & 0.955  \\
          &       &       &       &       &       &       &       &       &       &       &       &       &  \\
    \multicolumn{1}{l}{Cauchy tail} & $\hat{\tau}_{\textnormal{naive-dim}}$ & 1.036  & 28.982  & 7.934  & 0.980  & -1.745  & 30.787  & 8.247  & 0.976  & 0.636  & 24.073  & /     & / \\
          & $\hat{\tau}_{\textnormal{str-dim}}$ & 1.053  & 29.509  & 7.919  & 0.982  & -1.747  & 30.851  & 8.247  & 0.976  & 0.632  & 24.325  & 7.758  & 0.979  \\
          & $\hat{\tau}_{\textnormal{md}}$ & -0.002  & 0.153  & 0.156  & 0.958  & 0.003  & 0.144  & 0.146  & 0.957  & 0.000  & 0.152  & /     & / \\
          & $\hat{\tau}_{\textnormal{wt-md}}$ & -0.002  & 0.143  & 0.146  & 0.950  & 0.003  & 0.144  & 0.146  & 0.957  & -0.002  & 0.147  & 0.146  & 0.945  \\
          & $\hat{\tau}_{\dim}$ & -0.002  & 0.131  & 0.132  & 0.950  & 0.004  & 0.126  & 0.123  & 0.939  & 0.001  & 0.131  & /     & / \\
          & $\hat{\tau}_{\str}$ & -0.002  & 0.120  & 0.122  & 0.953  & 0.004  & 0.126  & 0.123  & 0.941  & -0.001  & 0.126  & 0.122  & 0.945  \\
    \multicolumn{1}{l}{$\tau$ = 1} &       &       &       &       &       &       &       &       &       &       &       &       &  \\
    \multicolumn{1}{l}{Normal tail} & $\hat{\tau}_{\textnormal{naive-dim}}$ & -0.001  & 0.087  & 0.089  & 0.954  & -0.001  & 0.073  & 0.073  & 0.944  & 0.001  & 0.084  & /     & / \\
          & $\hat{\tau}_{\textnormal{str-dim}}$ & -0.001  & 0.073  & 0.073  & 0.947  & -0.002  & 0.073  & 0.073  & 0.946  & 0.000  & 0.071  & 0.073  & 0.955  \\
          & $\hat{\tau}_{\textnormal{md}}$ & 0.002  & 0.113  & 0.126  & 0.973  & -0.001  & 0.097  & 0.113  & 0.981  & -0.001  & 0.112  & /     & / \\
          & $\hat{\tau}_{\textnormal{wt-md}}$ & 0.002  & 0.101  & 0.113  & 0.964  & -0.002  & 0.097  & 0.113  & 0.980  & -0.002  & 0.100  & 0.113  & 0.978  \\
          & $\hat{\tau}_{\dim}$ & -0.001  & 0.096  & 0.098  & 0.953  & -0.001  & 0.083  & 0.087  & 0.957  & 0.002  & 0.094  & /     & / \\
          & $\hat{\tau}_{\str}$ & -0.002  & 0.082  & 0.087  & 0.962  & -0.001  & 0.083  & 0.087  & 0.958  & 0.001  & 0.084  & 0.087  & 0.958  \\
          &       &       &       &       &       &       &       &       &       &       &       &       &  \\
    \multicolumn{1}{l}{Laplace tail} & $\hat{\tau}_{\textnormal{naive-dim}}$ & 0.000  & 0.109  & 0.111  & 0.957  & -0.001  & 0.096  & 0.099  & 0.960  & 0.004  & 0.107  & /     & / \\
          & $\hat{\tau}_{\textnormal{str-dim}}$ & -0.001  & 0.100  & 0.099  & 0.953  & -0.001  & 0.096  & 0.099  & 0.963  & 0.004  & 0.097  & 0.099  & 0.951  \\
          & $\hat{\tau}_{\textnormal{md}}$ & 0.002  & 0.123  & 0.127  & 0.950  & -0.004  & 0.107  & 0.116  & 0.963  & 0.003  & 0.116  & /     & / \\
          & $\hat{\tau}_{\textnormal{wt-md}}$ & 0.000  & 0.110  & 0.116  & 0.955  & -0.004  & 0.107  & 0.116  & 0.961  & 0.002  & 0.107  & 0.116  & 0.970  \\
          & $\hat{\tau}_{\dim}$ & 0.002  & 0.110  & 0.109  & 0.951  & -0.003  & 0.093  & 0.098  & 0.958  & 0.004  & 0.103  & /     & / \\
          & $\hat{\tau}_{\str}$ & 0.001  & 0.097  & 0.098  & 0.946  & -0.003  & 0.093  & 0.098  & 0.958  & 0.004  & 0.092  & 0.098  & 0.962  \\
          &       &       &       &       &       &       &       &       &       &       &       &       &  \\
    \multicolumn{1}{l}{Cauchy tail} & $\hat{\tau}_{\textnormal{naive-dim}}$ & 1.661  & 36.639  & 8.216  & 0.980  & -1.423  & 31.413  & 7.947  & 0.976  & 0.469  & 32.971  & /     & / \\
          & $\hat{\tau}_{\textnormal{str-dim}}$ & 1.600  & 35.465  & 8.198  & 0.980  & -1.421  & 31.421  & 7.947  & 0.976  & 0.509  & 34.453  & 7.792  & 0.975  \\
          & $\hat{\tau}_{\textnormal{md}}$ & 0.007  & 0.156  & 0.155  & 0.941  & 0.004  & 0.147  & 0.146  & 0.938  & -0.003  & 0.152  & /     & / \\
          & $\hat{\tau}_{\textnormal{wt-md}}$ & 0.006  & 0.146  & 0.146  & 0.948  & 0.003  & 0.146  & 0.146  & 0.940  & -0.004  & 0.144  & 0.146  & 0.949  \\
          & $\hat{\tau}_{\dim}$ & 0.004  & 0.138  & 0.133  & 0.935  & 0.000  & 0.125  & 0.123  & 0.941  & 0.002  & 0.134  & /     & / \\
          & $\hat{\tau}_{\str}$ & 0.003  & 0.127  & 0.123  & 0.929  & 0.000  & 0.125  & 0.123  & 0.940  & 0.001  & 0.126  & 0.123  & 0.947  \\
    \bottomrule
    \end{tabular}%
    }
    \end{threeparttable}
  \label{tab:unequalpi-model1}
  \begin{tablenotes}
				\footnotesize
				\item[] Note: $\hat{\tau}_{\textnormal{naive-dim}}$, standard difference-in-means estimator; $\hat{\tau}_{\textnormal{str-dim}}$, stratified difference-in-means esti-\\mator; $\hat{\tau}_{\textnormal{md}}$, difference-in-medians estimator; $\hat{\tau}_{\textnormal{wt-md}}$, difference-in-weighted-medians estimator; $\hat{\tau}_{\dim}$,\\ transformed difference-in-means estimator; $\hat{\tau}_{\str}$, stratified transformed difference-in-means estimator;\\ SR, simple randomization; STR, stratified randomization; MIN, Pocock and Simon's minimization;\\ SD, standard deviation; SE, standard error; CP, coverage probability.
			\end{tablenotes}
\end{table}%

\begin{table}[htbp]
  \centering
  \caption{Simulation results under Model 2 for $n=1000$ and $\pi=1/3$}
  \begin{threeparttable}
    \resizebox{\textwidth}{!}{
    \begin{tabular}{rlrrrrrrrrrrrr}
    \toprule
          &       & \multicolumn{4}{c}{SR}        & \multicolumn{4}{c}{STR}       & \multicolumn{4}{c}{MIN} \\
\cmidrule{3-14}          &       & \multicolumn{1}{l}{Bias} & \multicolumn{1}{l}{SD} & \multicolumn{1}{l}{SE} & \multicolumn{1}{l}{CP} & \multicolumn{1}{l}{Bias} & \multicolumn{1}{l}{SD} & \multicolumn{1}{l}{SE} & \multicolumn{1}{l}{CP} & \multicolumn{1}{l}{Bias} & \multicolumn{1}{l}{SD} & \multicolumn{1}{l}{SE} & \multicolumn{1}{l}{CP} \\
    \midrule
    \multicolumn{1}{l}{$\tau$ = 0} &       &       &       &       &       &       &       &       &       &       &       &       &  \\
    \multicolumn{1}{l}{Normal tail} & $\hat{\tau}_{\textnormal{naive-dim}}$ & 0.001  & 0.089  & 0.088  & 0.951  & -0.002  & 0.072  & 0.074  & 0.952  & 0.005  & 0.079  & /     & / \\
          & $\hat{\tau}_{\textnormal{str-dim}}$ & 0.001  & 0.075  & 0.074  & 0.946  & -0.002  & 0.072  & 0.074  & 0.951  & 0.005  & 0.068  & 0.074  & 0.966  \\
          & $\hat{\tau}_{\textnormal{md}}$ & -0.001  & 0.106  & 0.116  & 0.967  & -0.001  & 0.094  & 0.107  & 0.968  & 0.006  & 0.097  & /     & / \\
          & $\hat{\tau}_{\textnormal{wt-md}}$ & -0.002  & 0.096  & 0.107  & 0.971  & 0.000  & 0.094  & 0.107  & 0.970  & 0.004  & 0.090  & 0.107  & 0.984  \\
          & $\hat{\tau}_{\dim}$ & 0.002  & 0.094  & 0.091  & 0.948  & 0.000  & 0.080  & 0.084  & 0.963  & 0.008  & 0.083  & /     & / \\
          & $\hat{\tau}_{\str}$ & 0.001  & 0.084  & 0.084  & 0.954  & 0.001  & 0.080  & 0.084  & 0.963  & 0.007  & 0.076  & 0.084  & 0.970  \\
          &       &       &       &       &       &       &       &       &       &       &       &       &  \\
    \multicolumn{1}{l}{Laplace tail} & $\hat{\tau}_{\textnormal{naive-dim}}$ & -0.008  & 0.112  & 0.110  & 0.944  & -0.004  & 0.098  & 0.100  & 0.949  & 0.004  & 0.106  & /     & / \\
          & $\hat{\tau}_{\textnormal{str-dim}}$ & -0.009  & 0.101  & 0.100  & 0.950  & -0.004  & 0.098  & 0.100  & 0.949  & 0.003  & 0.097  & 0.100  & 0.949  \\
          & $\hat{\tau}_{\textnormal{md}}$ & -0.008  & 0.113  & 0.114  & 0.951  & 0.000  & 0.105  & 0.105  & 0.953  & 0.003  & 0.109  & /     & / \\
          & $\hat{\tau}_{\textnormal{wt-md}}$ & -0.009  & 0.103  & 0.105  & 0.951  & 0.000  & 0.105  & 0.105  & 0.956  & 0.002  & 0.103  & 0.106  & 0.949  \\
          & $\hat{\tau}_{\dim}$ & -0.005  & 0.101  & 0.102  & 0.956  & 0.002  & 0.093  & 0.095  & 0.955  & 0.007  & 0.100  & /     & / \\
          & $\hat{\tau}_{\str}$ & -0.005  & 0.092  & 0.095  & 0.957  & 0.002  & 0.093  & 0.095  & 0.954  & 0.006  & 0.093  & 0.095  & 0.956  \\
          &       &       &       &       &       &       &       &       &       &       &       &       &  \\
    \multicolumn{1}{l}{Cauchy tail} & $\hat{\tau}_{\textnormal{naive-dim}}$ & 8.394  & 270.475  & 15.776  & 0.979  & -18.836  & 535.206  & 23.786  & 0.979  & 9.654  & 268.566  & /     & / \\
          & $\hat{\tau}_{\textnormal{str-dim}}$ & 8.977  & 288.486  & 15.755  & 0.980  & -18.904  & 537.352  & 23.786  & 0.978  & 10.390  & 287.253  & 16.386  & 0.982  \\
          & $\hat{\tau}_{\textnormal{md}}$ & 0.000  & 0.143  & 0.146  & 0.951  & -0.005  & 0.132  & 0.138  & 0.964  & 0.003  & 0.133  & /     & / \\
          & $\hat{\tau}_{\textnormal{wt-md}}$ & -0.001  & 0.134  & 0.139  & 0.959  & -0.005  & 0.132  & 0.138  & 0.964  & 0.002  & 0.126  & 0.139  & 0.973  \\
          & $\hat{\tau}_{\dim}$ & -0.002  & 0.133  & 0.130  & 0.949  & -0.008  & 0.118  & 0.122  & 0.957  & 0.000  & 0.124  & /     & / \\
          & $\hat{\tau}_{\str}$ & -0.002  & 0.124  & 0.122  & 0.950  & -0.007  & 0.118  & 0.122  & 0.958  & 0.000  & 0.117  & 0.122  & 0.952  \\
    \multicolumn{1}{l}{$\tau$ = 1} &       &       &       &       &       &       &       &       &       &       &       &       &  \\
    \multicolumn{1}{l}{Normal tail} & $\hat{\tau}_{\textnormal{naive-dim}}$ & 0.000  & 0.088  & 0.088  & 0.950  & 0.003  & 0.076  & 0.074  & 0.949  & 0.003  & 0.083  & /     & / \\
          & $\hat{\tau}_{\textnormal{str-dim}}$ & -0.001  & 0.075  & 0.074  & 0.944  & 0.003  & 0.076  & 0.074  & 0.950  & 0.003  & 0.072  & 0.074  & 0.955  \\
          & $\hat{\tau}_{\textnormal{md}}$ & -0.002  & 0.105  & 0.116  & 0.963  & 0.001  & 0.097  & 0.107  & 0.968  & 0.004  & 0.105  & /     & / \\
          & $\hat{\tau}_{\textnormal{wt-md}}$ & -0.003  & 0.095  & 0.107  & 0.971  & 0.001  & 0.097  & 0.107  & 0.967  & 0.004  & 0.098  & 0.107  & 0.970  \\
          & $\hat{\tau}_{\dim}$ & 0.000  & 0.092  & 0.090  & 0.949  & 0.001  & 0.081  & 0.084  & 0.953  & 0.004  & 0.088  & /     & / \\
          & $\hat{\tau}_{\str}$ & -0.001  & 0.083  & 0.084  & 0.955  & 0.001  & 0.081  & 0.084  & 0.951  & 0.003  & 0.082  & 0.084  & 0.951  \\
          &       &       &       &       &       &       &       &       &       &       &       &       &  \\
    \multicolumn{1}{l}{Laplace tail} & $\hat{\tau}_{\textnormal{naive-dim}}$ & 0.004  & 0.111  & 0.111  & 0.950  & 0.000  & 0.096  & 0.100  & 0.949  & 0.001  & 0.109  & /     & / \\
          & $\hat{\tau}_{\textnormal{str-dim}}$ & 0.004  & 0.101  & 0.100  & 0.949  & 0.000  & 0.096  & 0.100  & 0.949  & 0.001  & 0.098  & 0.100  & 0.954  \\
          & $\hat{\tau}_{\textnormal{md}}$ & 0.001  & 0.113  & 0.114  & 0.948  & 0.001  & 0.107  & 0.105  & 0.941  & 0.004  & 0.113  & /     & / \\
          & $\hat{\tau}_{\textnormal{wt-md}}$ & 0.001  & 0.102  & 0.105  & 0.954  & 0.001  & 0.107  & 0.105  & 0.946  & 0.004  & 0.104  & 0.105  & 0.950  \\
          & $\hat{\tau}_{\dim}$ & 0.003  & 0.105  & 0.102  & 0.947  & 0.000  & 0.093  & 0.095  & 0.945  & 0.004  & 0.103  & /     & / \\
          & $\hat{\tau}_{\str}$ & 0.002  & 0.095  & 0.095  & 0.940  & 0.000  & 0.093  & 0.095  & 0.942  & 0.003  & 0.095  & 0.095  & 0.948  \\
          &       &       &       &       &       &       &       &       &       &       &       &       &  \\
    \multicolumn{1}{l}{Cauchy tail} & $\hat{\tau}_{\textnormal{naive-dim}}$ & -1.328  & 50.021  & 9.697  & 0.983  & -0.551  & 40.249  & 8.842  & 0.984  & 0.665  & 32.721  & /     & / \\
          & $\hat{\tau}_{\textnormal{str-dim}}$ & -1.262  & 49.057  & 9.675  & 0.982  & -0.554  & 40.244  & 8.842  & 0.984  & 0.641  & 33.373  & 8.507  & 0.975  \\
          & $\hat{\tau}_{\textnormal{md}}$ & 0.001  & 0.137  & 0.146  & 0.960  & 0.006  & 0.141  & 0.139  & 0.939  & 0.001  & 0.146  & /     & / \\
          & $\hat{\tau}_{\textnormal{wt-md}}$ & 0.000  & 0.129  & 0.139  & 0.958  & 0.006  & 0.142  & 0.139  & 0.939  & 0.002  & 0.140  & 0.139  & 0.943  \\
          & $\hat{\tau}_{\dim}$ & 0.002  & 0.123  & 0.130  & 0.966  & 0.004  & 0.125  & 0.122  & 0.939  & -0.005  & 0.130  & /     & / \\
          & $\hat{\tau}_{\str}$ & 0.001  & 0.116  & 0.122  & 0.965  & 0.004  & 0.125  & 0.122  & 0.940  & -0.005  & 0.123  & 0.122  & 0.943  \\
    \bottomrule
    \end{tabular}%
    }
    \end{threeparttable}
  \label{tab:unequalpi-model2}
  \begin{tablenotes}
				\footnotesize
				\item[] Note: $\hat{\tau}_{\textnormal{naive-dim}}$, standard difference-in-means estimator; $\hat{\tau}_{\textnormal{str-dim}}$, stratified difference-in-means esti-\\mator; $\hat{\tau}_{\textnormal{md}}$, difference-in-medians estimator; $\hat{\tau}_{\textnormal{wt-md}}$, difference-in-weighted-medians estimator; $\hat{\tau}_{\dim}$,\\ transformed difference-in-means estimator; $\hat{\tau}_{\str}$, stratified transformed difference-in-means estimator;\\ SR, simple randomization; STR, stratified randomization; MIN, Pocock and Simon's minimization;\\ SD, standard deviation; SE, standard error; CP, coverage probability.
			\end{tablenotes}
\end{table}%

\begin{table}[htbp]
  \centering
  \caption{Simulation results under Model 3 for $n=1000$ and $\pi=1/3$}
  \begin{threeparttable}
    \resizebox{\textwidth}{!}{
    \begin{tabular}{rlrrrrrrrrrrrr}
    \toprule
          &       & \multicolumn{4}{c}{SR}        & \multicolumn{4}{c}{STR}       & \multicolumn{4}{c}{MIN} \\
\cmidrule{3-14}          &       & \multicolumn{1}{l}{Bias} & \multicolumn{1}{l}{SD} & \multicolumn{1}{l}{SE} & \multicolumn{1}{l}{CP} & \multicolumn{1}{l}{Bias} & \multicolumn{1}{l}{SD} & \multicolumn{1}{l}{SE} & \multicolumn{1}{l}{CP} & \multicolumn{1}{l}{Bias} & \multicolumn{1}{l}{SD} & \multicolumn{1}{l}{SE} & \multicolumn{1}{l}{CP} \\
    \midrule
    \multicolumn{1}{l}{$\tau$ = 0} &       &       &       &       &       &       &       &       &       &       &       &       &  \\
    \multicolumn{1}{l}{Normal tail} & $\hat{\tau}_{\textnormal{naive-dim}}$ & -0.003  & 0.088  & 0.090  & 0.953  & -0.004  & 0.072  & 0.077  & 0.964  & 0.007  & 0.090  & /     & / \\
          & $\hat{\tau}_{\textnormal{str-dim}}$ & -0.003  & 0.075  & 0.076  & 0.942  & -0.004  & 0.072  & 0.077  & 0.962  & 0.005  & 0.077  & 0.077  & 0.957  \\
          & $\hat{\tau}_{\textnormal{md}}$ & -0.004  & 0.109  & 0.123  & 0.973  & -0.003  & 0.093  & 0.112  & 0.978  & 0.008  & 0.109  & /     & / \\
          & $\hat{\tau}_{\textnormal{wt-md}}$ & -0.004  & 0.099  & 0.112  & 0.970  & -0.003  & 0.093  & 0.112  & 0.977  & 0.007  & 0.098  & 0.112  & 0.977  \\
          & $\hat{\tau}_{\dim}$ & -0.001  & 0.094  & 0.095  & 0.948  & -0.002  & 0.079  & 0.087  & 0.973  & 0.009  & 0.093  & /     & / \\
          & $\hat{\tau}_{\str}$ & -0.002  & 0.083  & 0.087  & 0.952  & -0.002  & 0.079  & 0.087  & 0.972  & 0.007  & 0.084  & 0.087  & 0.955  \\
          &       &       &       &       &       &       &       &       &       &       &       &       &  \\
    \multicolumn{1}{l}{Laplace tail} & $\hat{\tau}_{\textnormal{naive-dim}}$ & -0.001  & 0.113  & 0.112  & 0.950  & 0.001  & 0.099  & 0.101  & 0.953  & 0.004  & 0.109  & /     & / \\
          & $\hat{\tau}_{\textnormal{str-dim}}$ & 0.000  & 0.102  & 0.101  & 0.948  & 0.001  & 0.098  & 0.101  & 0.953  & 0.002  & 0.099  & 0.102  & 0.953  \\
          & $\hat{\tau}_{\textnormal{md}}$ & 0.004  & 0.120  & 0.120  & 0.944  & 0.002  & 0.106  & 0.111  & 0.953  & 0.006  & 0.113  & /     & / \\
          & $\hat{\tau}_{\textnormal{wt-md}}$ & 0.004  & 0.110  & 0.111  & 0.950  & 0.002  & 0.106  & 0.111  & 0.952  & 0.004  & 0.104  & 0.111  & 0.961  \\
          & $\hat{\tau}_{\dim}$ & 0.004  & 0.107  & 0.108  & 0.950  & 0.005  & 0.094  & 0.099  & 0.965  & 0.005  & 0.105  & /     & / \\
          & $\hat{\tau}_{\str}$ & 0.004  & 0.097  & 0.098  & 0.951  & 0.005  & 0.094  & 0.099  & 0.965  & 0.003  & 0.095  & 0.098  & 0.957  \\
          &       &       &       &       &       &       &       &       &       &       &       &       &  \\
    \multicolumn{1}{l}{Cauchy tail} & $\hat{\tau}_{\textnormal{naive-dim}}$ & 5.178  & 174.538  & 17.515  & 0.982  & 1.498  & 203.644  & 19.977  & 0.983  & 8.757  & 167.820  & /     & / \\
          & $\hat{\tau}_{\textnormal{str-dim}}$ & 5.195  & 171.669  & 17.484  & 0.979  & 1.521  & 203.289  & 19.977  & 0.983  & 8.507  & 163.560  & 16.408  & 0.975  \\
          & $\hat{\tau}_{\textnormal{md}}$ & 0.001  & 0.152  & 0.152  & 0.950  & -0.003  & 0.150  & 0.144  & 0.941  & 0.000  & 0.146  & /     & / \\
          & $\hat{\tau}_{\textnormal{wt-md}}$ & 0.002  & 0.143  & 0.145  & 0.955  & -0.003  & 0.150  & 0.144  & 0.942  & -0.001  & 0.137  & 0.145  & 0.960  \\
          & $\hat{\tau}_{\dim}$ & -0.001  & 0.136  & 0.133  & 0.940  & -0.003  & 0.131  & 0.124  & 0.937  & 0.001  & 0.134  & /     & / \\
          & $\hat{\tau}_{\str}$ & -0.001  & 0.125  & 0.124  & 0.947  & -0.003  & 0.131  & 0.124  & 0.937  & -0.001  & 0.123  & 0.124  & 0.953  \\
    \multicolumn{1}{l}{$\tau$ = 1} &       &       &       &       &       &       &       &       &       &       &       &       &  \\
    \multicolumn{1}{l}{Normal tail} & $\hat{\tau}_{\textnormal{naive-dim}}$ & -0.005  & 0.092  & 0.090  & 0.945  & -0.002  & 0.078  & 0.077  & 0.944  & -0.002  & 0.086  & /     & / \\
          & $\hat{\tau}_{\textnormal{str-dim}}$ & -0.003  & 0.077  & 0.077  & 0.946  & -0.002  & 0.078  & 0.077  & 0.947  & -0.002  & 0.074  & 0.076  & 0.955  \\
          & $\hat{\tau}_{\textnormal{md}}$ & -0.003  & 0.111  & 0.123  & 0.966  & -0.001  & 0.098  & 0.112  & 0.977  & -0.002  & 0.109  & /     & / \\
          & $\hat{\tau}_{\textnormal{wt-md}}$ & -0.001  & 0.098  & 0.112  & 0.976  & -0.001  & 0.098  & 0.112  & 0.976  & -0.002  & 0.100  & 0.112  & 0.974  \\
          & $\hat{\tau}_{\dim}$ & -0.005  & 0.097  & 0.095  & 0.945  & -0.001  & 0.083  & 0.087  & 0.963  & 0.001  & 0.092  & /     & / \\
          & $\hat{\tau}_{\str}$ & -0.003  & 0.085  & 0.087  & 0.955  & -0.001  & 0.083  & 0.087  & 0.963  & 0.000  & 0.083  & 0.087  & 0.957  \\
          &       &       &       &       &       &       &       &       &       &       &       &       &  \\
    \multicolumn{1}{l}{Laplace tail} & $\hat{\tau}_{\textnormal{naive-dim}}$ & -0.008  & 0.114  & 0.112  & 0.949  & 0.000  & 0.099  & 0.102  & 0.956  & 0.001  & 0.105  & /     & / \\
          & $\hat{\tau}_{\textnormal{str-dim}}$ & -0.006  & 0.103  & 0.102  & 0.948  & 0.000  & 0.099  & 0.102  & 0.955  & 0.001  & 0.098  & 0.102  & 0.954  \\
          & $\hat{\tau}_{\textnormal{md}}$ & -0.004  & 0.121  & 0.120  & 0.948  & 0.002  & 0.107  & 0.111  & 0.951  & 0.001  & 0.112  & /     & / \\
          & $\hat{\tau}_{\textnormal{wt-md}}$ & -0.001  & 0.108  & 0.111  & 0.954  & 0.002  & 0.108  & 0.111  & 0.949  & 0.002  & 0.105  & 0.111  & 0.962  \\
          & $\hat{\tau}_{\dim}$ & -0.002  & 0.112  & 0.108  & 0.947  & 0.003  & 0.096  & 0.099  & 0.960  & 0.006  & 0.103  & /     & / \\
          & $\hat{\tau}_{\str}$ & -0.001  & 0.099  & 0.098  & 0.955  & 0.003  & 0.096  & 0.099  & 0.960  & 0.006  & 0.096  & 0.099  & 0.951  \\
          &       &       &       &       &       &       &       &       &       &       &       &       &  \\
    \multicolumn{1}{l}{Cauchy tail} & $\hat{\tau}_{\textnormal{naive-dim}}$ & 1.596  & 31.809  & 7.780  & 0.978  & -0.697  & 50.088  & 8.542  & 0.967  & 1.153  & 32.051  & /     & / \\
          & $\hat{\tau}_{\textnormal{str-dim}}$ & 1.688  & 32.813  & 7.766  & 0.980  & -0.692  & 49.925  & 8.542  & 0.967  & 1.206  & 33.546  & 7.863  & 0.983  \\
          & $\hat{\tau}_{\textnormal{md}}$ & 0.001  & 0.154  & 0.152  & 0.941  & 0.000  & 0.143  & 0.145  & 0.955  & 0.004  & 0.151  & /     & / \\
          & $\hat{\tau}_{\textnormal{wt-md}}$ & 0.003  & 0.147  & 0.145  & 0.946  & 0.000  & 0.144  & 0.145  & 0.952  & 0.003  & 0.144  & 0.145  & 0.947  \\
          & $\hat{\tau}_{\dim}$ & -0.001  & 0.138  & 0.132  & 0.945  & -0.002  & 0.128  & 0.124  & 0.940  & -0.001  & 0.132  & /     & / \\
          & $\hat{\tau}_{\str}$ & 0.001  & 0.130  & 0.124  & 0.943  & -0.001  & 0.128  & 0.124  & 0.940  & -0.001  & 0.126  & 0.124  & 0.945  \\
    \bottomrule
    \end{tabular}%
    }
    \end{threeparttable}
  \label{tab:unequalpi-model3}
  \begin{tablenotes}
				\footnotesize
				\item[] Note: $\hat{\tau}_{\textnormal{naive-dim}}$, standard difference-in-means estimator; $\hat{\tau}_{\textnormal{str-dim}}$, stratified difference-in-means esti-\\mator; $\hat{\tau}_{\textnormal{md}}$, difference-in-medians estimator; $\hat{\tau}_{\textnormal{wt-md}}$, difference-in-weighted-medians estimator; $\hat{\tau}_{\dim}$,\\ transformed difference-in-means estimator; $\hat{\tau}_{\str}$, stratified transformed difference-in-means estimator;\\ SR, simple randomization; STR, stratified randomization; MIN, Pocock and Simon's minimization;\\ SD, standard deviation; SE, standard error; CP, coverage probability.
			\end{tablenotes}
\end{table}%

\subsection{Smaller sample size}

Tables \ref{tab:smallsample-model1}-\ref{tab:smallsample-model3} present the results for a sample size of $n=500$. The other settings remain identical to those in the main text. We find that $\hat{\tau}_{\dim}$ and $\hat{\tau}_{\str}$ exhibit similar performance to the case when $n=1000$, except for greater variability,  which is attributed to the precision of the density estimation. Density estimation typically requires a considerably large sample size. In practice, when the sample size is insufficient, we can leverage domain knowledge or auxiliary data to assist in estimating the density. For example, when testing the effectiveness of a new drug relative to an existing one, even if the number of patients enrolled in the randomized trial is limited, we can still use data from patients who have been taking the existing drug to estimate the distribution of relevant physiological indicators for the control group. If additional information is unavailable, we may need to estimate the density functions separately using different hyperparameters (e.g., the bandwidth) when constructing the point estimator and the variance estimator. Specifically, for estimating the treatment effect, we adopt more aggressive hyperparameter settings (more localized bandwidth) to capture the details of each sample point; for estimating the variance, we adopt more conservative hyperparameter settings (more global bandwidth) to ensure the coverage rate of the confidence intervals.

\begin{table}[htbp]
  \centering
  \caption{Simulation results under Model 1 for $n=500$ and $\pi=1/2$}
  \begin{threeparttable}
    \resizebox{\textwidth}{!}{
    \begin{tabular}{rlrrrrrrrrrrrr}
    \toprule
          &       & \multicolumn{4}{c}{SR}        & \multicolumn{4}{c}{STR}       & \multicolumn{4}{c}{MIN} \\
\cmidrule{3-14}          &       & \multicolumn{1}{l}{Bias} & \multicolumn{1}{l}{SD} & \multicolumn{1}{l}{SE} & \multicolumn{1}{l}{CP} & \multicolumn{1}{l}{Bias} & \multicolumn{1}{l}{SD} & \multicolumn{1}{l}{SE} & \multicolumn{1}{l}{CP} & \multicolumn{1}{l}{Bias} & \multicolumn{1}{l}{SD} & \multicolumn{1}{l}{SE} & \multicolumn{1}{l}{CP} \\
    \midrule
    \multicolumn{1}{l}{$\tau$ = 0} &       &       &       &       &       &       &       &       &       &       &       &       &  \\
    \multicolumn{1}{l}{Normal tail} & $\hat{\tau}_{\textnormal{naive-dim}}$ & -0.006  & 0.119  & 0.118  & 0.953  & 0.000  & 0.099  & 0.097  & 0.942  & 0.004  & 0.109  & /     & / \\
          & $\hat{\tau}_{\textnormal{str-dim}}$ & -0.003  & 0.097  & 0.097  & 0.953  & 0.000  & 0.099  & 0.097  & 0.941  & 0.002  & 0.093  & 0.097  & 0.963  \\
          & $\hat{\tau}_{\textnormal{md}}$ & -0.004  & 0.151  & 0.165  & 0.960  & -0.001  & 0.128  & 0.148  & 0.974  & 0.007  & 0.141  & /     & / \\
          & $\hat{\tau}_{\textnormal{wt-md}}$ & 0.001  & 0.131  & 0.148  & 0.974  & -0.001  & 0.129  & 0.148  & 0.974  & 0.005  & 0.128  & 0.148  & 0.972  \\
          & $\hat{\tau}_{\dim}$ & -0.004  & 0.132  & 0.127  & 0.924  & -0.001  & 0.110  & 0.114  & 0.952  & 0.006  & 0.122  & /     & / \\
          & $\hat{\tau}_{\str}$ & -0.002  & 0.111  & 0.114  & 0.956  & -0.001  & 0.110  & 0.114  & 0.950  & 0.004  & 0.109  & 0.114  & 0.958  \\
          &       &       &       &       &       &       &       &       &       &       &       &       &  \\
    \multicolumn{1}{l}{Laplace tail} & $\hat{\tau}_{\textnormal{naive-dim}}$ & -0.003  & 0.144  & 0.148  & 0.957  & -0.002  & 0.131  & 0.132  & 0.956  & 0.002  & 0.144  & /     & / \\
          & $\hat{\tau}_{\textnormal{str-dim}}$ & 0.000  & 0.127  & 0.132  & 0.962  & -0.002  & 0.131  & 0.132  & 0.957  & 0.001  & 0.134  & 0.132  & 0.949  \\
          & $\hat{\tau}_{\textnormal{md}}$ & 0.000  & 0.159  & 0.167  & 0.954  & -0.003  & 0.147  & 0.151  & 0.950  & 0.006  & 0.160  & /     & / \\
          & $\hat{\tau}_{\textnormal{wt-md}}$ & 0.004  & 0.142  & 0.152  & 0.967  & -0.003  & 0.147  & 0.151  & 0.950  & 0.005  & 0.150  & 0.151  & 0.943  \\
          & $\hat{\tau}_{\dim}$ & -0.005  & 0.145  & 0.145  & 0.952  & -0.004  & 0.131  & 0.130  & 0.948  & 0.008  & 0.141  & /     & / \\
          & $\hat{\tau}_{\str}$ & -0.002  & 0.127  & 0.130  & 0.962  & -0.004  & 0.131  & 0.130  & 0.949  & 0.007  & 0.129  & 0.131  & 0.958  \\
          &       &       &       &       &       &       &       &       &       &       &       &       &  \\
    \multicolumn{1}{l}{Cauchy tail} & $\hat{\tau}_{\textnormal{naive-dim}}$ & -6.532  & 221.631  & 13.877  & 0.984  & 6.983  & 214.727  & 13.966  & 0.976  & -7.462  & 223.691  & /     & / \\
          & $\hat{\tau}_{\textnormal{str-dim}}$ & -5.988  & 204.184  & 13.834  & 0.980  & 6.953  & 213.873  & 13.966  & 0.976  & -7.294  & 224.395  & 14.077  & 0.976  \\
          & $\hat{\tau}_{\textnormal{md}}$ & 0.006  & 0.207  & 0.216  & 0.959  & 0.002  & 0.201  & 0.204  & 0.948  & 0.004  & 0.206  & /     & / \\
          & $\hat{\tau}_{\textnormal{wt-md}}$ & 0.011  & 0.195  & 0.204  & 0.960  & 0.002  & 0.202  & 0.204  & 0.946  & 0.003  & 0.199  & 0.204  & 0.951  \\
          & $\hat{\tau}_{\dim}$ & -0.001  & 0.192  & 0.199  & 0.954  & 0.005  & 0.178  & 0.182  & 0.950  & 0.011  & 0.186  & /     & / \\
          & $\hat{\tau}_{\str}$ & 0.003  & 0.181  & 0.185  & 0.955  & 0.005  & 0.178  & 0.182  & 0.950  & 0.009  & 0.176  & 0.183  & 0.959  \\
    \multicolumn{1}{l}{$\tau$ = 1} &       &       &       &       &       &       &       &       &       &       &       &       &  \\
    \multicolumn{1}{l}{Normal tail} & $\hat{\tau}_{\textnormal{naive-dim}}$ & 0.000  & 0.120  & 0.118  & 0.939  & -0.005  & 0.094  & 0.097  & 0.955  & 0.005  & 0.109  & /     & / \\
          & $\hat{\tau}_{\textnormal{str-dim}}$ & 0.000  & 0.099  & 0.097  & 0.943  & -0.005  & 0.094  & 0.097  & 0.955  & 0.003  & 0.093  & 0.097  & 0.965  \\
          & $\hat{\tau}_{\textnormal{md}}$ & 0.002  & 0.152  & 0.166  & 0.962  & -0.009  & 0.129  & 0.148  & 0.974  & 0.005  & 0.148  & /     & / \\
          & $\hat{\tau}_{\textnormal{wt-md}}$ & 0.003  & 0.134  & 0.149  & 0.968  & -0.009  & 0.129  & 0.148  & 0.975  & 0.005  & 0.135  & 0.149  & 0.964  \\
          & $\hat{\tau}_{\dim}$ & 0.000  & 0.133  & 0.127  & 0.935  & -0.008  & 0.105  & 0.114  & 0.966  & 0.005  & 0.121  & /     & / \\
          & $\hat{\tau}_{\str}$ & 0.000  & 0.114  & 0.114  & 0.949  & -0.008  & 0.105  & 0.114  & 0.965  & 0.004  & 0.108  & 0.114  & 0.956  \\
          &       &       &       &       &       &       &       &       &       &       &       &       &  \\
    \multicolumn{1}{l}{Laplace tail} & $\hat{\tau}_{\textnormal{naive-dim}}$ & 0.004  & 0.146  & 0.148  & 0.959  & 0.000  & 0.135  & 0.131  & 0.949  & -0.002  & 0.137  & /     & / \\
          & $\hat{\tau}_{\textnormal{str-dim}}$ & 0.005  & 0.130  & 0.131  & 0.956  & 0.000  & 0.135  & 0.131  & 0.945  & -0.004  & 0.125  & 0.131  & 0.955  \\
          & $\hat{\tau}_{\textnormal{md}}$ & 0.002  & 0.167  & 0.164  & 0.926  & -0.001  & 0.147  & 0.149  & 0.949  & 0.000  & 0.155  & /     & / \\
          & $\hat{\tau}_{\textnormal{wt-md}}$ & 0.002  & 0.149  & 0.149  & 0.937  & -0.001  & 0.147  & 0.149  & 0.948  & -0.002  & 0.141  & 0.150  & 0.954  \\
          & $\hat{\tau}_{\dim}$ & 0.001  & 0.148  & 0.145  & 0.945  & -0.001  & 0.131  & 0.130  & 0.953  & -0.005  & 0.136  & /     & / \\
          & $\hat{\tau}_{\str}$ & 0.002  & 0.132  & 0.130  & 0.938  & -0.001  & 0.131  & 0.130  & 0.950  & -0.006  & 0.123  & 0.131  & 0.959  \\
          &       &       &       &       &       &       &       &       &       &       &       &       &  \\
    \multicolumn{1}{l}{Cauchy tail} & $\hat{\tau}_{\textnormal{naive-dim}}$ & 2.273  & 71.961  & 11.065  & 0.974  & 1.112  & 73.012  & 11.102  & 0.973  & -1.240  & 72.535  & /     & / \\
          & $\hat{\tau}_{\textnormal{str-dim}}$ & 2.363  & 71.992  & 11.031  & 0.978  & 1.118  & 73.000  & 11.102  & 0.973  & -1.344  & 71.392  & 11.064  & 0.982  \\
          & $\hat{\tau}_{\textnormal{md}}$ & 0.009  & 0.210  & 0.216  & 0.955  & -0.002  & 0.187  & 0.203  & 0.970  & -0.004  & 0.198  & /     & / \\
          & $\hat{\tau}_{\textnormal{wt-md}}$ & 0.008  & 0.198  & 0.203  & 0.955  & -0.002  & 0.187  & 0.203  & 0.969  & -0.005  & 0.189  & 0.202  & 0.968  \\
          & $\hat{\tau}_{\dim}$ & 0.008  & 0.190  & 0.199  & 0.961  & -0.001  & 0.176  & 0.181  & 0.959  & -0.003  & 0.186  & /     & / \\
          & $\hat{\tau}_{\str}$ & 0.008  & 0.177  & 0.184  & 0.958  & -0.001  & 0.176  & 0.181  & 0.957  & -0.005  & 0.177  & 0.181  & 0.957  \\
    \bottomrule
    \end{tabular}%
    }
    \end{threeparttable}
  \label{tab:smallsample-model1}
  \begin{tablenotes}
				\footnotesize
				\item[] Note: $\hat{\tau}_{\textnormal{naive-dim}}$, standard difference-in-means estimator; $\hat{\tau}_{\textnormal{str-dim}}$, stratified difference-in-means esti-\\mator; $\hat{\tau}_{\textnormal{md}}$, difference-in-medians estimator; $\hat{\tau}_{\textnormal{wt-md}}$, difference-in-weighted-medians estimator; $\hat{\tau}_{\dim}$,\\ transformed difference-in-means estimator; $\hat{\tau}_{\str}$, stratified transformed difference-in-means estimator;\\ SR, simple randomization; STR, stratified randomization; MIN, Pocock and Simon's minimization;\\ SD, standard deviation; SE, standard error; CP, coverage probability.
			\end{tablenotes}
\end{table}%

\begin{table}[htbp]
  \centering
  \caption{Simulation results under Model 2 for $n=500$ and $\pi=1/2$}
  \begin{threeparttable}
    \resizebox{\textwidth}{!}{
    \begin{tabular}{rlrrrrrrrrrrrr}
    \toprule
          &       & \multicolumn{4}{c}{SR}        & \multicolumn{4}{c}{STR}       & \multicolumn{4}{c}{MIN} \\
\cmidrule{3-14}          &       & \multicolumn{1}{l}{Bias} & \multicolumn{1}{l}{SD} & \multicolumn{1}{l}{SE} & \multicolumn{1}{l}{CP} & \multicolumn{1}{l}{Bias} & \multicolumn{1}{l}{SD} & \multicolumn{1}{l}{SE} & \multicolumn{1}{l}{CP} & \multicolumn{1}{l}{Bias} & \multicolumn{1}{l}{SD} & \multicolumn{1}{l}{SE} & \multicolumn{1}{l}{CP} \\
    \midrule
    \multicolumn{1}{l}{$\tau$ = 0} &       &       &       &       &       &       &       &       &       &       &       &       &  \\
    \multicolumn{1}{l}{Normal tail} & $\hat{\tau}_{\textnormal{naive-dim}}$ & 0.000  & 0.118  & 0.117  & 0.954  & 0.004  & 0.099  & 0.099  & 0.951  & 0.001  & 0.111  & /     & / \\
          & $\hat{\tau}_{\textnormal{str-dim}}$ & 0.000  & 0.100  & 0.099  & 0.950  & 0.004  & 0.099  & 0.099  & 0.951  & 0.001  & 0.098  & 0.099  & 0.949  \\
          & $\hat{\tau}_{\textnormal{md}}$ & -0.003  & 0.146  & 0.152  & 0.953  & 0.010  & 0.131  & 0.139  & 0.960  & -0.001  & 0.137  & /     & / \\
          & $\hat{\tau}_{\textnormal{wt-md}}$ & -0.004  & 0.132  & 0.139  & 0.961  & 0.009  & 0.131  & 0.139  & 0.959  & -0.001  & 0.129  & 0.139  & 0.963  \\
          & $\hat{\tau}_{\dim}$ & -0.001  & 0.125  & 0.118  & 0.931  & 0.010  & 0.112  & 0.110  & 0.942  & 0.001  & 0.117  & /     & / \\
          & $\hat{\tau}_{\str}$ & -0.001  & 0.112  & 0.110  & 0.946  & 0.011  & 0.112  & 0.110  & 0.942  & 0.000  & 0.110  & 0.110  & 0.938  \\
          &       &       &       &       &       &       &       &       &       &       &       &       &  \\
    \multicolumn{1}{l}{Laplace tail} & $\hat{\tau}_{\textnormal{naive-dim}}$ & 0.002  & 0.146  & 0.147  & 0.947  & 0.002  & 0.133  & 0.133  & 0.948  & 0.004  & 0.144  & /     & / \\
          & $\hat{\tau}_{\textnormal{str-dim}}$ & 0.002  & 0.132  & 0.133  & 0.950  & 0.002  & 0.133  & 0.133  & 0.949  & 0.004  & 0.132  & 0.133  & 0.947  \\
          & $\hat{\tau}_{\textnormal{md}}$ & -0.002  & 0.153  & 0.146  & 0.930  & 0.000  & 0.139  & 0.135  & 0.934  & 0.006  & 0.147  & /     & / \\
          & $\hat{\tau}_{\textnormal{wt-md}}$ & -0.002  & 0.138  & 0.135  & 0.943  & 0.000  & 0.140  & 0.135  & 0.932  & 0.005  & 0.137  & 0.135  & 0.932  \\
          & $\hat{\tau}_{\dim}$ & 0.001  & 0.138  & 0.136  & 0.945  & 0.001  & 0.125  & 0.126  & 0.949  & 0.009  & 0.133  & /     & / \\
          & $\hat{\tau}_{\str}$ & 0.001  & 0.125  & 0.127  & 0.953  & 0.001  & 0.125  & 0.126  & 0.947  & 0.008  & 0.125  & 0.127  & 0.958  \\
          &       &       &       &       &       &       &       &       &       &       &       &       &  \\
    \multicolumn{1}{l}{Cauchy tail} & $\hat{\tau}_{\textnormal{naive-dim}}$ & 0.133  & 95.123  & 13.143  & 0.980  & -3.813  & 96.727  & 13.181  & 0.975  & -5.432  & 97.120  & /     & / \\
          & $\hat{\tau}_{\textnormal{str-dim}}$ & 0.069  & 91.833  & 13.105  & 0.979  & -3.789  & 96.562  & 13.181  & 0.975  & -5.710  & 99.134  & 13.248  & 0.984  \\
          & $\hat{\tau}_{\textnormal{md}}$ & 0.000  & 0.194  & 0.205  & 0.954  & -0.011  & 0.185  & 0.196  & 0.958  & -0.002  & 0.189  & /     & / \\
          & $\hat{\tau}_{\textnormal{wt-md}}$ & 0.000  & 0.185  & 0.195  & 0.958  & -0.011  & 0.185  & 0.196  & 0.958  & -0.002  & 0.181  & 0.196  & 0.966  \\
          & $\hat{\tau}_{\dim}$ & 0.000  & 0.189  & 0.194  & 0.954  & -0.012  & 0.178  & 0.179  & 0.948  & -0.005  & 0.177  & /     & / \\
          & $\hat{\tau}_{\str}$ & -0.001  & 0.178  & 0.183  & 0.954  & -0.012  & 0.178  & 0.179  & 0.947  & -0.006  & 0.169  & 0.180  & 0.966  \\
    \multicolumn{1}{l}{$\tau$ = 1} &       &       &       &       &       &       &       &       &       &       &       &       &  \\
    \multicolumn{1}{l}{Normal tail} & $\hat{\tau}_{\textnormal{naive-dim}}$ & 0.005  & 0.113  & 0.117  & 0.956  & 0.009  & 0.102  & 0.099  & 0.942  & -0.003  & 0.109  & /     & / \\
          & $\hat{\tau}_{\textnormal{str-dim}}$ & 0.007  & 0.098  & 0.099  & 0.961  & 0.009  & 0.102  & 0.099  & 0.944  & -0.005  & 0.094  & 0.099  & 0.957  \\
          & $\hat{\tau}_{\textnormal{md}}$ & 0.003  & 0.135  & 0.152  & 0.968  & 0.008  & 0.129  & 0.139  & 0.965  & 0.003  & 0.133  & /     & / \\
          & $\hat{\tau}_{\textnormal{wt-md}}$ & 0.005  & 0.124  & 0.139  & 0.972  & 0.008  & 0.130  & 0.139  & 0.966  & 0.002  & 0.122  & 0.140  & 0.971  \\
          & $\hat{\tau}_{\dim}$ & 0.007  & 0.120  & 0.118  & 0.946  & 0.010  & 0.113  & 0.110  & 0.948  & 0.000  & 0.116  & /     & / \\
          & $\hat{\tau}_{\str}$ & 0.007  & 0.108  & 0.110  & 0.958  & 0.011  & 0.113  & 0.110  & 0.947  & -0.002  & 0.107  & 0.110  & 0.950  \\
          &       &       &       &       &       &       &       &       &       &       &       &       &  \\
    \multicolumn{1}{l}{Laplace tail} & $\hat{\tau}_{\textnormal{naive-dim}}$ & 0.000  & 0.148  & 0.147  & 0.938  & -0.003  & 0.133  & 0.133  & 0.955  & 0.002  & 0.137  & /     & / \\
          & $\hat{\tau}_{\textnormal{str-dim}}$ & 0.001  & 0.134  & 0.133  & 0.947  & -0.003  & 0.133  & 0.133  & 0.954  & 0.001  & 0.128  & 0.133  & 0.957  \\
          & $\hat{\tau}_{\textnormal{md}}$ & -0.007  & 0.157  & 0.146  & 0.922  & -0.001  & 0.146  & 0.135  & 0.915  & 0.000  & 0.142  & /     & / \\
          & $\hat{\tau}_{\textnormal{wt-md}}$ & -0.005  & 0.144  & 0.135  & 0.932  & -0.001  & 0.146  & 0.135  & 0.917  & -0.001  & 0.134  & 0.135  & 0.947  \\
          & $\hat{\tau}_{\dim}$ & -0.003  & 0.142  & 0.137  & 0.933  & 0.001  & 0.129  & 0.127  & 0.947  & 0.004  & 0.126  & /     & / \\
          & $\hat{\tau}_{\str}$ & -0.002  & 0.129  & 0.127  & 0.934  & 0.001  & 0.129  & 0.127  & 0.949  & 0.003  & 0.120  & 0.127  & 0.962  \\
          &       &       &       &       &       &       &       &       &       &       &       &       &  \\
    \multicolumn{1}{l}{Cauchy tail} & $\hat{\tau}_{\textnormal{naive-dim}}$ & -0.288  & 21.346  & 6.814  & 0.968  & -1.178  & 20.716  & 6.827  & 0.976  & -0.974  & 20.829  & /     & / \\
          & $\hat{\tau}_{\textnormal{str-dim}}$ & -0.307  & 21.053  & 6.793  & 0.966  & -1.177  & 20.676  & 6.827  & 0.976  & -0.971  & 20.541  & 6.805  & 0.972  \\
          & $\hat{\tau}_{\textnormal{md}}$ & -0.015  & 0.189  & 0.207  & 0.966  & 0.006  & 0.187  & 0.195  & 0.958  & 0.005  & 0.187  & /     & / \\
          & $\hat{\tau}_{\textnormal{wt-md}}$ & -0.012  & 0.181  & 0.196  & 0.966  & 0.006  & 0.187  & 0.195  & 0.960  & 0.004  & 0.179  & 0.196  & 0.962  \\
          & $\hat{\tau}_{\dim}$ & -0.013  & 0.185  & 0.195  & 0.963  & 0.003  & 0.172  & 0.179  & 0.955  & -0.003  & 0.181  & /     & / \\
          & $\hat{\tau}_{\str}$ & -0.011  & 0.177  & 0.184  & 0.959  & 0.004  & 0.172  & 0.179  & 0.957  & -0.005  & 0.174  & 0.181  & 0.953  \\
    \bottomrule
    \end{tabular}%
    }
    \end{threeparttable}
  \label{tab:smallsample-model2}
  \begin{tablenotes}
				\footnotesize
				\item[] Note: $\hat{\tau}_{\textnormal{naive-dim}}$, standard difference-in-means estimator; $\hat{\tau}_{\textnormal{str-dim}}$, stratified difference-in-means esti-\\mator; $\hat{\tau}_{\textnormal{md}}$, difference-in-medians estimator; $\hat{\tau}_{\textnormal{wt-md}}$, difference-in-weighted-medians estimator; $\hat{\tau}_{\dim}$,\\ transformed difference-in-means estimator; $\hat{\tau}_{\str}$, stratified transformed difference-in-means estimator;\\ SR, simple randomization; STR, stratified randomization; MIN, Pocock and Simon's minimization;\\ SD, standard deviation; SE, standard error; CP, coverage probability.
			\end{tablenotes}
\end{table}%

\begin{table}[htbp]
  \centering
  \caption{Simulation results under Model 3 for $n=500$ and $\pi=1/2$}
  \begin{threeparttable}
    \resizebox{\textwidth}{!}{
    \begin{tabular}{rlrrrrrrrrrrrr}
    \toprule
          &       & \multicolumn{4}{c}{SR}        & \multicolumn{4}{c}{STR}       & \multicolumn{4}{c}{MIN} \\
\cmidrule{3-14}          &       & \multicolumn{1}{l}{Bias} & \multicolumn{1}{l}{SD} & \multicolumn{1}{l}{SE} & \multicolumn{1}{l}{CP} & \multicolumn{1}{l}{Bias} & \multicolumn{1}{l}{SD} & \multicolumn{1}{l}{SE} & \multicolumn{1}{l}{CP} & \multicolumn{1}{l}{Bias} & \multicolumn{1}{l}{SD} & \multicolumn{1}{l}{SE} & \multicolumn{1}{l}{CP} \\
    \midrule
    \multicolumn{1}{l}{$\tau$ = 0} &       &       &       &       &       &       &       &       &       &       &       &       &  \\
    \multicolumn{1}{l}{Normal tail} & $\hat{\tau}_{\textnormal{naive-dim}}$ & 0.000  & 0.126  & 0.120  & 0.940  & -0.006  & 0.100  & 0.102  & 0.957  & 0.003  & 0.115  & /     & / \\
          & $\hat{\tau}_{\textnormal{str-dim}}$ & -0.001  & 0.105  & 0.102  & 0.946  & -0.006  & 0.100  & 0.102  & 0.957  & 0.000  & 0.101  & 0.102  & 0.956  \\
          & $\hat{\tau}_{\textnormal{md}}$ & -0.002  & 0.154  & 0.161  & 0.955  & -0.005  & 0.135  & 0.145  & 0.964  & 0.002  & 0.144  & /     & / \\
          & $\hat{\tau}_{\textnormal{wt-md}}$ & -0.004  & 0.136  & 0.147  & 0.964  & -0.005  & 0.136  & 0.145  & 0.964  & 0.000  & 0.134  & 0.146  & 0.965  \\
          & $\hat{\tau}_{\dim}$ & 0.002  & 0.137  & 0.125  & 0.925  & -0.002  & 0.111  & 0.114  & 0.951  & 0.006  & 0.121  & /     & / \\
          & $\hat{\tau}_{\str}$ & 0.001  & 0.119  & 0.114  & 0.939  & -0.002  & 0.111  & 0.114  & 0.950  & 0.004  & 0.112  & 0.114  & 0.956  \\
          &       &       &       &       &       &       &       &       &       &       &       &       &  \\
    \multicolumn{1}{l}{Laplace tail} & $\hat{\tau}_{\textnormal{naive-dim}}$ & -0.006  & 0.148  & 0.149  & 0.942  & 0.001  & 0.136  & 0.135  & 0.942  & 0.006  & 0.141  & /     & / \\
          & $\hat{\tau}_{\textnormal{str-dim}}$ & -0.007  & 0.133  & 0.135  & 0.942  & 0.001  & 0.136  & 0.135  & 0.944  & 0.003  & 0.129  & 0.135  & 0.967  \\
          & $\hat{\tau}_{\textnormal{md}}$ & 0.002  & 0.160  & 0.156  & 0.941  & 0.002  & 0.140  & 0.143  & 0.945  & 0.000  & 0.148  & /     & / \\
          & $\hat{\tau}_{\textnormal{wt-md}}$ & 0.002  & 0.144  & 0.144  & 0.954  & 0.002  & 0.140  & 0.143  & 0.945  & -0.002  & 0.137  & 0.144  & 0.953  \\
          & $\hat{\tau}_{\dim}$ & 0.000  & 0.145  & 0.143  & 0.953  & -0.001  & 0.127  & 0.131  & 0.945  & 0.006  & 0.132  & /     & / \\
          & $\hat{\tau}_{\str}$ & -0.001  & 0.129  & 0.131  & 0.954  & -0.001  & 0.127  & 0.131  & 0.945  & 0.004  & 0.121  & 0.131  & 0.965  \\
          &       &       &       &       &       &       &       &       &       &       &       &       &  \\
    \multicolumn{1}{l}{Cauchy tail} & $\hat{\tau}_{\textnormal{naive-dim}}$ & 84.118  & 2821.234  & 109.258  & 0.980  & -98.881  & 2988.736  & 106.743  & 0.987  & -101.940  & 3099.181  & /     & / \\
          & $\hat{\tau}_{\textnormal{str-dim}}$ & 93.859  & 3130.406  & 108.938  & 0.978  & -98.910  & 2990.158  & 106.743  & 0.987  & -93.180  & 2821.739  & 104.067  & 0.978  \\
          & $\hat{\tau}_{\textnormal{md}}$ & 0.006  & 0.196  & 0.213  & 0.975  & 0.001  & 0.182  & 0.200  & 0.974  & 0.009  & 0.194  & /     & / \\
          & $\hat{\tau}_{\textnormal{wt-md}}$ & 0.004  & 0.181  & 0.202  & 0.972  & 0.001  & 0.182  & 0.200  & 0.972  & 0.006  & 0.184  & 0.201  & 0.967  \\
          & $\hat{\tau}_{\dim}$ & 0.005  & 0.188  & 0.198  & 0.963  & -0.001  & 0.177  & 0.181  & 0.958  & 0.001  & 0.186  & /     & / \\
          & $\hat{\tau}_{\str}$ & 0.003  & 0.174  & 0.186  & 0.967  & -0.001  & 0.177  & 0.181  & 0.960  & -0.001  & 0.175  & 0.182  & 0.963  \\
    \multicolumn{1}{l}{$\tau$ = 1} &       &       &       &       &       &       &       &       &       &       &       &       &  \\
    \multicolumn{1}{l}{Normal tail} & $\hat{\tau}_{\textnormal{naive-dim}}$ & 0.000  & 0.123  & 0.120  & 0.944  & 0.002  & 0.103  & 0.102  & 0.946  & -0.004  & 0.109  & /     & / \\
          & $\hat{\tau}_{\textnormal{str-dim}}$ & -0.002  & 0.104  & 0.102  & 0.949  & 0.002  & 0.103  & 0.102  & 0.945  & -0.004  & 0.095  & 0.102  & 0.964  \\
          & $\hat{\tau}_{\textnormal{md}}$ & 0.003  & 0.151  & 0.160  & 0.964  & 0.001  & 0.134  & 0.146  & 0.961  & 0.000  & 0.134  & /     & / \\
          & $\hat{\tau}_{\textnormal{wt-md}}$ & 0.001  & 0.137  & 0.146  & 0.961  & 0.001  & 0.134  & 0.146  & 0.961  & 0.001  & 0.126  & 0.146  & 0.971  \\
          & $\hat{\tau}_{\dim}$ & 0.004  & 0.131  & 0.125  & 0.941  & 0.003  & 0.115  & 0.114  & 0.943  & -0.002  & 0.117  & /     & / \\
          & $\hat{\tau}_{\str}$ & 0.003  & 0.117  & 0.114  & 0.939  & 0.003  & 0.115  & 0.114  & 0.943  & -0.002  & 0.107  & 0.114  & 0.962  \\
          &       &       &       &       &       &       &       &       &       &       &       &       &  \\
    \multicolumn{1}{l}{Laplace tail} & $\hat{\tau}_{\textnormal{naive-dim}}$ & 0.004  & 0.155  & 0.149  & 0.944  & 0.000  & 0.133  & 0.135  & 0.956  & -0.004  & 0.141  & /     & / \\
          & $\hat{\tau}_{\textnormal{str-dim}}$ & 0.002  & 0.140  & 0.135  & 0.950  & 0.000  & 0.133  & 0.135  & 0.957  & -0.005  & 0.128  & 0.135  & 0.957  \\
          & $\hat{\tau}_{\textnormal{md}}$ & 0.005  & 0.162  & 0.157  & 0.938  & 0.002  & 0.145  & 0.144  & 0.945  & -0.002  & 0.156  & /     & / \\
          & $\hat{\tau}_{\textnormal{wt-md}}$ & 0.003  & 0.149  & 0.145  & 0.940  & 0.002  & 0.145  & 0.144  & 0.945  & -0.003  & 0.145  & 0.144  & 0.947  \\
          & $\hat{\tau}_{\dim}$ & 0.004  & 0.145  & 0.144  & 0.948  & 0.000  & 0.129  & 0.131  & 0.946  & 0.003  & 0.139  & /     & / \\
          & $\hat{\tau}_{\str}$ & 0.003  & 0.131  & 0.132  & 0.950  & 0.001  & 0.129  & 0.131  & 0.946  & 0.002  & 0.128  & 0.131  & 0.948  \\
          &       &       &       &       &       &       &       &       &       &       &       &       &  \\
    \multicolumn{1}{l}{Cauchy tail} & $\hat{\tau}_{\textnormal{naive-dim}}$ & 1.387  & 83.917  & 9.725  & 0.976  & 3.544  & 87.378  & 9.797  & 0.978  & 3.096  & 85.248  & /     & / \\
          & $\hat{\tau}_{\textnormal{str-dim}}$ & 1.518  & 84.781  & 9.694  & 0.972  & 3.552  & 87.676  & 9.797  & 0.980  & 3.205  & 84.577  & 9.710  & 0.980  \\
          & $\hat{\tau}_{\textnormal{md}}$ & 0.004  & 0.207  & 0.211  & 0.957  & 0.002  & 0.190  & 0.200  & 0.962  & -0.014  & 0.198  & /     & / \\
          & $\hat{\tau}_{\textnormal{wt-md}}$ & 0.003  & 0.195  & 0.201  & 0.951  & 0.002  & 0.190  & 0.200  & 0.962  & -0.016  & 0.189  & 0.202  & 0.962  \\
          & $\hat{\tau}_{\dim}$ & 0.003  & 0.195  & 0.196  & 0.942  & 0.006  & 0.180  & 0.181  & 0.949  & -0.010  & 0.187  & /     & / \\
          & $\hat{\tau}_{\str}$ & 0.003  & 0.183  & 0.184  & 0.944  & 0.007  & 0.180  & 0.181  & 0.950  & -0.011  & 0.179  & 0.183  & 0.945  \\
    \bottomrule
    \end{tabular}%
    }
    \end{threeparttable}
  \label{tab:smallsample-model3}
  \begin{tablenotes}
				\footnotesize
				\item[] Note: $\hat{\tau}_{\textnormal{naive-dim}}$, standard difference-in-means estimator; $\hat{\tau}_{\textnormal{str-dim}}$, stratified difference-in-means esti-\\mator; $\hat{\tau}_{\textnormal{md}}$, difference-in-medians estimator; $\hat{\tau}_{\textnormal{wt-md}}$, difference-in-weighted-medians estimator; $\hat{\tau}_{\dim}$,\\ transformed difference-in-means estimator; $\hat{\tau}_{\str}$, stratified transformed difference-in-means estimator;\\ SR, simple randomization; STR, stratified randomization; MIN, Pocock and Simon's minimization;\\ SD, standard deviation; SE, standard error; CP, coverage probability.
			\end{tablenotes}
\end{table}%

\newpage

\subsection{Initial estimators}

In this section, we evaluate the influence of the initial estimator on the final treatment effect estimators, adopting the same settings as in the main text. Tables \ref{tab:changeinitial-model1}-\ref{tab:changeinitial-model3} present the results for the weighted average difference-in-medians estimator $\hat{\tau}_{\textnormal{str-md}}$, the difference-in-weighted-medians estimator $\hat{\tau}_{\textnormal{wt-md}}$, and the stratified transformed difference-in-means estimator $\hat{\tau}_{\str}$ when using each of the first two as the initial estimator. To distinguish them, we denote $\hat{\tau}_{\str}$ as $\hat{\tau}_{\str,1}$ and $\hat{\tau}_{\str,2}$ when using $\hat{\tau}_{\textnormal{str-md}}$ and $\hat{\tau}_{\textnormal{wt-md}}$ as initial estimators, respectively. To the best of our knowledge, the asymptotic properties of  $\hat{\tau}_{\textnormal{str-md}}$ under covariate-adaptive randomization have not yet been studied, so we do not report inference-related results for it. We find that $\hat{\tau}_{\str}$ exhibits similar performance when using these two different initial estimators, with the differences being smaller than the Monte Carlo error. Thus, $\hat{\tau}_{\str}$ appears to be robust to the choice of the initial estimator. For the two initial estimators themselves, $\hat{\tau}_{\textnormal{str-md}}$ demonstrates comparable or smaller variance compared to $\hat{\tau}_{\textnormal{wt-md}}$ in all cases. Notably, under a Cauchy tail distribution, $\hat{\tau}_{\textnormal{str-md}}$ has a variance close to that of $\hat{\tau}_{\str}$. Further investigation reveals that the error in density estimation contributes significantly to the variance of $\hat{\tau}_{\str}$. When the sample size used for density estimation exceeds, for example, 350, $\hat{\tau}_{\str}$ begins to exhibit significantly smaller variance.

\begin{table}[htbp]
  \centering
  \caption{Simulation results under Model 1 for different initial estimators}
  \begin{threeparttable}
    \resizebox{\textwidth}{!}{
    \begin{tabular}{rlrrrrrrrrrrrr}
    \toprule
          &       & \multicolumn{4}{c}{SR}        & \multicolumn{4}{c}{STR}       & \multicolumn{4}{c}{MIN} \\
\cmidrule{3-14}          &       & \multicolumn{1}{l}{Bias} & \multicolumn{1}{l}{SD} & \multicolumn{1}{l}{SE} & \multicolumn{1}{l}{CP} & \multicolumn{1}{l}{Bias} & \multicolumn{1}{l}{SD} & \multicolumn{1}{l}{SE} & \multicolumn{1}{l}{CP} & \multicolumn{1}{l}{Bias} & \multicolumn{1}{l}{SD} & \multicolumn{1}{l}{SE} & \multicolumn{1}{l}{CP} \\
    \midrule
    \multicolumn{1}{l}{$\tau$ = 0} &       &       &       &       &       &       &       &       &       &       &       &       &  \\
    \multicolumn{1}{l}{Normal tail} & $\hat{\tau}_{\textnormal{str-md}}$ & 0.005  & 0.090  & /     & /     & -0.003  & 0.090  & /     & /     & -0.004  & 0.086  & /     & / \\
          & $\hat{\tau}_{\textnormal{wt-md}}$ & 0.005  & 0.095  & 0.106  & 0.961  & -0.004  & 0.094  & 0.106  & 0.962  & -0.003  & 0.094  & 0.106  & 0.960  \\
          & $\hat{\tau}_{\str,1}$ & 0.005  & 0.079  & 0.082  & 0.950  & -0.005  & 0.078  & 0.081  & 0.964  & -0.003  & 0.080  & 0.082  & 0.948  \\
          & $\hat{\tau}_{\str,2}$ & 0.005  & 0.079  & 0.082  & 0.945  & -0.005  & 0.078  & 0.081  & 0.965  & -0.003  & 0.080  & 0.082  & 0.949  \\
          &       &       &       &       &       &       &       &       &       &       &       &       &  \\
    \multicolumn{1}{l}{Laplace tail} & $\hat{\tau}_{\textnormal{str-md}}$ & 0.003  & 0.092  & /     & /     & -0.005  & 0.093  & /     & /     & 0.001  & 0.090  & /     & / \\
          & $\hat{\tau}_{\textnormal{wt-md}}$ & 0.007  & 0.107  & 0.109  & 0.948  & -0.003  & 0.107  & 0.109  & 0.959  & 0.000  & 0.108  & 0.109  & 0.945  \\
          & $\hat{\tau}_{\str,1}$ & 0.006  & 0.096  & 0.092  & 0.932  & -0.005  & 0.093  & 0.092  & 0.942  & 0.001  & 0.093  & 0.092  & 0.954  \\
          & $\hat{\tau}_{\str,2}$ & 0.006  & 0.095  & 0.092  & 0.942  & -0.006  & 0.092  & 0.092  & 0.943  & 0.001  & 0.091  & 0.092  & 0.958  \\
          &       &       &       &       &       &       &       &       &       &       &       &       &  \\
    \multicolumn{1}{l}{Cauchy tail} & $\hat{\tau}_{\textnormal{str-md}}$ & -0.001  & 0.113  & /     & /     & 0.007  & 0.121  & /     & /     & 0.005  & 0.112  & /     & / \\
          & $\hat{\tau}_{\textnormal{wt-md}}$ & -0.001  & 0.138  & 0.138  & 0.942  & 0.005  & 0.135  & 0.138  & 0.946  & 0.004  & 0.130  & 0.138  & 0.965  \\
          & $\hat{\tau}_{\str,1}$ & 0.002  & 0.118  & 0.118  & 0.948  & 0.003  & 0.118  & 0.118  & 0.943  & 0.006  & 0.116  & 0.118  & 0.958  \\
          & $\hat{\tau}_{\str,2}$ & 0.001  & 0.118  & 0.118  & 0.949  & 0.003  & 0.118  & 0.118  & 0.945  & 0.005  & 0.116  & 0.118  & 0.956  \\
    \multicolumn{1}{l}{$\tau$ = 1} &       &       &       &       &       &       &       &       &       &       &       &       &  \\
    \multicolumn{1}{l}{Normal tail} & $\hat{\tau}_{\textnormal{str-md}}$ & 0.003  & 0.087  & /     & /     & 0.009  & 0.087  & /     & /     & -0.003  & 0.083  & /     & / \\
          & $\hat{\tau}_{\textnormal{wt-md}}$ & 0.001  & 0.095  & 0.106  & 0.968  & 0.006  & 0.092  & 0.106  & 0.977  & 0.001  & 0.091  & 0.106  & 0.974  \\
          & $\hat{\tau}_{\str,1}$ & 0.003  & 0.077  & 0.082  & 0.963  & 0.006  & 0.075  & 0.081  & 0.964  & -0.002  & 0.075  & 0.082  & 0.970  \\
          & $\hat{\tau}_{\str,2}$ & 0.003  & 0.078  & 0.082  & 0.968  & 0.007  & 0.075  & 0.081  & 0.960  & -0.002  & 0.075  & 0.082  & 0.967  \\
          &       &       &       &       &       &       &       &       &       &       &       &       &  \\
    \multicolumn{1}{l}{Laplace tail} & $\hat{\tau}_{\textnormal{str-md}}$ & 0.000  & 0.090  & /     & /     & 0.004  & 0.089  & /     & /     & 0.003  & 0.083  & /     & / \\
          & $\hat{\tau}_{\textnormal{wt-md}}$ & 0.003  & 0.100  & 0.109  & 0.972  & 0.004  & 0.108  & 0.109  & 0.946  & -0.002  & 0.100  & 0.109  & 0.965  \\
          & $\hat{\tau}_{\str,1}$ & 0.002  & 0.092  & 0.092  & 0.941  & 0.004  & 0.094  & 0.092  & 0.946  & -0.001  & 0.089  & 0.092  & 0.962  \\
          & $\hat{\tau}_{\str,2}$ & 0.002  & 0.091  & 0.092  & 0.944  & 0.004  & 0.092  & 0.092  & 0.945  & -0.001  & 0.088  & 0.092  & 0.965  \\
          &       &       &       &       &       &       &       &       &       &       &       &       &  \\
    \multicolumn{1}{l}{Cauchy tail} & $\hat{\tau}_{\textnormal{str-md}}$ & 0.000  & 0.117  & /     & /     & -0.008  & 0.119  & /     & /     & -0.006  & 0.114  & /     & / \\
          & $\hat{\tau}_{\textnormal{wt-md}}$ & 0.002  & 0.138  & 0.138  & 0.948  & -0.004  & 0.143  & 0.138  & 0.938  & -0.007  & 0.138  & 0.138  & 0.948  \\
          & $\hat{\tau}_{\str,1}$ & -0.002  & 0.119  & 0.118  & 0.944  & -0.007  & 0.121  & 0.118  & 0.934  & -0.006  & 0.119  & 0.118  & 0.944  \\
          & $\hat{\tau}_{\str,2}$ & -0.002  & 0.119  & 0.118  & 0.944  & -0.007  & 0.122  & 0.118  & 0.931  & -0.006  & 0.120  & 0.118  & 0.942  \\
    \bottomrule
    \end{tabular}%
    }
    \end{threeparttable}
  \label{tab:changeinitial-model1}
  \begin{tablenotes}
				\footnotesize
				\item[] Note: $\hat{\tau}_{\textnormal{str-md}}$, stratified difference-in-medians estimator; $\hat{\tau}_{\textnormal{wt-md}}$, difference-in-weighted-medians esti-\\mator; $\hat{\tau}_{\str,1}$, stratified transformed difference-in-means estimator with $\hat{\tau}_{\textnormal{str-md}}$ as initial estimator;\\ $\hat{\tau}_{\str,2}$, stratified transformed difference-in-means estimator with $\hat{\tau}_{\textnormal{wt-md}}$ as initial estimator; SR, simple\\ randomization; STR, stratified randomization; MIN, Pocock and Simon's minimization; SD, standard\\ deviation; SE, standard error; CP, coverage probability.
			\end{tablenotes}
\end{table}%

\begin{table}[htbp]
  \centering
  \caption{Simulation results under Model 2 for different initial estimators}
  \begin{threeparttable}
    \resizebox{\textwidth}{!}{
    \begin{tabular}{rlrrrrrrrrrrrr}
    \toprule
          &       & \multicolumn{4}{c}{SR}        & \multicolumn{4}{c}{STR}       & \multicolumn{4}{c}{MIN} \\
\cmidrule{3-14}          &       & \multicolumn{1}{l}{Bias} & \multicolumn{1}{l}{SD} & \multicolumn{1}{l}{SE} & \multicolumn{1}{l}{CP} & \multicolumn{1}{l}{Bias} & \multicolumn{1}{l}{SD} & \multicolumn{1}{l}{SE} & \multicolumn{1}{l}{CP} & \multicolumn{1}{l}{Bias} & \multicolumn{1}{l}{SD} & \multicolumn{1}{l}{SE} & \multicolumn{1}{l}{CP} \\
    \midrule
    \multicolumn{1}{l}{$\tau$ = 0} &       &       &       &       &       &       &       &       &       &       &       &       &  \\
    \multicolumn{1}{l}{Normal tail} & $\hat{\tau}_{\textnormal{str-md}}$ & 0.000  & 0.090  & /     & /     & -0.003  & 0.090  & /     & /     & 0.001  & 0.087  & /     & / \\
          & $\hat{\tau}_{\textnormal{wt-md}}$ & 0.004  & 0.094  & 0.100  & 0.965  & -0.002  & 0.090  & 0.100  & 0.969  & -0.003  & 0.088  & 0.100  & 0.961  \\
          & $\hat{\tau}_{\str,1}$ & 0.001  & 0.079  & 0.079  & 0.944  & -0.003  & 0.076  & 0.078  & 0.951  & -0.001  & 0.078  & 0.079  & 0.948  \\
          & $\hat{\tau}_{\str,2}$ & 0.001  & 0.079  & 0.079  & 0.943  & -0.003  & 0.077  & 0.078  & 0.954  & 0.000  & 0.078  & 0.079  & 0.945  \\
          &       &       &       &       &       &       &       &       &       &       &       &       &  \\
    \multicolumn{1}{l}{Laplace tail} & $\hat{\tau}_{\textnormal{str-md}}$ & 0.000  & 0.090  & /     & /     & 0.002  & 0.092  & /     & /     & 0.005  & 0.088  & /     & / \\
          & $\hat{\tau}_{\textnormal{wt-md}}$ & 0.003  & 0.097  & 0.098  & 0.951  & 0.000  & 0.101  & 0.098  & 0.939  & 0.002  & 0.096  & 0.098  & 0.944  \\
          & $\hat{\tau}_{\str,1}$ & 0.004  & 0.091  & 0.090  & 0.949  & 0.002  & 0.093  & 0.090  & 0.947  & 0.003  & 0.089  & 0.090  & 0.946  \\
          & $\hat{\tau}_{\str,2}$ & 0.004  & 0.090  & 0.090  & 0.951  & 0.001  & 0.092  & 0.090  & 0.949  & 0.002  & 0.088  & 0.090  & 0.951  \\
          &       &       &       &       &       &       &       &       &       &       &       &       &  \\
    \multicolumn{1}{l}{Cauchy tail} & $\hat{\tau}_{\textnormal{str-md}}$ & 0.005  & 0.118  & /     & /     & -0.003  & 0.117  & /     & /     & 0.001  & 0.119  & /     & / \\
          & $\hat{\tau}_{\textnormal{wt-md}}$ & 0.003  & 0.129  & 0.131  & 0.946  & 0.000  & 0.127  & 0.132  & 0.956  & -0.001  & 0.130  & 0.132  & 0.954  \\
          & $\hat{\tau}_{\str,1}$ & 0.005  & 0.116  & 0.117  & 0.946  & 0.001  & 0.116  & 0.118  & 0.951  & 0.001  & 0.118  & 0.118  & 0.951  \\
          & $\hat{\tau}_{\str,2}$ & 0.005  & 0.116  & 0.117  & 0.946  & 0.000  & 0.115  & 0.118  & 0.955  & 0.000  & 0.118  & 0.118  & 0.949  \\
    \multicolumn{1}{l}{$\tau$ = 1} &       &       &       &       &       &       &       &       &       &       &       &       &  \\
    \multicolumn{1}{l}{Normal tail} & $\hat{\tau}_{\textnormal{str-md}}$ & -0.005  & 0.091  & /     & /     & 0.001  & 0.090  & /     & /     & 0.003  & 0.088  & /     & / \\
          & $\hat{\tau}_{\textnormal{wt-md}}$ & -0.003  & 0.091  & 0.101  & 0.963  & -0.004  & 0.092  & 0.101  & 0.969  & 0.003  & 0.089  & 0.101  & 0.973  \\
          & $\hat{\tau}_{\str,1}$ & -0.002  & 0.079  & 0.079  & 0.947  & -0.001  & 0.079  & 0.079  & 0.931  & 0.001  & 0.078  & 0.079  & 0.955  \\
          & $\hat{\tau}_{\str,2}$ & -0.002  & 0.079  & 0.079  & 0.945  & -0.001  & 0.079  & 0.079  & 0.929  & 0.002  & 0.078  & 0.079  & 0.957  \\
          &       &       &       &       &       &       &       &       &       &       &       &       &  \\
    \multicolumn{1}{l}{Laplace tail} & $\hat{\tau}_{\textnormal{str-md}}$ & -0.002  & 0.092  & /     & /     & -0.003  & 0.093  & /     & /     & -0.001  & 0.087  & /     & / \\
          & $\hat{\tau}_{\textnormal{wt-md}}$ & 0.000  & 0.098  & 0.098  & 0.938  & -0.006  & 0.099  & 0.098  & 0.945  & 0.000  & 0.099  & 0.098  & 0.936  \\
          & $\hat{\tau}_{\str,1}$ & 0.000  & 0.091  & 0.089  & 0.949  & -0.003  & 0.090  & 0.089  & 0.955  & 0.000  & 0.090  & 0.089  & 0.945  \\
          & $\hat{\tau}_{\str,2}$ & 0.000  & 0.090  & 0.089  & 0.947  & -0.003  & 0.089  & 0.089  & 0.955  & 0.000  & 0.090  & 0.089  & 0.953  \\
          &       &       &       &       &       &       &       &       &       &       &       &       &  \\
    \multicolumn{1}{l}{Cauchy tail} & $\hat{\tau}_{\textnormal{str-md}}$ & 0.001  & 0.120  & /     & /     & 0.002  & 0.114  & /     & /     & -0.004  & 0.121  & /     & / \\
          & $\hat{\tau}_{\textnormal{wt-md}}$ & 0.004  & 0.133  & 0.131  & 0.935  & 0.002  & 0.127  & 0.131  & 0.953  & -0.001  & 0.131  & 0.131  & 0.952  \\
          & $\hat{\tau}_{\str,1}$ & 0.002  & 0.118  & 0.117  & 0.940  & 0.001  & 0.116  & 0.118  & 0.946  & -0.002  & 0.122  & 0.118  & 0.948  \\
          & $\hat{\tau}_{\str,2}$ & 0.001  & 0.118  & 0.117  & 0.939  & 0.000  & 0.115  & 0.118  & 0.947  & -0.003  & 0.122  & 0.118  & 0.945  \\
    \bottomrule
    \end{tabular}%
     }
    \end{threeparttable}
  \label{tab:changeinitial-model2}
  \begin{tablenotes}
				\footnotesize
				\item[] Note: $\hat{\tau}_{\textnormal{str-md}}$, stratified difference-in-medians estimator; $\hat{\tau}_{\textnormal{wt-md}}$, difference-in-weighted-medians esti-\\mator; $\hat{\tau}_{\str,1}$, stratified transformed difference-in-means estimator with $\hat{\tau}_{\textnormal{str-md}}$ as initial estimator;\\ $\hat{\tau}_{\str,2}$, stratified transformed difference-in-means estimator with $\hat{\tau}_{\textnormal{wt-md}}$ as initial estimator; SR, simple\\ randomization; STR, stratified randomization; MIN, Pocock and Simon's minimization; SD, standard\\ deviation; SE, standard error; CP, coverage probability.
			\end{tablenotes}
\end{table}%

\begin{table}[htbp]
  \centering
  \caption{Simulation results under Model 3 for different initial estimators}
  \begin{threeparttable}
    \resizebox{\textwidth}{!}{
    \begin{tabular}{rlrrrrrrrrrrrr}
    \toprule
          &       & \multicolumn{4}{c}{SR}        & \multicolumn{4}{c}{STR}       & \multicolumn{4}{c}{MIN} \\
\cmidrule{3-14}          &       & \multicolumn{1}{l}{Bias} & \multicolumn{1}{l}{SD} & \multicolumn{1}{l}{SE} & \multicolumn{1}{l}{CP} & \multicolumn{1}{l}{Bias} & \multicolumn{1}{l}{SD} & \multicolumn{1}{l}{SE} & \multicolumn{1}{l}{CP} & \multicolumn{1}{l}{Bias} & \multicolumn{1}{l}{SD} & \multicolumn{1}{l}{SE} & \multicolumn{1}{l}{CP} \\
    \midrule
    \multicolumn{1}{l}{$\tau$ = 0} &       &       &       &       &       &       &       &       &       &       &       &       &  \\
    \multicolumn{1}{l}{Normal tail} & $\hat{\tau}_{\textnormal{str-md}}$ & -0.002  & 0.095  & /     & /     & 0.001  & 0.089  & /     & /     & -0.003  & 0.088  & /     & / \\
          & $\hat{\tau}_{\textnormal{wt-md}}$ & -0.002  & 0.097  & 0.105  & 0.970  & 0.003  & 0.095  & 0.105  & 0.968  & 0.001  & 0.090  & 0.105  & 0.980  \\
          & $\hat{\tau}_{\str,1}$ & 0.000  & 0.082  & 0.082  & 0.946  & 0.002  & 0.079  & 0.081  & 0.964  & -0.001  & 0.077  & 0.081  & 0.957  \\
          & $\hat{\tau}_{\str,2}$ & 0.000  & 0.081  & 0.082  & 0.949  & 0.002  & 0.079  & 0.081  & 0.958  & -0.002  & 0.077  & 0.082  & 0.951  \\
          &       &       &       &       &       &       &       &       &       &       &       &       &  \\
    \multicolumn{1}{l}{Laplace tail} & $\hat{\tau}_{\textnormal{str-md}}$ & -0.001  & 0.096  & /     & /     & 0.005  & 0.093  & /     & /     & -0.006  & 0.094  & /     & / \\
          & $\hat{\tau}_{\textnormal{wt-md}}$ & 0.000  & 0.105  & 0.104  & 0.945  & 0.002  & 0.103  & 0.104  & 0.956  & -0.001  & 0.102  & 0.104  & 0.953  \\
          & $\hat{\tau}_{\str,1}$ & 0.001  & 0.095  & 0.093  & 0.943  & 0.003  & 0.093  & 0.093  & 0.948  & 0.000  & 0.092  & 0.093  & 0.950  \\
          & $\hat{\tau}_{\str,2}$ & 0.000  & 0.093  & 0.093  & 0.948  & 0.003  & 0.092  & 0.093  & 0.951  & -0.001  & 0.091  & 0.093  & 0.957  \\
          &       &       &       &       &       &       &       &       &       &       &       &       &  \\
    \multicolumn{1}{l}{Cauchy tail} & $\hat{\tau}_{\textnormal{str-md}}$ & 0.004  & 0.128  & /     & /     & 0.000  & 0.118  & /     & /     & 0.000  & 0.119  & /     & / \\
          & $\hat{\tau}_{\textnormal{wt-md}}$ & 0.005  & 0.137  & 0.137  & 0.948  & -0.004  & 0.133  & 0.138  & 0.950  & -0.001  & 0.131  & 0.138  & 0.958  \\
          & $\hat{\tau}_{\str,1}$ & 0.003  & 0.125  & 0.120  & 0.932  & -0.002  & 0.116  & 0.120  & 0.963  & -0.002  & 0.116  & 0.120  & 0.951  \\
          & $\hat{\tau}_{\str,2}$ & 0.003  & 0.126  & 0.120  & 0.929  & -0.002  & 0.116  & 0.120  & 0.961  & -0.002  & 0.117  & 0.120  & 0.950  \\
    \multicolumn{1}{l}{$\tau$ = 1} &       &       &       &       &       &       &       &       &       &       &       &       &  \\
    \multicolumn{1}{l}{Normal tail} & $\hat{\tau}_{\textnormal{str-md}}$ & 0.000  & 0.093  & /     & /     & 0.005  & 0.093  & /     & /     & 0.003  & 0.091  & /     & / \\
          & $\hat{\tau}_{\textnormal{wt-md}}$ & -0.001  & 0.094  & 0.105  & 0.961  & 0.003  & 0.099  & 0.105  & 0.951  & 0.007  & 0.095  & 0.105  & 0.964  \\
          & $\hat{\tau}_{\str,1}$ & 0.000  & 0.079  & 0.081  & 0.949  & 0.005  & 0.082  & 0.081  & 0.945  & 0.005  & 0.078  & 0.082  & 0.956  \\
          & $\hat{\tau}_{\str,2}$ & 0.000  & 0.079  & 0.081  & 0.948  & 0.004  & 0.082  & 0.081  & 0.945  & 0.005  & 0.077  & 0.082  & 0.957  \\
          &       &       &       &       &       &       &       &       &       &       &       &       &  \\
    \multicolumn{1}{l}{Laplace tail} & $\hat{\tau}_{\textnormal{str-md}}$ & 0.000  & 0.090  & /     & /     & -0.001  & 0.092  & /     & /     & -0.005  & 0.093  & /     & / \\
          & $\hat{\tau}_{\textnormal{wt-md}}$ & 0.002  & 0.101  & 0.104  & 0.950  & -0.002  & 0.103  & 0.104  & 0.948  & 0.001  & 0.103  & 0.104  & 0.949  \\
          & $\hat{\tau}_{\str,1}$ & 0.003  & 0.093  & 0.093  & 0.949  & -0.003  & 0.090  & 0.093  & 0.947  & 0.002  & 0.094  & 0.093  & 0.951  \\
          & $\hat{\tau}_{\str,2}$ & 0.002  & 0.092  & 0.093  & 0.954  & -0.003  & 0.089  & 0.093  & 0.949  & 0.001  & 0.092  & 0.093  & 0.950  \\
          &       &       &       &       &       &       &       &       &       &       &       &       &  \\
    \multicolumn{1}{l}{Cauchy tail} & $\hat{\tau}_{\textnormal{str-md}}$ & -0.005  & 0.121  & /     & /     & 0.009  & 0.122  & /     & /     & 0.004  & 0.123  & /     & / \\
          & $\hat{\tau}_{\textnormal{wt-md}}$ & -0.004  & 0.134  & 0.137  & 0.951  & 0.009  & 0.134  & 0.136  & 0.956  & 0.000  & 0.137  & 0.137  & 0.949  \\
          & $\hat{\tau}_{\str,1}$ & -0.003  & 0.119  & 0.119  & 0.954  & 0.006  & 0.119  & 0.120  & 0.950  & 0.001  & 0.126  & 0.119  & 0.932  \\
          & $\hat{\tau}_{\str,2}$ & -0.003  & 0.119  & 0.119  & 0.954  & 0.006  & 0.119  & 0.120  & 0.952  & 0.001  & 0.126  & 0.119  & 0.933  \\
    \bottomrule
    \end{tabular}%
    }
    \end{threeparttable}
  \label{tab:changeinitial-model3}
  \begin{tablenotes}
				\footnotesize
				\item[] Note: $\hat{\tau}_{\textnormal{str-md}}$, stratified difference-in-medians estimator; $\hat{\tau}_{\textnormal{wt-md}}$, difference-in-weighted-medians esti-\\mator; $\hat{\tau}_{\str,1}$, stratified transformed difference-in-means estimator with $\hat{\tau}_{\textnormal{str-md}}$ as initial estimator;\\ $\hat{\tau}_{\str,2}$, stratified transformed difference-in-means estimator with $\hat{\tau}_{\textnormal{wt-md}}$ as initial estimator; SR, simple\\ randomization; STR, stratified randomization; MIN, Pocock and Simon's minimization; SD, standard\\ deviation; SE, standard error; CP, coverage probability.
			\end{tablenotes}
\end{table}%

\subsection{Splitting techniques}
We investigate the influence of sample splitting techniques within model 1 for $n=1000$. The results for the other two models are similar so we omit them. We assess two splitting approaches: (i) simple random splitting devoid of stratum and treatment information, akin to the method in \cite{athey2021semiparametric} (SR splitting), and (ii) splitting the sample within each stratum and treatment arm, as outlined in Section~\ref{sec:ate-estimators} (CAR splitting). Table~\ref{tab:splitting} shows the results. Both splitting techniques exhibit comparable performance, with neither method dominating the other.

\begin{table}[htbp]
  \centering
  \caption{Simulation results under Model 1 for different splitting techniques}
  \begin{threeparttable}
    \resizebox{\textwidth}{!}{
    \begin{tabular}{rllrrrrrrrrrrrr}
    \toprule
          &       &       & \multicolumn{4}{c}{SR}        & \multicolumn{4}{c}{STR}       & \multicolumn{4}{c}{MIN} \\
\cmidrule{4-15}          &       &       & \multicolumn{1}{r}{Bias} & \multicolumn{1}{r}{SD} & \multicolumn{1}{r}{SE} & \multicolumn{1}{r}{CP} & \multicolumn{1}{r}{Bias} & \multicolumn{1}{r}{SD} & \multicolumn{1}{r}{SE} & \multicolumn{1}{r}{CP} & \multicolumn{1}{r}{Bias} & \multicolumn{1}{r}{SD} & \multicolumn{1}{r}{SE} & \multicolumn{1}{r}{CP} \\
    \midrule
    \multicolumn{1}{l}{$\tau$ = 0} &       &       &       &       &       &       &       &       &       &       &       &       &       &  \\
    \multicolumn{1}{l}{Normal tail} & \multicolumn{1}{l}{$\hat{\tau}_{\dim}$} & CAR splitting & 0.005  & 0.090  & 0.091  & 0.941  & -0.005  & 0.078  & 0.081  & 0.963  & -0.002  & 0.089  & /  & /  \\
          &       & SR splitting & 0.000  & 0.090  & 0.091  & 0.949  & -0.002  & 0.077  & 0.081  & 0.966  & 0.003  & 0.085  & /  & /  \\
          & \multicolumn{1}{l}{$\hat{\tau}_{\str}$} & CAR splitting & 0.005  & 0.079  & 0.082  & 0.945  & -0.005  & 0.078  & 0.081  & 0.965  & -0.003  & 0.080  & 0.082  & 0.949  \\
          &       & SR splitting & 0.002  & 0.079  & 0.081  & 0.952  & -0.002  & 0.077  & 0.081  & 0.967  & 0.003  & 0.076  & 0.082  & 0.962  \\
          &       &       &       &       &       &       &       &       &       &       &       &       &       &  \\
    \multicolumn{1}{l}{Laplace tail} & \multicolumn{1}{l}{$\hat{\tau}_{\dim}$} & CAR splitting & 0.006  & 0.107  & 0.103  & 0.932  & -0.006  & 0.092  & 0.092  & 0.943  & 0.002  & 0.101  & /  & /  \\
          &       & SR splitting & -0.003  & 0.107  & 0.102  & 0.933  & -0.001  & 0.091  & 0.092  & 0.957  & 0.000  & 0.100  & /  & /  \\
          & \multicolumn{1}{l}{$\hat{\tau}_{\str}$} & CAR splitting & 0.006  & 0.095  & 0.092  & 0.942  & -0.006  & 0.092  & 0.092  & 0.943  & 0.001  & 0.091  & 0.092  & 0.958  \\
          &       & SR splitting & -0.001  & 0.095  & 0.092  & 0.943  & -0.001  & 0.091  & 0.092  & 0.957  & 0.000  & 0.091  & 0.092  & 0.950  \\
          &       &       &       &       &       &       &       &       &       &       &       &       &       &  \\
    \multicolumn{1}{l}{Cauchy tail} & \multicolumn{1}{l}{$\hat{\tau}_{\dim}$} & CAR splitting & 0.002  & 0.128  & 0.128  & 0.940  & 0.003  & 0.118  & 0.118  & 0.946  & 0.006  & 0.124  & /  & /  \\
          &       & SR splitting & -0.003  & 0.126  & 0.128  & 0.950  & -0.004  & 0.115  & 0.118  & 0.950  & -0.001  & 0.123  & /  & /  \\
          & \multicolumn{1}{l}{$\hat{\tau}_{\str}$} & CAR splitting & 0.001  & 0.118  & 0.118  & 0.949  & 0.003  & 0.118  & 0.118  & 0.945  & 0.005  & 0.116  & 0.118  & 0.956  \\
          &       & SR splitting & -0.001  & 0.117  & 0.118  & 0.948  & -0.004  & 0.115  & 0.118  & 0.948  & -0.002  & 0.116  & 0.118  & 0.962  \\
    \multicolumn{1}{l}{$\tau$ = 1} &       &       &       &       &       &       &       &       &       &       &       &       &       &  \\
    \multicolumn{1}{l}{Normal tail} & \multicolumn{1}{l}{$\hat{\tau}_{\dim}$} & CAR splitting & 0.003  & 0.091  & 0.091  & 0.948  & 0.007  & 0.075  & 0.081  & 0.961  & -0.002  & 0.084  & /  & /  \\
          &       & SR splitting & 0.002  & 0.089  & 0.091  & 0.945  & -0.001  & 0.078  & 0.081  & 0.960  & -0.004  & 0.086  & /  & /  \\
          & \multicolumn{1}{l}{$\hat{\tau}_{\str}$} & CAR splitting & 0.003  & 0.078  & 0.082  & 0.968  & 0.007  & 0.075  & 0.081  & 0.960  & -0.002  & 0.075  & 0.082  & 0.967  \\
          &       & SR splitting & 0.000  & 0.078  & 0.082  & 0.955  & -0.001  & 0.078  & 0.081  & 0.959  & -0.004  & 0.076  & 0.082  & 0.968  \\
          &       &       &       &       &       &       &       &       &       &       &       &       &       &  \\
    \multicolumn{1}{l}{Laplace tail} & \multicolumn{1}{l}{$\hat{\tau}_{\dim}$} & CAR splitting & 0.002  & 0.105  & 0.103  & 0.937  & 0.004  & 0.092  & 0.092  & 0.945  & -0.001  & 0.097  & /  & /  \\
          &       & SR splitting & 0.002  & 0.102  & 0.103  & 0.950  & -0.002  & 0.089  & 0.092  & 0.959  & -0.003  & 0.096  & /  & /  \\
          & \multicolumn{1}{l}{$\hat{\tau}_{\str}$} & CAR splitting & 0.002  & 0.091  & 0.092  & 0.944  & 0.004  & 0.092  & 0.092  & 0.945  & -0.001  & 0.088  & 0.092  & 0.965  \\
          &       & SR splitting & 0.000  & 0.089  & 0.092  & 0.954  & -0.002  & 0.089  & 0.092  & 0.958  & -0.003  & 0.088  & 0.092  & 0.964  \\
          &       &       &       &       &       &       &       &       &       &       &       &       &       &  \\
    \multicolumn{1}{l}{Cauchy tail} & \multicolumn{1}{l}{$\hat{\tau}_{\dim}$} & CAR splitting & -0.002  & 0.129  & 0.128  & 0.946  & -0.007  & 0.122  & 0.118  & 0.933  & -0.006  & 0.126  & /  & /  \\
          &       & SR splitting & 0.003  & 0.132  & 0.128  & 0.945  & 0.001  & 0.119  & 0.119  & 0.951  & -0.003  & 0.128  & /  & /  \\
          & \multicolumn{1}{l}{$\hat{\tau}_{\str}$} & CAR splitting & -0.002  & 0.119  & 0.118  & 0.944  & -0.007  & 0.122  & 0.118  & 0.931  & -0.006  & 0.120  & 0.118  & 0.942  \\
          &       & SR splitting & 0.001  & 0.122  & 0.119  & 0.941  & 0.001  & 0.119  & 0.119  & 0.951  & -0.004  & 0.120  & 0.118  & 0.950  \\
    \bottomrule
    \end{tabular}%
    }
    \end{threeparttable}
  \label{tab:splitting}
  \begin{tablenotes}
				\footnotesize
				\item[] Note: $\hat{\tau}_{\dim}$, transformed difference-in-means estimator; $\hat{\tau}_{\str}$, stratified transformed difference-in-\\means estimator; SR, simple randomization; STR, stratified randomization; MIN, Pocock and\\ Simon's minimization; SD, standard deviation; SE, standard error; CP, coverage probability.
			\end{tablenotes}
\end{table}%

\section{Proof of main results}
\label{sec:proof-main}

Let $S^{(n)}= \{S_1,\ldots,S_n\}$, $A^{(n)}=\{A_1,\ldots,A_n\}$, $U^{(n)} = \{I_{1 \in \mathcal{L}_1},\ldots,I_{n \in \mathcal{L}_1}\}$, $Y^{(n)}=\{Y_1,\ldots,Y_n\}$, $D^{(\mathcal{L}_j)} = \{Y_i, A_i, S_i \}_{i \in \mathcal{L}_j}$, for $j = 1, 2$. Consider a measurable function $f(Y^{(n)}, A^{(n)}, S^{(n)}, U^{(n)})$ such as $\hat f_{0(1)}^{'} / \hat f_{0(1)}$. Given the independence of the splitting indicators $I_{i \in \mathcal{L}_1}$'s from $\{Y^{(n)}, A^{(n)}, S^{(n)}\}$, and that Assumptions~\ref{a1}--\ref{a2} hold, using the technique of \cite{Bugni2018} for deriving the law of large numbers and central limit theory, we can conclude that conditional on $(S^{(n)},A^{(n)}, U^{(n)})$, $f(Y^{(n)}, A^{(n)}, S^{(n)}, U^{(n)})$ has the same distribution as $f(Y^{'(n)},A^{(n)}, S^{(n)}, U^{(n)})$, where $Y^{'(n)} = \{Y_1^{'}, \ldots, Y_n^{'}\}$ with $Y_i{'} = A_i Y_i{'}(1) + (1 - A_i) Y_i^{'}(0)$. Furthermore, $\{Y_i^{'} \}_{i \in \mathcal{L}_1, S_i = k, A_i = a}$ are i.i.d. following the same distribution as $\{Y_i(a) \mid S_i = k\}$ and are independent of $(S^{(n)},A^{(n)}, U^{(n)})$. For simplicity, we will say that conditional on $(S^{(n)},A^{(n)}, U^{(n)})$, $\{Y_i\}_{i \in \mathcal{L}_1, S_i = k, A_i = a}$ are i.i.d. following the same distribution as $\{Y_i(a) \mid S_i = k \}$. This technique will be repeatedly used for convenience.

\subsection{Proof of Theorem~\ref{l0}}
\label{sec:proof-second-moment}

We need to show that 
    \begin{equation}\label{eqn:f-bounded}
         \sup_{y \in \mathbb R} \bigg\{\bigg|\frac{\hat{f}_{0(j)}^{\prime}}{\hat{f}_{0(j)}}\bigg|(y) \bigg\} = o_P(n^{1/2}),
    \end{equation}
   \begin{equation}
    \label{eqn:f-consistent}
    \int\left(\frac{\hat{f}_{0(j)}^{\prime}}{\hat{f}_{0(j)}}-\frac{f_0^{\prime}}{f_0}\right)^2(y) d F_0(y) = o_P(1),
    \end{equation}
    and for $m$ such that $m/n$ converges to a positive constant, with an i.i.d. sample $\{Y_1^{\prime},...,Y_m^{\prime}\}$ drawn from the density $f_{[k]0}$ that is independent of $\hat{f}_{0(j)}^{\prime}/\hat{f}_{0(j)}$, and any ${\delta}_n=O_P\left(n^{-1 / 2}\right)$, as $n \rightarrow \infty$, we have
    \begin{equation}
    \label{eqn:abs-f-continuity}
    \frac{1}{m} \sum_{i=1}^m \left|\left(\frac{\hat{f}_{0(j)}^{\prime}}{\hat{f}_{0(j)}}\right)^{\prime}\right|\left(Y_i^{\prime}+{\delta}_n\right)=O_P(1),
    \end{equation}
    \begin{equation}
    \label{eqn:f-continuity}
    \frac{1}{m} \sum_{i=1}^m \bigg(\frac{\hat{f}_{0(j)}^{\prime}}{\hat{f}_{0(j)}}\bigg)^{\prime}\left(Y_i^{\prime}+{\delta}_n\right) = \frac{1}{m} \sum_{i=1}^m \bigg(\frac{f_0^{\prime}}{f_0}\bigg)^{\prime}\left(Y_i^{\prime}\right) + \op.
    \end{equation}

\begin{proof}

    By symmetry, we only need to prove the results for $j = 1$. We will prove equations \eqref{eqn:f-bounded}--\eqref{eqn:f-continuity} in steps (I)--(IV), respectively.
    
    (I) We prove equation \eqref{eqn:f-bounded}. By definition and Assumption~\ref{cond:density-bickel}, it is easy to see that
    $$\sup_{y \in \mathbb R} \bigg\{\bigg|\frac{\hat{f}_{0(j)}^{\prime}}{\hat{f}_{0(j)}}\bigg|(y) \bigg\} \leq c_n = o(\sigma_n^{-1}) = o_P(n^{1/2}).$$

    (II) We prove equation \eqref{eqn:f-consistent}. Let $\mathcal{D}$ be the truncation field defined as:
    $$
    \mathcal{D} = \{ y \in \mathbb R : \hat{f}_{\sigma_n}(y) \geq d_n, \ |y| \leq e_n, \  |\hat{f}_{\sigma_n}^{\prime}(y)| \leq c_n \hat{f}_{\sigma_n}(y), \  |\hat{f}_{\sigma_n}^{\prime\prime}(y)| \leq b_n \hat{f}_{\sigma_n}(y) \}.
    $$

    Firstly, we show that 
    \begin{equation}
      \label{bickel6.1}
      \int_{\{ y: f_0(y)>0 \}}\left\{\frac{\hat{f}_{\sigma_n}^{\prime}}{\hat{f}_{\sigma_n}}(y)I_{\mathcal{D}}(y) -\frac{f_{\sigma_n}^{\prime}}{f_{\sigma_n}}(y)\right\}^2 f_{\sigma_n}(y) d y \cp 0,
    \end{equation}
    \begin{equation}
      \label{bickel6.1-doubleprime}
		\int_{\{ y: f_0(y)>0 \}}\left\{\frac{\hat{f}_{\sigma_n}^{\prime\prime}}{\hat{f}_{\sigma_n}}(y)I_{\mathcal{D}}(y) -\frac{f_{\sigma_n}^{\prime\prime}}{f_{\sigma_n}}(y)\right\}^2 f_{\sigma_n}(y) d y \cp 0, 
    \end{equation}
    where $f_{\sigma_n}(y)$ is the convolution of $f_0$ and $\phi_{\sigma_n}$, defined as $f_{\sigma_n}(y)=\int_{-\infty}^{\infty} f_0(x)\phi_{\sigma_n}(y-x) dx$. \cite{bickel1982adaptive} proved similar results for i.i.d. samples. We will generalize it to the samples under covariate-adaptive randomization with considerations of dependent treatment assignments. 

    We partition the integral in equation \eqref{bickel6.1} into two components: one over $\mathcal{D}$  and the other over $\mathcal{D}^C$, the complement of $\mathcal{D}$. Let's denote these two components as $T_1$ and $T_2$. Note that $T_1 \geq 0$ and $T_2 \geq 0$.

    By the dominated convergence theorem, it suffices for $T_1 \cp 0$ to show that $P(T_1 > \epsilon \mid S^{(n)},A^{(n)},U^{(n)}) \cp 0$ for any $\epsilon > 0$. By the Markov inequality, 
    \begin{eqnarray}
        P(T_1 > \epsilon \mid S^{(n)},A^{(n)},U^{(n)}) \leq  \epsilon^{-1} E\{ T_1 \mid S^{(n)},A^{(n)},U^{(n)} \}. \nonumber 
    \end{eqnarray}

    On $\mathcal{D}$, we have
    \begin{eqnarray}\label{eqn:t1}
    && E\{ T_1 \mid S^{(n)},A^{(n)},U^{(n)} \}  \leq 2\int_{\mathcal{D}} f_{\sigma_n}^{-1}(y) E\Big[ \big\{\hat{f}_{\sigma_n}^{\prime}(y)-f_{\sigma_n}^{\prime}(y)\big\}^2 \mid S^{(n)},A^{(n)},U^{(n)} \Big] d y \nonumber \\
     && \quad \quad \quad  + 2 \int_{\mathcal{D}} c_n^2 f_{\sigma_n}^{-1}(y) E\Big[\big\{\hat{f}_{\sigma_n}(y)-f_{\sigma_n}(y)\big\}^2 \mid S^{(n)},A^{(n)},U^{(n)} \Big] d y. 
    \end{eqnarray}

    Let the superscript $(l)$ denote the $l$-th order derivative with respect to $y$. Define $f_{\sigma_n[k]}(y) = \int_{-\infty}^{\infty} f_{[k]0}(x)\phi_{\sigma_n}(y-x) dx.$
Then, we have
\begin{eqnarray}\label{eqn:var}
    && \operatorname{Var} \{\hat{f}_{\sigma_n}^{(l)}(y) \mid S^{(n)},A^{(n)},U^{(n)}\} \nonumber \\
    &=& \operatorname{Var}\left( \sum_{k=1}^K 
    \frac{n_{[k]0(1)}}{\ncone} \cdot \frac{1}{n_{[k]0(1)}}\sum_{i \in \mathcal{L}_1, S_i = k, A_i = 0}\phi^{(l)}_i \mid S^{(n)},A^{(n)},U^{(n)} \right) \nonumber \\
    & =& \sum_{k=1}^K  \left(\frac{n_{[k]0(1)}}{\ncone}\right)^2 \operatorname{Var}\left(\frac{1}{n_{[k]0(1)}}\sum_{i \in \mathcal{L}_1, S_i = k, A_i = 0}\phi^{(l)}_i \mid S^{(n)},A^{(n)},U^{(n)} \right) \nonumber \\
    & \leq & \sum_{k=1}^K  \left(\frac{n_{[k]0(1)}}{\ncone}\right)^2 \cdot \kappa_l^{SR} \sigma_n^{-2l-1} n_{[k]0(1)}^{-1} f_{\sigma_n[k]}(y) \nonumber \\
    & = & \kappa_l^{SR} \sigma_n^{-2l-1} \sum_{k=1}^K \frac{n_{[k]0(1)}}{\ncone^2}f_{\sigma_n[k]}(y), \quad l=0,1,2,
\end{eqnarray}
where $\phi_i^{(l)} = \phi_{\sigma_n}^{(l)}(y-Y_i)$. The second equality holds because, conditional on $S^{(n)}, A^{(n)}$ and $U^{(n)}$, the $\phi_i^{(l)}$'s from different strata are independent (refer to the discussion at the start of Section~\ref{sec:proof-main}). The inequality in the third line results from applying the results in the proof of \citet[][Lemma 6.1]{bickel1982adaptive} to each stratum.

By a similar argument, we have
$$
\begin{aligned}
    E \{\hat{f}_{\sigma_n}^{(l)}(y) \mid S^{(n)},A^{(n)},U^{(n)}\} & = \sum_{k=1}^K  \frac{n_{[k]0(1)}}{\ncone} E\left(\frac{1}{n_{[k]0(1)}}\sum_{i \in \mathcal{L}_1, S_i = k, A_i = 0}\phi^{(l)}_i \mid S^{(n)},A^{(n)},U^{(n)} \right) \\
    &= \sum_{k=1}^K  \frac{n_{[k]0(1)}}{\ncone} f_{\sigma_n[k]}^{(l)}(y).
\end{aligned}
$$
Since $f_{\sigma_n}^{(l)}(y) = \sumk \ps f_{\sigma_n[k]}^{(l)}(y)$, then
$$
\Big | E \{\hat{f}_{\sigma_n}^{(l)}(y) \mid S^{(n)},A^{(n)},U^{(n)}\} - f_{\sigma_n}^{(l)}(y) \Big | \leq  \frac{ \Delta_n }{  \min_{k=1,\ldots,K} \ps }  f_{\sigma_n}^{(l)}(y).
$$

Let $\Pi$ denote the ratio of a circle's circumference to its diameter. A straightforward calculation yields
$$
\begin{aligned}
    f_{\sigma_n[k]}(y) =& \int_{-\infty}^{\infty} f_{[k]0}(x)\phi_{\sigma_n}(y-x) dx \leq \int_{-\infty}^{\infty} (\sqrt{2\Pi}\sigma_n)^{-1}f_{[k]0}(x) dx =(\sqrt{2\Pi}\sigma_n)^{-1},  \\
    \big\{f_{\sigma_n}(y) \big\}^2 =& f_{\sigma_n}(y) \sumk \ps f_{\sigma_n[k]}(y) \leq (\sqrt{2\Pi}\sigma_n)^{-1} f_{\sigma_n}(y),\\ 
    \big\{f_{\sigma_n[k]}^{\prime}(y)\big\}^2 =&  \bigg(\int_{-\infty}^{\infty} f_{[k]0}(x)\phi_{\sigma_n}^{\prime}(y-x) dx\bigg)^2 \\
    =& \int_{-\infty}^{\infty}f_{[k]0}(u) du \int_{-\infty}^{\infty} f_{[k]0}(v)dv \cdot \phi_{\sigma_n}^{\prime}(y-u)\phi_{\sigma_n}^{\prime}(y-v)\\
    =& \int_{-\infty}^{\infty}f_{[k]0}(u) du \int_{-\infty}^{\infty} f_{[k]0}(v)dv \cdot \sigma_n^{-4}(y-u)(y-v)\phi_{\sigma_n}(y-u)\phi_{\sigma_n}(y-v)\\
    \leq& \sigma_n^{-4}\int_{-\infty}^{\infty}f_{[k]0}(u) du \int_{-\infty}^{\infty} f_{[k]0}(v)dv \cdot \frac{(y-u)^2+(y-v)^2}{2}\phi_{\sigma_n}(y-u)\phi_{\sigma_n}(y-v)\\
    =& \sigma_n^{-4}\int_{-\infty}^{\infty}f_{[k]0}(u) (y-u)^2\phi_{\sigma_n}(y-u) du \int_{-\infty}^{\infty} \phi_{\sigma_n}(y-v)f_{[k]0}(v)dv\\
    \leq&{\sigma_n^{-2}}\int_{-\infty}^{\infty} 2e^{-1}(\sqrt{2\Pi}\sigma_n)^{-1}f_{[k]0}(u) du \int_{-\infty}^{\infty} \phi_{\sigma_n}(y-v)f_{[k]0}(v)dv\\
    =& 2 (\sqrt{2\Pi} e)^{-1} {\sigma_n^{-3}}f_{\sigma_n[k]}(y),
\end{aligned}
$$
$$
\begin{aligned}    
    \big\{f_{\sigma_n[k]}^{\prime\prime}(y)\big\}^2 =&  \bigg(\int_{-\infty}^{\infty} f_{[k]0}(x)\phi_{\sigma_n}^{\prime\prime}(y-x) dx\bigg)^2 \\
    =& \int_{-\infty}^{\infty}f_{[k]0}(u) du \int_{-\infty}^{\infty} f_{[k]0}(v)dv \cdot \phi_{\sigma_n}^{\prime\prime}(y-u)\phi_{\sigma_n}^{\prime\prime}(y-v)\\
    =& \int_{-\infty}^{\infty}f_{[k]0}(u) du \int_{-\infty}^{\infty} f_{[k]0}(v)dv \\
    & \cdot \sigma_n^{-4}\bigg\{\frac{(y-u)^2}{\sigma_n^2}-1\bigg\}\bigg\{\frac{(y-v)^2}{\sigma_n^2}-1\bigg\} \phi_{\sigma_n}(y-u)\phi_{\sigma_n}(y-v)\\
    \leq& \sigma_n^{-4}\int_{-\infty}^{\infty}f_{[k]0}(u) du \int_{-\infty}^{\infty} f_{[k]0}(v)dv \cdot \frac{(y-u)^4+(y-v)^4}{2\sigma_n^4}\phi_{\sigma_n}(y-u)\phi_{\sigma_n}(y-v)\\
    & + \sigma_n^{-4}\int_{-\infty}^{\infty}f_{[k]0}(u) du \int_{-\infty}^{\infty} f_{[k]0}(v)dv \cdot \phi_{\sigma_n}(y-u)\phi_{\sigma_n}(y-v)\\
    =& \sigma_n^{-4}\int_{-\infty}^{\infty}f_{[k]0}(u) \cdot (y-u)^4/\sigma_n^4 \cdot \phi_{\sigma_n}(y-u) du \int_{-\infty}^{\infty} \phi_{\sigma_n}(y-v)f_{[k]0}(v)dv\\
    &+ \sigma_n^{-4}\int_{-\infty}^{\infty}f_{[k]0}(u) \phi_{\sigma_n}(y-u) du \int_{-\infty}^{\infty} f_{[k]0}(v)\phi_{\sigma_n}(y-v)dv\\
    \leq&{\sigma_n^{-4}}\int_{-\infty}^{\infty} 16e^{-2}(\sqrt{2\Pi}\sigma_n)^{-1}f_{[k]0}(u) du \int_{-\infty}^{\infty} \phi_{\sigma_n}(y-v)f_{[k]0}(v)dv + \sigma_n^{-4}f_{\sigma_n[k]}^2(y)\\
    =& 16 (\sqrt{2\Pi} e^2)^{-1} {\sigma_n^{-5}}f_{\sigma_n[k]}(y) + (\sqrt{2\Pi})^{-1}\sigma_n^{-5}f_{\sigma_n[k]}(y).
\end{aligned}$$
Then, applying the Cauchy--Schwarz inequality and noting that $\sumk \ps = 1$, 
\begin{eqnarray}
\big\{f_{\sigma_n}^{\prime}(y)\big\}^2 &=& \Big\{\sumk \ps f_{\sigma_n[k]}^{\prime}(y) \Big\}^2 \leq \sumk \ps \big\{f_{\sigma_n[k]}^{\prime}(y)\big\}^2 \leq 2 (\sqrt{2\Pi} e)^{-1} {\sigma_n^{-3}}f_{\sigma_n}(y),  \nonumber \\
\big\{f_{\sigma_n}^{\prime\prime}(y)\big\}^2 &=& \Big\{\sumk \ps f_{\sigma_n[k]}^{\prime\prime}(y) \Big\}^2 \leq \sumk \ps \big\{f_{\sigma_n[k]}^{\prime\prime}(y)\big\}^2 \leq (16+e^2)(\sqrt{2\Pi}e^2)^{-1} {\sigma_n^{-5}}f_{\sigma_n}(y).  \nonumber
\end{eqnarray}
Therefore,
\begin{eqnarray}\label{eqn:bias}
    \Big | E \{\hat{f}_{\sigma_n}^{(l)}(y) \mid S^{(n)},A^{(n)},U^{(n)}\} - f_{\sigma_n}^{(l)}(y) \Big |^2 \leq  \frac{ \Delta_n^2 }{  \min_{k=1,\ldots,K} \ps^2 } \frac{4}{\sqrt{2\Pi}} \sigma_n^{-2l-1} f_{\sigma_n}(y), \ l=0,1,2.
\end{eqnarray}

Let $C_l = \max\{ \kappa_l^{SR}, 4/\sqrt{2\Pi} \} /  \min_{k=1,\ldots,K} \ps^2 $ and $\Lambda_n = \max_{k=1,\ldots,K} n_{[k]0(1)}/ \ncone$. Then $C_l$ is a constant independent of $n$ and $\Lambda_n \cp  \max_{k=1,\ldots,K} \ps $. Combining \eqref{eqn:var} and \eqref{eqn:bias}, we have
$$
E\Big[ \big\{\hat{f}_{\sigma_n}^{(l)}(y)-f_{\sigma_n}^{(l)}(y)\big\}^2 \mid S^{(n)},A^{(n)},U^{(n)} \Big] \leq C_l \Big\{ \frac{ \Lambda_n }{\ncone} + \Delta_n^2  \Big\} \sigma_n^{-2l-1} f_{\sigma_n}(y). 
$$
Substituting it into equation~\eqref{eqn:t1}, we have
$$
E\{ T_1 \mid S^{(n)},A^{(n)},U^{(n)} \} \leq 2 C_1 \Big\{ \frac{ \Lambda_n }{\ncone} + \Delta_n^2  \Big\} \frac{e_n}{\sigma_n^3} + 2 C_0 \Big\{ \frac{ \Lambda_n }{\ncone} + \Delta_n^2  \Big\} \frac{c_n^2 e_n}{\sigma_n}.
$$
By Assumption~\ref{cond:density-bickel}, we have $E\{ T_1 \mid S^{(n)},A^{(n)},U^{(n)} \} \cp 0$. Therefore $T_1 \cp 0$.

Similarly, let $T_1^{\prime}$ and $T_2^{\prime}$ denote the integrals in equation \eqref{bickel6.1-doubleprime} over $\mathcal{D}$ and  $\mathcal{D}^C$, respectively. Then, we have
$$
E\{ T_1^{\prime} \mid S^{(n)},A^{(n)},U^{(n)} \} \leq 2 C_2 \Big\{ \frac{ \Lambda_n }{\ncone} + \Delta_n^2  \Big\} \frac{e_n}{\sigma_n^5} + 2 C_0 \Big\{ \frac{ \Lambda_n }{\ncone} + \Delta_n^2  \Big\} \frac{b_n^2 e_n}{\sigma_n}.
$$
By Assumption~\ref{cond:density-bickel}, we have $E\{ T_1^{\prime} \mid S^{(n)},A^{(n)},U^{(n)} \} \cp 0$. Therefore $T_1^{\prime} \cp 0$.

On $\mathcal{D}^C$, we have  
\begin{equation}
\label{T2bound}
\begin{aligned}
E(T_2) \leq&    \int \frac{f_{\sigma_n}^{\prime 2}}{f_{\sigma_n}}(y)\bigg[P\left\{\left|\hat{f}_{\sigma_n}^{\prime}(y)\right|>c_n \hat{f}_{\sigma_n}(y)\right\} + P\left\{\left|\hat{f}_{\sigma_n}^{\prime\prime}(y)\right|>b_n \hat{f}_{\sigma_n}(y)\right\}\nonumber\\
&+P\left\{\hat{f}_{\sigma_n}(y)<d_n, f_0(y)>0\right\}+I\left(|y|>e_n\right)\bigg] d y, \nonumber \\
E(T_2^{\prime}) \leq&    \int \frac{f_{\sigma_n}^{\prime\prime 2}}{f_{\sigma_n}}(y)\bigg[P\left\{\left|\hat{f}_{\sigma_n}^{\prime}(y)\right|>c_n \hat{f}_{\sigma_n}(y)\right\} + P\left\{\left|\hat{f}_{\sigma_n}^{\prime\prime}(y)\right|>b_n \hat{f}_{\sigma_n}(y)\right\}\nonumber\\
&+P\left\{\hat{f}_{\sigma_n}(y)<d_n, f_0(y)>0\right\}+I\left(|y|>e_n\right)\bigg] d y. \nonumber 
\end{aligned} 
\end{equation}
Both of them converge to zero if 
\begin{eqnarray}
\label{e:f_hat_properties1}
    \hat{f}_{\sigma_n}(y) \cp f_0(y),&&\forall y, \quad \text{if} \quad n \sigma_n \rightarrow \infty,\\
\label{e:f_hat_properties2}
    \hat{f}_{\sigma_n}^{\prime}(y) \cp f_0^{\prime}(y),&& a.e. \quad y, \quad \text{if} \quad n \sigma_n^3 \rightarrow \infty,\\
\label{e:f_hat_properties3}
    \hat{f}_{\sigma_n}^{\prime\prime}(y) \cp f_0^{\prime\prime}(y),&& a.e. \quad y, \quad \text{if} \quad n \sigma_n^5 \rightarrow \infty,\\
\label{e:f_hat_properties4}
    \int \frac{f_{\sigma_n}^{\prime 2}}{f_{\sigma_n}}(y) d y &\leq& \int \frac{f_0^{\prime 2}}{f_0}(y) d y, \quad \forall n,\\
\label{e:f_hat_properties5}
    \int \frac{f_{\sigma_n}^{\prime\prime 2}}{f_{\sigma_n}}(y) d y &\leq& \int \frac{f_0^{\prime\prime 2}}{f_0}(y) d y, \quad \forall n;
\end{eqnarray}
see \cite{bickel1982adaptive}.

For deriving \eqref{e:f_hat_properties1}, conditional on $S^{(n)}, A^{(n)}$ and $U^{(n)}$, we have 
$$n_{[k]0(1)}^{-1}\sum_{i \in \mathcal{L}_1, S_i = k, A_i = 0} \phi_{\sigma}(y-Y_i) \cp f_{[k]0}(y), \quad n \sigma_n \rightarrow \infty.$$
Thus, conditional on $S^{(n)}, A^{(n)}$ and $U^{(n)}$, $\hat{f}_{\sigma_n}(y) - \sumk (n_{[k]0(1)}/\ncone)  f_{[k]0}(y) \cp 0$, implying unconditional convergence in probability. Moreover, since $p_{n[k](1)} = n_{[k]0(1)} / n_{0(1)} \cp \ps$, we have $\hat{f}_{\sigma_n}(y) \cp f_0(y)$. Similarly, we can derive \eqref{e:f_hat_properties2} and \eqref{e:f_hat_properties3}. Since \eqref{e:f_hat_properties4} is not related to $Y_i$, it holds by the proof in \cite{bickel1982adaptive}. Finally, applying the Cauchy--Schwarz inequality gives
$$\frac{f_{\sigma_n}^{\prime\prime 2}(y)}{f_{\sigma_n}(y)} = \frac{\big\{\int f_0^{\prime\prime}(y-x) \phi_{\sigma_n}(x) dx\big\}^2}{\int f_0(y-x)\phi_{\sigma_n}(x) dx }\leq \int \frac{f_0^{\prime\prime 2}(y-x)}{f_0(y-x)}\phi_{\sigma_n}(x) dx, \quad \forall y.$$ 
Integrating the above inequality with respect to $y$, we have
\begin{eqnarray*}
  \int_{-\infty}^{\infty} \frac{f_{\sigma_n}^{\prime\prime 2}(y)}{f_{\sigma_n}(y)} dy &\leq& \int_{-\infty}^{\infty} \frac{f_0^{\prime\prime 2}(y-x)}{f_0(y-x)} dy \int_{-\infty}^{\infty} \phi_{\sigma_n}(x) dx \\
  &=& \int_{-\infty}^{\infty} \frac{f_0^{\prime\prime 2}(z)}{f_0(z)} dz \int_{-\infty}^{\infty} \phi_{\sigma_n}(x) dx = \int_{-\infty}^{\infty} \frac{f_0^{\prime\prime 2}(z)}{f_0(z)} dz,
\end{eqnarray*}
that is, \eqref{e:f_hat_properties5} holds.

Therefore, we obtain $E(T_2) \to 0$ and $E(T_2^{\prime}) \to 0$, implying $T_2 \cp 0$ and $T_2^{\prime} \cp 0$. Consequently, equations \eqref{bickel6.1} and \eqref{bickel6.1-doubleprime} hold.

Moreover, by \eqref{e:f_hat_properties1}--\eqref{e:f_hat_properties5}, we have, as $n\rightarrow \infty$,
\begin{equation}
\label{bickel6.2}
    \int_{\{y:f_0(y)>0\}}\left\{\frac{f_{\sigma_n}^{\prime}}{\sqrt{f_{\sigma_n}}}(y)-\frac{f_0^{\prime}}{\sqrt{f_0}}(y)\right\}^2 d y \rightarrow 0,
\end{equation}
and 
\begin{equation}
\nonumber
    \int_{\{y:f_0(y)>0\}}\left\{\frac{f_{\sigma_n}^{\prime\prime}}{\sqrt{f_{\sigma_n}}}(y)-\frac{f_0^{\prime\prime}}{\sqrt{f_0}}(y)\right\}^2 d y \rightarrow 0.
\end{equation}

By Lemmas 6.3 in \cite{bickel1982adaptive}, as $n\rightarrow \infty$, the following term is bounded,
\begin{equation}
\nonumber
    \int_{\{y:f_0(y)>0\}} \sigma_n^{-2}\left(\sqrt{f_{\sigma_n}(y)}-\sqrt{f_0(y)}\right)^2 d y.
\end{equation}
Since $|\hat{f}_{\sigma_n}^{\prime}/\hat{f}_{\sigma_n}|(y) \leq c_n = o(\sigma_n^{-1})$, $|\hat{f}_{\sigma_n}^{\prime\prime}/\hat{f}_{\sigma_n}|(y) \leq b_n = o(\sigma_n^{-1})$, then we have
\begin{equation}
\label{bickel6.3}
    \int_{\{y:f_0(y)>0\}} \bigg(\frac{\hat{f}_{\sigma_n}^{\prime}}{\hat{f}_{\sigma_n}}\bigg)^2(y)\left(\sqrt{f_{\sigma_n}(y)}-\sqrt{f_0(y)}\right)^2 d y \cp 0,
\end{equation}
\begin{equation}
    \int_{\{y:f_0(y)>0\}} \bigg(\frac{\hat{f}_{\sigma_n}^{\prime\prime}}{\hat{f}_{\sigma_n}}\bigg)^2(y)\left(\sqrt{f_{\sigma_n}(y)}-\sqrt{f_0(y)}\right)^2 d y \cp 0. \nonumber
\end{equation}

Finally, we bound the $L_2$-loss of $g_n(y)$ by
$$
\begin{aligned}
& \int\left\{\frac{\hat{f}_{\sigma_n}^{\prime}}{\hat{f}_{\sigma_n}}(y)I_{\mathcal{D}}(y)-\frac{f_0^{\prime}}{f_0}(y)\right\}^2 f_0(y) d y \\
& \leq 3\left[\int_{\{y:f_0(y)>0\}}\left\{\frac{\hat{f}_{\sigma_n}^{\prime}}{\hat{f}_{\sigma_n}}(y)-\bigg(\frac{\hat{f}_{\sigma_n}^{\prime}}{\hat{f}_{\sigma_n}}\bigg)\bigg(\frac{f_{\sigma_n}}{f_0}\bigg)^{1 / 2}(y)\right\}^2I_{\mathcal{D}}(y) f_0(y) d y\right. \\
& \quad+\int_{\{y:f_0(y)>0\}}\left\{\bigg(\frac{\hat{f}_{\sigma_n}^{\prime}}{\hat{f}_{\sigma_n}}\bigg)\bigg(\frac{f_{\sigma_n}}{f_0}\bigg)^{1 / 2}(y)I_{\mathcal{D}}(y)-\bigg(\frac{f_{\sigma_n}^{\prime}}{f_{\sigma_n}}\bigg)\bigg(\frac{f_{\sigma_n}}{f_0}\bigg)^{1 / 2}(y)\right\}^2 f_0(y) d y \\
& \left.\quad+\int_{\{y:f_0(y)>0\}}\left\{\bigg(\frac{f_{\sigma_n}^{\prime}}{f_{\sigma_n}}\bigg)\bigg(\frac{f_{\sigma_n}}{f_0}\bigg)^{1 / 2}(y)-\frac{f_0^{\prime}}{f_0}(y)\right\}^2 f_0(y) d y\right],
\end{aligned}
$$
where the first term tends to zero in probability by (\ref{bickel6.3}), the second term tends to zero in probability by (\ref{bickel6.1}), and the last term tends to zero in probability by (\ref{bickel6.2}). Thus, equation~\eqref{eqn:f-consistent} holds.

Similarly, we can show that 
$$\int\left\{\frac{\hat{f}_{\sigma_n}^{\prime\prime}}{\hat{f}_{\sigma_n}}(y)I_{\mathcal{D}}(y)-\frac{f_0^{\prime\prime}}{f_0}(y)\right\}^2 f_0(y) d y = o_P(1).$$

(III) We prove equation \eqref{eqn:abs-f-continuity}. We slightly abuse notation by using $(\hat{f}_{0(1)}^{\prime\prime}/\hat{f}_{0(1)})(y)$ to denote $(\hat{f}_{\sigma_n}^{\prime\prime}/\hat{f}_{\sigma_n})(y) \cdot I_{\mathcal{D}}(y)$. Leveraging the Cauchy--Schwarz inequality, we have
    $$\int\bigg|\frac{\hat{f}_{0(1)}^{\prime\prime}}{\hat{f}_{0(1)}}-\frac{f_0^{\prime\prime}}{f_0}\bigg|(y) d F_0(y) = o_P(1).$$
For a specific stratum $k$, since $f_{[k]0}(y) \leq (\min_{k=1,...,K} \ps)^{-1}  f_{0}(y)$, then we have \begin{equation}
\label{e:L1-doubleprime}
    \int\bigg|\frac{\hat{f}_{0(1)}^{\prime\prime}}{\hat{f}_{0(1)}}-\frac{f_0^{\prime\prime}}{f_0}\bigg|(y) d F_{[k]0}(y) = o_P(1).
\end{equation}
    
    Next, by Cauchy--Schwarz inequality, we have $$ \begin{aligned}
        &\left(\int \bigg|\bigg(\frac{\hat{f}_{0(1)}^{\prime}}{\hat{f}_{0(1)}}\bigg)^2-\bigg(\frac{f_{0}^{\prime}}{f_{0}}\bigg)^2\bigg|(y) d F_{0}(y)\right)^2 \\
        &\leq \int \left(\frac{\hat{f}_{0(1)}^{\prime}}{\hat{f}_{0(1)}}-\frac{f_{0}^{\prime}}{f_{0}}\right)^2(y) d F_{0}(y) \cdot \int \left(\frac{\hat{f}_{0(1)}^{\prime}}{\hat{f}_{0(1)}}+\frac{f_{0}^{\prime}}{f_{0}}\right)^2(y) d F_{0}(y).
    \end{aligned}$$
The first integration on the right side of the above inequality is $o_P(1)$ by equation \eqref{eqn:f-consistent}. Moreover, the second integration is $O_P(1)$ due to equation \eqref{eqn:f-consistent} and $$ I(f_0) =  \int \left(\frac{f_{0}^{\prime}}{f_{0}}\right)^2(y) d F_{0}(y) < \infty.$$
Thus, 
$$
\left(\int \bigg|\bigg(\frac{\hat{f}_{0(1)}^{\prime}}{\hat{f}_{0(1)}}\bigg)^2-\bigg(\frac{f_{0}^{\prime}}{f_{0}}\bigg)^2\bigg|(y) d F_{0}(y)\right)^2  = \op.
$$
Since $f_{[k]0}(y) \leq (\min_{k=1,...,K} \ps)^{-1}  f_{0}(y)$, then we have 
    \begin{equation}
    \label{e:L1-square}
        \int \bigg|\bigg(\frac{\hat{f}_{0(1)}^{\prime}}{\hat{f}_{0(1)}}\bigg)^2- \bigg(\frac{f_{0}^{\prime}}{f_{0}}\bigg)^2\bigg|(y) d F_{[k]0}(y) = \op.
    \end{equation}

    Combining \eqref{e:L1-doubleprime} and \eqref{e:L1-square}, by the triangle inequality, we have
    $$\begin{aligned}
        \int \bigg|\bigg(\frac{\hat{f}_{0(1)}^{\prime}}{\hat{f}_{0(1)}}\bigg)^{\prime} \bigg|(y) dF_{[k]0}(y) =&  \int \bigg|\frac{\hat{f}_{0(1)}^{\prime\prime}}{\hat{f}_{0(1)}} - \bigg(\frac{\hat{f}_{0(1)}^{\prime}}{\hat{f}_{0(1)}}\bigg)^{2} \bigg|(y) dF_{[k]0}(y)\\
        \leq& \int \bigg|\frac{f_0^{\prime\prime}}{f_0}\bigg|(y) d F_{[k]0}(y) + \int \bigg(\frac{f_{0}^{\prime}}{f_{0}}\bigg)^2(y) d F_{[k]0}(y) + o_P(1).
    \end{aligned}$$
    Noting that $$\sum_{k=1}^K \ps \cdot \int \bigg|\frac{f_0^{\prime\prime}}{f_0}\bigg|(y) d F_{[k]0}(y) = \int |f_0^{\prime\prime}(y)| dy < \infty,$$
    $$\sum_{k=1}^K \ps \cdot \int \bigg(\frac{f_{0}^{\prime}}{f_{0}}\bigg)^2(y) d F_{[k]0}(y) = \int \bigg(\frac{f_{0}^{\prime}}{f_{0}}\bigg)^2(y) d F_{0}(y) = I(f_0) < \infty,$$
    we have $$\int |(\hat{f}_{0(1)}^{\prime}/\hat{f}_{0(1)})^{\prime} |(y) dF_{[k]0}(y) \leq (\min_{k=1,...,K} \ps)^{-1} \big\{\int |f_0^{\prime\prime}(y)| dy + I(f_0)\big\} + \op.$$

    By the argument in \cite{bickel1982adaptive}, we establish \eqref{eqn:abs-f-continuity} by substituting  ${\delta}_n$ with $\tilde{\delta}_n=t_n n^{-1 / 2} \to 0$, where $t_n$ is an arbitrary convergent deterministic sequence.

    By Assumption \ref{cond:partial-score}, we have
    \begin{equation}
        \label{e:prop2-tmp}
    \begin{aligned}
         \int \bigg|\bigg(\frac{\hat{f}_{0(1)}^{\prime}}{\hat{f}_{0(1)}}\bigg)^{\prime} \bigg|(y+\tilde{\delta}_n)f_{[k]0}(y) dy= &\int \bigg|\bigg(\frac{\hat{f}_{0(1)}^{\prime}}{\hat{f}_{0(1)}}\bigg)^{\prime} \bigg|(y)f_{[k]0}(y-\tilde{\delta}_n) dy\\ 
    =& \int \bigg|\bigg(\frac{\hat{f}_{0(1)}^{\prime}}{\hat{f}_{0(1)}}\bigg)^{\prime} \bigg|(y) \cdot f_{[k]0}(y)\big\{1+o(1)\big\} dy\\
    \leq& \big(\min_{k=1,...,K} \ps\big)^{-1} \cdot \bigg\{\int |f_0^{\prime\prime}(y)| dy + I(f_0)\bigg\} + \op.
    \end{aligned}
    \end{equation}
    
If $\{Y_i^{\prime}\}_{i=1,...,m}$ are i.i.d. samples drawn from the density $f_{[k]0}$ and are independent of $\hat{f}_{0(1)}^{\prime}/\hat{f}_{0(1)}$, by the weak law of large numbers, we have
$$\begin{aligned}
    &\frac{1}{m} \sum_{i=1}^m \bigg|\bigg(\frac{\hat{f}_{0(1)}^{\prime}}{\hat{f}_{0(1)}}\bigg)^{\prime}\bigg|\left(Y_i^{\prime}+\tilde{\delta}_n\right) \mid \hat{f}_{0(1)} \\
&= \int \bigg|\bigg(\frac{\hat{f}_{0(1)}^{\prime}}{\hat{f}_{0(1)}}\bigg)^{\prime} \bigg|(y+\tilde{\delta}_n)f_{[k]0}(y) dy + \op \\
&\leq \big(\min_{k=1,...,K} \ps\big)^{-1} \cdot \bigg\{\int |f_0^{\prime\prime}(y)| dy + I(f_0)\bigg\} + \op.\end{aligned}$$
Therefore,
$$\frac{1}{m} \sum_{i=1}^m \bigg|\bigg(\frac{\hat{f}_{0(1)}^{\prime}}{\hat{f}_{0(1)}}\bigg)^{\prime}\bigg|\left(Y_i^{\prime}+\tilde{\delta}_n\right) = O_P(1).$$

(IV) We prove equation \eqref{eqn:f-continuity}. Similar to (III), if $\{Y_i^{\prime}\}_{i=1,...,m}$ are i.i.d. samples drawn from the density $f_{[k]0}$ and are independent of $\hat{f}_{0(1)}^{\prime}/\hat{f}_{0(1)}$, by the weak law of large numbers, we have
$$\begin{aligned}
    &\frac{1}{m} \sum_{i=1}^m \bigg\{\bigg(\frac{\hat{f}_{0(1)}^{\prime}}{\hat{f}_{0(1)}}\bigg)^{\prime}(Y_i^{\prime}+\tilde{\delta}_n) - \bigg(\frac{f_{0}^{\prime}}{f_0}\bigg)^{\prime}(Y_i^{\prime}) \bigg\} \mid \hat{f}_{0(1)}  \\
=& \int \bigg\{\bigg(\frac{\hat{f}_{0(1)}^{\prime}}{\hat{f}_{0(1)}}\bigg)^{\prime} (y+\tilde{\delta}_n)- \bigg(\frac{f_{0}^{\prime}}{f_0}\bigg)^{\prime}(y)\bigg\}f_{[k]0}(y) dy + \op \\
=& \int \frac{\hat{f}_{0(1)}^{\prime\prime}}{\hat{f}_{0(1)}}(y)f_{[k]0}(y-\tilde{\delta}_n) dy - \int \bigg(\frac{\hat{f}_{0(1)}^{\prime}}{\hat{f}_{0(1)}}\bigg)^{2} (y)f_{[k]0}(y-\tilde{\delta}_n) dy - \int \bigg(\frac{f_{0}^{\prime}}{f_0}\bigg)^{\prime}(y)f_{[k]0}(y) dy + \op \\
\leq& \int \frac{\hat{f}_{0(1)}^{\prime\prime}}{\hat{f}_{0(1)}}(y)f_{[k]0}(y) dy - \int \bigg(\frac{\hat{f}_{0(1)}^{\prime}}{\hat{f}_{0(1)}}\bigg)^{2} (y)f_{[k]0}(y) dy + o(1) \cdot \bigg\{\int \bigg|\frac{\hat{f}_{0(1)}^{\prime\prime}}{\hat{f}_{0(1)}}\bigg|(y)f_{[k]0}(y) dy \\
&+ \int \bigg|\bigg(\frac{\hat{f}_{0(1)}^{\prime}}{\hat{f}_{0(1)}}\bigg)^{2}\bigg| (y)  f_{[k]0}(y)dy\bigg\} - \bigg\{\int\frac{f_0^{\prime\prime}}{f_0}(y)f_{[k]0}(y) dy - \int \bigg(\frac{f_0^{\prime}}{f_0}\bigg)^{2} (y)f_{[k]0}(y) dy\bigg\}+ \op \\
\leq& \int\bigg|\frac{\hat{f}_{0(1)}^{\prime\prime}}{\hat{f}_{0(1)}}-\frac{f_0^{\prime\prime}}{f_0}\bigg|(y)f_{[k]0}(y) dy + \int \bigg|\bigg(\frac{\hat{f}_{0(1)}^{\prime}}{\hat{f}_{0(1)}}\bigg)^2- \bigg(\frac{f_{0}^{\prime}}{f_{0}}\bigg)^2\bigg|(y)f_{[k]0}(y) dy + \op\\
=& \op.\end{aligned}$$
Therefore,
$$\frac{1}{m} \sum_{i=1}^m \bigg(\frac{\hat{f}_{0(1)}^{\prime}}{\hat{f}_{0(1)}}\bigg)^{\prime}\left(Y_i^{\prime}+\tilde{\delta}_n\right) = \frac{1}{m} \sum_{i=1}^m \bigg(\frac{f_0^{\prime}}{f_0}\bigg)^{\prime}\left(Y_i^{\prime}\right) + \op.$$

\end{proof}

\subsection{Proof of Proposition~\ref{prop:root-n-consistent}}
\begin{proof}
    Proposition~\ref{prop:root-n-consistent} directly follows from \citet[][Theorem 3.1 and Theorem 3.2]{zhang2020quantile}, and therefore, we omit its proof. 
\end{proof}

\subsection{Proof of Proposition~\ref{prop:root-n-consistent-str}}
\label{sec:proof-initial}

\begin{proof}
Let $\mathcal{Y}$ be a set of real numbers or random variables, and $\textnormal{median}(\mathcal{Y})$ denotes the median of its elements. Condition on $(S^{(n)}, A^{(n)})$, and employing the results under simple randomization \citep[][Theorem A.1]{athey2021semiparametric}, we have
$$\sqrt{n_{[k]}}\big\{\textnormal{median}(\left\{Y_i\right\}_{S_i=k,A_i=1})-\textnormal{median}(\left\{Y_i\right\}_{S_i=k,A_i=0})-\tau_{[k]}\big\} \mid (S^{(n)}, A^{(n)}) \xrightarrow{d} N(0, \sigma_{[k], \textnormal{median}}^2),$$
where $\sigma_{[k], \textnormal{median}}^2 = \{4\pi(1-\pi)f_{[k]0}^{2}(F_{[k]0}^{-1}(1/2))\}^{-1}$.

As $f_{[k]0}(F_{[k]0}^{-1}(1/2)) > 0$ and $n_{[k]} / n \xrightarrow{P} \ps $, it follows that 
$$
\sqrt{n}\{\textnormal{median}(\{Y_i\}_{S_i=k,A_i=1})-\textnormal{median}(\{Y_i\}_{S_i=k,A_i=0})-\tau_{[k]}\} =  O_{P}(1).
$$

By weighting the above term with the stratum proportion $p_{n[k]}$ and and summing over $k$, we have
$$\sqrt{n} \bigg[\sum_{k=1}^K p_{n[k]}\left\{\textnormal{median}(\left\{Y_i\right\}_{S_i=k,A_i=1})-\textnormal{median}(\left\{Y_i\right\}_{S_i=k,A_i=0}) \right\} -  \sum_{k=1}^K p_{n[k]}\tau_{[k]} \bigg] = O_P(1).$$
Given $\sqrt{n}(p_{n[k]}- p_{[k]}) = O_P(1)$ and $\sum_{k=1}^K p_{[k]}\tau_{[k]} = \tau$, we have $$\sqrt{n} \bigg[\sum_{k=1}^K p_{n[k]}\left\{\textnormal{median}(\left\{Y_i\right\}_{S_i=k,A_i=1})-\textnormal{median}(\left\{Y_i\right\}_{S_i=k,A_i=0}) \right\} - \tau \bigg] =O_P(1).$$
Therefore,
$
\sqrt{n} ( \hat{\tau}_{\textnormal{str-md}} - \tau) =O_P(1).
$

\end{proof}

	\subsection{Proof of Theorem \ref{t1}}

Before proving the result, we introduce two useful lemmas whose proofs are provided in Section~\ref{sec:lemma}. Note that Lemma~\ref{l1} below relies on the less stringent Assumption~\ref{a3} rather than the more restrictive Assumption~\ref{a4}.

 \begin{lemma}
    \label{l1}
    Under Assumptions \ref{a1}--\ref{a3} and \ref{cond:second-moment}, $\hat{I}(f_0)=I\left(f_0\right)+o_P(1).$    
\end{lemma}

 \begin{remark}
    Lemma \ref{l1} requires weaker conditions compared to Theorem \ref{t1}. Even in cases where the designs do not satisfy Assumption \ref{a4} (e.g., Pocock and Simon's minimization), the consistency of $\hat{I}(f_0)$ can still be preserved as long as the design satisfies Assumption~\ref{a3}.
 \end{remark}

 \begin{lemma}
\label{l2}
    Under Assumptions \ref{a1}, \ref{a2}, \ref{a4} and \ref{cond:second-moment}, we have 
    $$\frac{1}{\sqrt{n}}\sum_{i \in \mathcal{L}_1} \psi_{\hat{f}_{0(2)}}(A_i, Y_i ; \tau) = \frac{1}{\sqrt{n}}\sum_{i \in \mathcal{L}_1} \psi_{f_0}(A_i, Y_i ; \tau) + o_P(1),$$
    $$\frac{1}{\sqrt{n}}\sum_{i \in \mathcal{L}_2} \psi_{\hat{f}_{0(1)}}(A_i, Y_i ; \tau) = \frac{1}{\sqrt{n}}\sum_{i \in \mathcal{L}_1} \psi_{f_0}(A_i, Y_i ; \tau) + o_P(1).$$
\end{lemma}

	\begin{proof}[Proof of Theorem~\ref{t1}]
	By definition and the first order Taylor expansion, we have
	$$ \begin{aligned}
		\hat{\tau}_{\dim}-\tau &= \tilde{\tau}-\tau+\frac{1}{n}\left(\sum_{i \in \mathcal{L}_1} \psi_{\hat{f}_{0(2)}}\left(A_i, Y_i ; \tilde{\tau}\right)+\sum_{i \in \mathcal{L}_2} \psi_{\hat{f}_{0(1)}}\left(A_i, Y_i ; \tilde{\tau}\right)\right) \\
		&= \tilde{\tau}-\tau+\frac{1}{n}\left(\sum_{i \in \mathcal{L}_1} \psi_{\hat{f}_{0(2)}}\left(A_i, Y_i ; \tau\right)+\sum_{i \in \mathcal{L}_2} \psi_{\hat{f}_{0(1)}}\left(A_i, Y_i ; \tau\right)\right) \\  & \quad+ \frac{1}{n} (\tilde{\tau}-\tau)\left(\sum_{i \in \mathcal{L}_1} \frac{\partial}{\partial \tau} \psi_{\hat{f}_{0(2)}}\left(A_i, Y_i ; \bar{\tau}\right)+\sum_{i \in \mathcal{L}_2} \frac{\partial}{\partial \tau} \psi_{\hat{f}_{0(1)}}\left(A_i, Y_i ; \bar{\tau}\right)\right),
	\end{aligned}$$
    where $\bar{\tau}$ is between $\tilde{\tau}$ and $\tau$. By Lemma \ref{l2}, we have
   \begin{eqnarray*}
   \hat{\tau}_{\dim}-\tau &=&  \frac{1}{n}\sum_{i=1}^{n} \psi_{f_0}(A_i, Y_i ; \tau)+ (\tilde{\tau}-\tau)  \bigg[1+\frac{1}{n}\Big\{\sum_{i \in \mathcal{L}_1} \frac{\partial}{\partial \tau} \psi_{\hat{f}_{0(2)}}\left(A_i, Y_i ; \bar{\tau}\right) \\
   && +\sum_{i \in \mathcal{L}_2} \frac{\partial}{\partial \tau} \psi_{\hat{f}_{0(1)}}\left(A_i, Y_i ; \bar{\tau}\right)\Big\}\bigg]+o_P\Big(\frac{1}{\sqrt{n}}\Big)  \\
 &=& B_1 + (\tilde{\tau}-\tau)(1+B_2) + o_P\Big(\frac{1}{\sqrt{n}}\Big),
   \end{eqnarray*}
	where 
	$$ \begin{aligned} B_1 =& \frac{1}{n}\sum_{i=1}^{n} \psi_{f_0}\left(A_i, Y_i ; \tau\right), \\
		B_2 =& \frac{1}{n} \left(\sum_{i \in \mathcal{L}_1} \frac{\partial}{\partial \tau} \psi_{\hat{f}_{0(2)}}\left(A_i, Y_i ; \bar{\tau}\right)+\sum_{i \in \mathcal{L}_2} \frac{\partial}{\partial \tau} \psi_{\hat{f}_{0(1)}}\left(A_i, Y_i ; \bar{\tau}\right)\right).\end{aligned} $$

Recall that $Z_i(a) = -I(f_a)^{-1} \cdot (f_a^{\prime}/f_a)(Y_i(a))$, $a=0,1$, and abbreviate $Z_i(A_i) = A_i Z_i(1) + (1 - A_i) Z_i(0)$ to $Z_i$. Since $f_1(y) = f_0(y-\tau)$, then $I(f_1) = I(f_0)$. Moreover,
$$\bigg(\int_{-\infty}^\infty |f_0^{\prime}(y)| dy \bigg)^2 \leq \int_{-\infty}^\infty \frac{\big(f_0^{\prime}\big)^2}{f_0}(y) dy \cdot \int_{-\infty}^\infty f_0(y) dy = I(f_0) \cdot 1 < \infty. $$ 
This implies $\int_{-\infty}^\infty |f_0^{\prime}(y)| dy < \infty$, concluding that $f_0(\infty)$ and $f_0(-\infty)$ both exist and equal 0. Consequently,
$$-I(f_0) \cdot E\{Z_i(0)\} = \int_{-\infty}^\infty \frac{\big(f_0^{\prime}\big)}{f_0}(y) d F_0(y) =  \int_{-\infty}^\infty f_0^{\prime}(y) d y = f_0(\infty) - f_0(-\infty) = 0. $$
Therefore, $E\{Z_i(0)\} = 0$. Similarly, $E\{Z_i(1)\}  = 0$.

For $Y_i$'s satisfying Assumptions \ref{a1}, \ref{a2} and \ref{a4}, it can be easily shown that these assumptions remain valid when $Y_i$ is replaced by $Z_i$. Moreover, the transformed outcomes $Z_i(a)$'s have finite second-order moment with $E\{Z_i(a)\}^2 = [{I(f_0)}]^{-1}$, and

 $$\begin{aligned}
     B_1 &= -\frac{1}{n\pi}\sumn A_i \frac{1}{I\left(f_0\right)}\cdot \frac{f_0^{\prime}}{f_0}(Y_i-\tau) + \frac{1}{n(1-\pi)}\sumn (1 - A_i) \frac{1}{I\left(f_0\right)}\cdot \frac{f_0^{\prime}}{f_0}(Y_i) \\
      &= \frac{1}{n\pi}\sumn A_i Z_i -  \frac{1}{n(1-\pi)}\sumn (1 - A_i) Z_i.
 \end{aligned}$$
Applying \citet[][Theorem 4.1]{Bugni2018} to $Z_i(a)$, we have  $n_1^{-1/2} \sumn A_i Z_i = O_P(1)$,  and $n_0^{-1/2} \sumn (1 - A_i) Z_i = O_P(1)$. Alongside $n_1/n \cp \pi$ and $n_0 /n \cp (1-\pi)$, we have
$$
B_1 =  \frac{1}{n_1}\sumn A_i Z_i -  \frac{1}{n_0}\sumn (1 - A_i) Z_i + o_P\Big(\frac{1}{\sqrt{n}}\Big).
$$

	By Theorem 4.1 in \cite{Bugni2018}, we have 
	$$ \sqrt{n} B_1 \cd  N(0, V_{Z}^2 + V_{H}^2 + V_{A}^2), $$
where
$$
V_{Z}^2 = \frac{\Var\{\tilde Z_i(1)\}}{\pi} + \frac{\Var\{\tilde Z_i(0)\}}{1-\pi}, \quad V_{H}^2 = E\{ \check{Z}_i(1) - \check{Z}_i(0) \}^2, \quad V_{A}^2 = E\Big[ q_{[S_i]} \Big\{ \frac{\check Z_i(1)}{\pi} + \frac{\check Z_i(0)}{1-\pi}  \Big\}^2  \Big],
$$
with $\tilde Z_i(a) = Z_i(a) - E\{Z_i(a) \mid S_i  \}$ and $\check Z_i(a) = E\{Z_i(a) \mid S_i  \} -  E\{Z_i(a)\}$, $a=0,1$.

 Moreover, since $\sqrt{n} (\tilde \tau - \tau) = O_P(1)$ and $n_{[k]1(1)}/n \xrightarrow{P} \pi\ps/2$, then applying Assumption \ref{cond:second-moment}(iii) for $\{Y_i\}_{i \in \mathcal{L}_1, S_i=k, A_i=a}$ conditional on $\{S^{(n)},A^{(n)},U^{(n)}\}$, we have
 $$\begin{aligned}
     \frac{1}{n_{[k]1(1)}} \sum_{i \in \mathcal{L}_1, S_i =k, A_i=1} \left(\frac{\hat{f}_{0(1)}^{\prime}}{\hat{f}_{0(1)}}\right)^{\prime}\left(Y_i-\tau+(\tau-\bar{\tau})\right)=\frac{1}{n_{[k]1(1)}}  \sum_{i \in \mathcal{L}_1, S_i =k, A_i=1} \left(\frac{{f_0}^{\prime}}{f_0}\right)^{\prime}\left(Y_i-\tau\right)+o_P(1).
 \end{aligned}
 $$
 A similar result holds for $\{Y_i\}_{i \in \mathcal{L}_2, S_i = k, A_i=a}$. Summing over $k$ with weights $n_{[k]1(1)}/n$, we have
	$$
	\begin{aligned}
		B_2 
		&=\frac{1}{n} \left(\sum_{i \in \mathcal{L}_1} \frac{\partial}{\partial \tau} \psi_{\hat{f}_{0(2)}}\left(A_i, Y_i ; \bar{\tau}\right)+\sum_{i \in \mathcal{L}_2} \frac{\partial}{\partial \tau} \psi_{\hat{f}_{0(1)}}\left(A_i, Y_i ; \bar{\tau}\right)\right) \\
		&=\frac{1}{n} \sum_{i=1}^{n} \frac{\partial}{\partial \tau} \psi_{f_0}\left(A_i, Y_i ; \tau\right)+o_P(1) \\
            &=-\frac{1}{nI(f_0)}\sum_{i=1}^n\bigg\{ \frac{A_i}{\pi}\cdot \frac{\partial}{\partial \tau}\frac{{f_0}^{\prime}}{f_0}\left(Y_i-\tau\right)-\frac{1-A_i}{\pi}\cdot \frac{\partial}{\partial \tau}\frac{{f_0}^{\prime}}{f_0}\left(Y_i\right)\bigg\}+o_P(1). \\
		&=\frac{1}{nI(f_0)}\sum_{i=1}^n \frac{A_i}{\pi}\left(\frac{{f_0}^{\prime}}{f_0}\right)^{\prime}\left(Y_i-\tau\right)+o_P(1). \\
	\end{aligned}
	$$
Since $\int_{-\infty}^\infty |f_0^{\prime\prime}(y)| dy < \infty$, then
 $$
 E \left(\frac{{f_0}^{\prime}}{f_0}\right)^{\prime}\left(Y_i(1)-\tau\right) = \int \left({f_0}^{\prime}/f_0\right)^{\prime}(y - \tau) d F_{1}(y) = \int \left({f_0}^{\prime}/f_0\right)^{\prime}(y) d F_{0}(y) = - I(f_0) < \infty.
 $$
Additionally, $\int_{-\infty}^\infty |f_0^{\prime\prime}(y)| dy < \infty$ also indicates that $|({f_0}^{\prime}/f_0)^{\prime}(Y_i(1)-\tau)|$ has a finite first-order moment. By the strong law of large numbers under covariate-adaptive randomization \citep[][Lemma B.3]{Bugni2018}, we have 
 $$
 \frac{1}{n_1} \sumn A_i \left(\frac{{f_0}^{\prime}}{f_0}\right)^{\prime}\left(Y_i-\tau\right) = -\pi I(f_0) + o_P(1).
 $$
 Together with $n_1/n \cp \pi$, we have $B_2 + 1 = o_P(1)$.

Since $\sqrt{n} (\tilde \tau - \tau) = O_P(1)$, then by Slutsky's theorem, we have
$$
\sqrt{n} (\hat{\tau}_{\dim}-\tau) = \sqrt{n} B_1 + o_P(1) \cd  N(0, V_{Z}^2 + V_{H}^2 + V_{A}^2). 
$$
 \end{proof}

 \subsection{Proof of Corollary \ref{cor1}}
 \begin{proof}
Recall that $V_{A}^2 = E[ q_{[S_i]} \{ \check Z_i(1)/\pi + \check Z_i(0)/(1-\pi)  \}^2  ]$, where $\check Z_i(a) = E\{Z_i(a) \mid S_i  \} -  E\{Z_i(a)\}$, $a=0,1$. When $q_{[k]} = 0$ for $k=1,...,K$, we have $V_A^2 = 0$. By Theorem \ref{t1}, we have 
        $$\sqrt{n}(\hat{\tau}_{\dim}-\tau) \stackrel{d}{\rightarrow} N(0, V_{Z}^2 + V_{H}^2).$$ 
 \end{proof}

 \subsection{Proof of Theorem \ref{thm:Athey-var}}
 \begin{proof}
     By Lemma \ref{l1}, we have $\hat{I}(f_0)=I(f_0)
     +o_P(1)$, which implies that $\{\pi(1-\pi)\hat{I}(f_0)\}^{-1}$ is a consistent estimator of $\{\pi(1-\pi)I(f_0)\}^{-1}$. Note that $ E\{Z_i(a)\} = 0$ and $\Var\{Z_i(a)\} = E\{Z_i(a)\}^2 = I(f_0)^{-1}$ for $a=0,1$. Thus, we can decompose $\{\pi(1-\pi)I(f_0)\}^{-1}$ as follows:
     $$\begin{aligned}
     &\{\pi(1-\pi)I(f_0)\}^{-1} \\
         = &\frac{\Var\{Z_i(1)\}}{\pi} +\frac{\Var\{Z_i(0)\}}{1-\pi}\\
         =& \frac{\Var\{E(Z_i(1) \mid S_i)\}}{\pi} + \frac{\Var\{E(Z_i(0) \mid S_i)\}}{1-\pi}+ \frac{E [\Var \{ Z_i(1) \mid S_i \} ] }{\pi} + \frac{E [\Var \{ Z_i(0) \mid S_i \} ] }{1 - \pi}.
     \end{aligned}$$
     Since 
     $
     \Var \{\tilde Z_i(a) \} = \Var [ Z_i(a) - E\{Z_i(a) \mid S_i \} ] =E [ \Var\{ Z_i(a) \mid S_i  \} ], \ a=0,1,
     $
     then $V_Z^2 = E [\Var \{ Z_i(1) \mid S_i \} ] / \pi +E [\Var \{ Z_i(0) \mid S_i \} ] / (1 - \pi) $. Thus,
     $$
     \{\pi(1-\pi)I(f_0)\}^{-1} = V_Z^2 + \frac{\Var\{E(Z_i(1) \mid S_i)\}}{\pi} + \frac{\Var\{E(Z_i(0) \mid S_i)\}}{1-\pi} =  V_Z^2 + \frac{E \{\check Z_i(1)\}^2}{\pi} + \frac{E \{\check Z_i(0)\}^2}{1-\pi}.
     $$
     A straightforward calculation yields 
     $$ \frac{ \check Z_i^2(1)}{\pi} + \frac{\check Z_i^2(0)}{1-\pi} = \{\check Z_i(1)-\check Z_i(0)\}^2 + \pi(1-\pi) \Big\{ \frac{\check Z_i(1)}{\pi} + \frac{\check Z_i(0)}{1-\pi}  \Big\}^2, $$
     which implies
     $$\frac{E \{\check Z_i(1)\}^2}{\pi} + \frac{E\{\check Z_i(0)\}^2}{1-\pi} = V_H^2 + E\Big[ \pi(1-\pi) \Big\{ \frac{\check Z_i(1)}{\pi} + \frac{\check Z_i(0)}{1-\pi}  \Big\}^2  \Big].
     $$
     Thus,
     $$
     \{\pi(1-\pi)I(f_0)\}^{-1} = V_Z^2 +  V_H^2 + E\Big[ \pi(1-\pi) \Big\{ \frac{\check Z_i(1)}{\pi} + \frac{\check Z_i(0)}{1-\pi}  \Big\}^2  \Big].
     $$

     Recall that $\sigma_{\dim}^2 = V_{Z}^2 + V_H^2 + V_A^2$, where  
     $$
      V_{A}^2 = E\Big[ q_{[S_i]} \Big\{ \frac{\check Z_i(1)}{\pi} + \frac{\check Z_i(0)}{1-\pi}  \Big\}^2  \Big].
     $$
     Since $0\leq q_{[k]} \leq \pi(1-\pi)$ for $k=1,...,K$, we have $\sigma_{\dim}^2 \leq \{\pi(1-\pi)I(f_0)\}^{-1}$ and the equality holds if $q_{[k]} = \pi(1-\pi)$ for  $k=1,\ldots,K$.
 \end{proof}

 \subsection{Proof of Theorem \ref{t2}}

We need Lemma~\ref{lem:mean-conv} below to prove Theorem \ref{t2}. The proof of Lemma~\ref{lem:mean-conv} is provided in Section~\ref{sec:lemma}. Let $\Bar{Z}_a = n_a^{-1} \sum_{i:A_i=a} Z_i$ and $\Bar{Z}_{[k]a} = n_{[k]a}^{-1} \sum_{i:S_i=k, A_i=a} Z_i$ for $a =0,1$.
\begin{lemma}\label{lem:mean-conv}
    Under Assumptions~\ref{a1}--\ref{a3}, and \ref{cond:second-moment}, we have $\Bar{\hat{Z}}_{[k]a}-\Bar{Z}_{[k]a} = \op$ and  $n_{[k]a}^{-1} \sum_{i:S_i=k, A_i=a} (\hat{Z}_i^2 - Z_i^2) = \op$ for $a=0,1$.
\end{lemma}
 
\begin{proof}[Proof of Theorem~\ref{t2}]

        Recall that $Z_i(a) = -I(f_a)^{-1} \cdot (f_a^{\prime}/f_a)(Y_i(a))$ ($a=0,1$) has finite second-order moment $I(f_0)^{-1}$, and Assumptions \ref{a1}--\ref{a4} hold if we replace $Y_i$ with $Z_i$. Define 
        $$
        \begin{aligned}
            \tilde{V}_{Z}^2 &= \frac{1}{\pi}\sum_{k=1}^K p_{n[k]}\Big\{\frac{1}{n_{[k]1}}\sumis A_i(Z_i-\Bar{Z}_{[k]1})^2 \Big\} + \frac{1}{1-\pi}\sum_{k=1}^K p_{n[k]}\Big\{ \frac{1}{n_{[k]0}}\sumis (1-A_i)(Z_i-\Bar{Z}_{[k]0})^2 \Big\},\\
            \tilde{V}_{H}^2 &= \sum_{k=1}^K p_{n[k]} \left\{(\Bar{Z}_{[k]1}-\Bar{Z}_{1})-(\Bar{Z}_{[k]0}-\Bar{Z}_{0}) \right\}^2,\\
            \tilde{V}_{A}^2 &=\sum_{k=1}^K p_{n[k]}q_{[k]}\Big\{\frac{(\Bar{Z}_{[k]1}-\Bar{Z}_{1})}{\pi}+\frac{(\Bar{Z}_{[k]0}-\Bar{Z}_{0})}{1-\pi} \Big\}^2.
        \end{aligned}
        $$

        Applying Proposition 1 in \cite{Ma2020Regression} to $Z_i(a)$, we have $\tilde{V}_{A}^2 - V_{A}^2 \cp 0 $ under Assumptions~\ref{a1}, \ref{a2}, and \ref{a4}. Applying Theorems 3.2 and 3.3 in \cite{Bugni2019} to $Z_i(a)$, we have $\tilde{V}_{Z}^2 - V_{Z}^2 \cp 0$ and $\tilde{V}_{H}^2 - V_{H}^2 \cp 0$ under Assumptions~\ref{a1}--\ref{a3}. The only difference between $(\hat{V}_{Z}^2,\hat{V}_{H}^2,\hat{V}_{A}^2)$ and $(\tilde{V}_{Z}^2,\tilde{V}_{H}^2,\tilde{V}_{A}^2)$ is the substitution of $\hat{Z}_i$ with $Z_i$.

        Lemma~\ref{lem:mean-conv} implies that
        $$\begin{aligned}
            \hat{V}_{Z}^2 =& \frac{1}{\pi}\sum_{k=1}^K p_{n[k]}\Big\{\frac{1}{n_{[k]1}}\sum_{S_i=k, A_i=1} (Z_i-\Bar{Z}_{[k]1})^2 + o_P(1) \Big\}\\
            &+ \frac{1}{1-\pi}\sum_{k=1}^K p_{n[k]}\Big\{ \frac{1}{n_{[k]0}}\sum_{S_i=k,A_i=0}(Z_i-\Bar{Z}_{[k]0})^2 + o_P(1)\Big\} = \tilde{V}_{Z}^2 +o_P(1),\\
            \hat{V}_{H}^2 =& \sum_{k=1}^K p_{n[k]} \left\{(\Bar{Z}_{[k]1}-\Bar{Z}_{1})-(\Bar{Z}_{[k]0}-\Bar{Z}_{0}) +o_P(1)\right\}^2 = \tilde{V}_{H}^2 + o_P(1),\\
            \hat{V}_{A}^2 =&\sum_{k=1}^K p_{n[k]}q_{[k]}\Big\{\frac{(\Bar{Z}_{[k]1}-\Bar{Z}_{1})+o_P(1)}{\pi}+\frac{(\Bar{Z}_{[k]0}-\Bar{Z}_{0})+o_P(1)}{1-\pi} \Big\}^2 =\tilde{V}_{A}^2+o_P(1).
        \end{aligned}$$
       Here, the $o_P(1)$ term can be moved outside because $\tilde{V}_{Z}^2$, $\tilde{V}_{H}^2$ and $\tilde{V}_{A}^2$ are all bounded in probability. Therefore, $ \hat{V}_{Z}^2$, $\hat{V}_{H}^2$, $\hat{V}_{A}^2$ are all consistent estimators, which implies that $\hat \sigma^2_\dim = \hat{V}_{Z}^2  + \hat{V}_{H}^2 + \hat{V}_{A}^2$ converges in probability to $\sigma^2_\dim = {V}_{Z}^2  + {V}_{H}^2 + {V}_{A}^2$.

        \end{proof}

	\subsection{Proof of Theorem \ref{t4}}

Before proving the result, we introduce a lemma, with its proof provided in Section~\ref{sec:lemma}.

 \begin{lemma}
    \label{l2:str}
    Under Assumptions \ref{a1}--\ref{a3} and \ref{cond:second-moment}, we have 
    $$\frac{1}{\sqrt{n}}\sum_{k=1}^K\sum_{i \in \mathcal{L}_1, S_i=k} \psi_{\hat{f}_{0(2)}}^{(k)}(A_i, Y_i ; \tau) = \frac{1}{\sqrt{n}}\sum_{k=1}^K\sum_{i \in \mathcal{L}_1, S_i=k} \psi_{f_0}^{(k)}(A_i, Y_i ; \tau) + o_P(1), $$
$$\frac{1}{\sqrt{n}}\sum_{k=1}^K\sum_{i \in \mathcal{L}_2, S_i=k} \psi_{\hat{f}_{0(1)}}^{(k)}(A_i, Y_i ; \tau) = \frac{1}{\sqrt{n}}\sum_{k=1}^K\sum_{i \in \mathcal{L}_2, S_i=k} \psi_{f_0}^{(k)}(A_i, Y_i ; \tau) + o_P(1).$$
\end{lemma}

 \begin{proof}[Proof of Theorem~\ref{t4}]	
	
	By definition and the first order Taylor expansion, we have
	$$ \begin{aligned}
		\hat{\tau}_{\str}-\tau &= \tilde{\tau}-\tau+\sum_{k=1}^{K}\frac{1}{n}\left(\sum_{\substack{i \in \mathcal{L}_1, S_i=k}} \psi_{\hat{f}_{0(2)}}^{(k)}\left(A_i, Y_i ; \tilde{\tau}\right)+\sum_{\substack{i \in \mathcal{L}_2, S_i=k}} \psi_{\hat{f}_{0(1)}}^{(k)}\left(A_i, Y_i ; \tilde{\tau}\right)\right) \\
		&= \tilde{\tau}-\tau+\sum_{k=1}^{K}\frac{1}{n}\left(\sum_{\substack{i \in \mathcal{L}_1, S_i=k}} \psi_{\hat{f}_{0(2)}}^{(k)}\left(A_i, Y_i ; \tau\right)+\sum_{\substack{i \in \mathcal{L}_2, S_i=k}} \psi_{\hat{f}_{0(1)}}^{(k)}\left(A_i, Y_i ; \tau\right)\right) \\  & \quad+ \sum_{k=1}^{K}\frac{1}{n} (\tilde{\tau}-\tau)\left(\sum_{\substack{i \in \mathcal{L}_1, S_i=k}} \frac{\partial}{\partial \tau} \psi_{\hat f_{0(2)}}^{(k)}\left(A_i, Y_i ; \bar{\tau}\right)+\sum_{\substack{i \in \mathcal{L}_2, S_i=k}} \frac{\partial}{\partial \tau} \psi_{\hat f_{0(1)}}^{(k)}\left(A_i, Y_i ; \bar{\tau}\right)\right),
	\end{aligned}$$
 where $\bar{\tau}$ is between $\tilde{\tau}$ and $\tau$.
	
	By Lemma \ref{l2:str}, we have
	$$
 \begin{aligned} 
 &\hat{\tau}_{\str}-\tau \\
  = &\frac{1}{n}\sum_{k=1}^{K}\sum_{\substack{S_i=k}} \psi_{f_0}^{(k)}(A_i, Y_i ; \tau) + o_P\Big(\frac{1}{\sqrt{n}}\Big) \\
    & +(\tilde{\tau}-\tau)\left[1+\frac{1}{n}\sum_{k=1}^{K}\left(\sum_{\substack{i \in \mathcal{L}_1, S_i=k}} \frac{\partial}{\partial \tau} \psi_{\hat f_{0(2)}}^{(k)}\left(A_i, Y_i ; \bar{\tau}\right)+\sum_{\substack{i \in \mathcal{L}_2, S_i=k}} \frac{\partial}{\partial \tau} \psi_{\hat f_{0(1)}}^{(k)}\left(A_i, Y_i ; \bar{\tau}\right)\right)\right]\\
		=& B_3 + (\tilde{\tau}-\tau)(1+B_4) + o_P\Big(\frac{1}{\sqrt{n}}\Big), \end{aligned}$$
	
	where 
	$$ \begin{aligned} B_3 =&  \frac{1}{n}\sum_{i=1}^{n}\psi_{f_0}^{(k)}(A_i, Y_i ; \tau), \\
		B_4 =& \frac{1}{n} \sum_{k=1}^{K}\left(\sum_{\substack{i \in \mathcal{L}_1, S_i=k}} \frac{\partial}{\partial \tau} \psi_{\hat f_{0(2)}}^{(k)}\left(A_i, Y_i ; \bar{\tau}\right)+\sum_{\substack{i \in \mathcal{L}_2, S_i=k}} \frac{\partial}{\partial \tau} \psi_{\hat f_{0(1)}}^{(k)}\left(A_i, Y_i ; \bar{\tau}\right)\right). \end{aligned}$$

Recall that $Z_i(a) = -I(f_a)^{-1} \cdot (f_a^{\prime}/f_a)(Y_i(a))$, $a=0,1$. For $Y_i$'s satisfying Assumption \ref{a1}--\ref{a3}, it is straightforward to verify that these assumptions still hold if we replace $Y_i$ with $Z_i$. By definition,
 $$\begin{aligned}
     B_3 &= -\frac{1}{n}\sum_{k=1}^K\sumis \frac{A_i}{\pi_{n[k]}} \cdot \frac{1}{I\left(f_0\right)}\cdot \frac{f_0^{\prime}}{f_0}(Y_i-\tau) + \frac{1}{n}\sum_{k=1}^K\sumis \frac{1-A_i}{1-\pi_{n[k]}}\cdot \frac{1}{I\left(f_0\right)}\cdot \frac{f_0^{\prime}}{f_0}(Y_i) \\
      &= \frac{1}{n}\sum_{k=1}^K\sumis \frac{A_iZ_i}{\pi_{n[k]}} -  \frac{1}{n}\sum_{k=1}^K\sumis \frac{(1-A_i)Z_i}{1-\pi_{n[k]}},
 \end{aligned}$$
 which is the fully saturated linear regression estimator proposed by \cite{Bugni2019}, applied to $Z_i$. Recall that $E\{Z_i(1)\} = E\{Z_i(0)\} = 0$. By Theorem 3.1 in \cite{Bugni2019}, we have
	$$ \sqrt{n} B_3 \stackrel{d}{\rightarrow} N(0, V_{Z}^2 + V_{H}^2). $$

 Next, we show that $B_4 = -1+o_P(1)$. By the proof of Theorem \ref{t1}, we have $B_2 = -1 +o_P(1)$ (note that this conclusion still holds under the weaker Assumption \ref{a3}). Moreover,
        $$\begin{aligned}
	B_4-B_2  = & \frac{1}{n}\sum_{k=1}^{K}\left(\sum_{\substack{i \in \mathcal{L}_1, S_i=k}} \frac{\partial}{\partial \tau} \psi_{\hat f_{0(2)}}^{(k)}\left(A_i, Y_i ; \bar{\tau}\right)+\sum_{\substack{i \in \mathcal{L}_2, S_i=k}} \frac{\partial}{\partial \tau} \psi_{\hat f_{0(1)}}^{(k)}\left(A_i, Y_i ; \bar{\tau}\right)\right) \\
  &- \frac{1}{n}\sum_{k=1}^{K}\left(\sum_{\substack{i \in \mathcal{L}_1, S_i=k}} \frac{\partial}{\partial \tau} \psi_{\hat f_{0(2)}}\left(A_i, Y_i ; \bar{\tau}\right)+\sum_{\substack{i \in \mathcal{L}_2, S_i=k}} \frac{\partial}{\partial \tau} \psi_{\hat f_{0(1)}}\left(A_i, Y_i ; \bar{\tau}\right)\right). \\
	\end{aligned}$$
       By symmetry, we only need to show that 
       $$\frac{1}{n}\sum_{k=1}^{K}\sum_{\substack{i \in \mathcal{L}_1, S_i=k, A_i=0}} \frac{\partial}{\partial \tau} \psi_{\hat f_{0(2)}}^{(k)}\left(A_i, Y_i ; \bar{\tau}\right) - \frac{1}{n}\sum_{k=1}^{K}\sum_{\substack{i \in \mathcal{L}_1, S_i=k, A_i=0}} \frac{\partial}{\partial \tau} \psi_{\hat f_{0(2)}}\left(A_i, Y_i ; \bar{\tau}\right) = o_P(1),$$
    which is implied by, for any ${\delta}_n = O_P(n^{-1/2})$, 
    \begin{equation}
    \label{e:str-continuity}
    \frac{1}{n}\sum_{k=1}^{K}\Bigg[ p_{n[k]}\bigg(\frac{1}{1-\pi_{n[k]}}-\frac{1}{1-\pi}\bigg)\cdot \sum_{\substack{i \in \mathcal{L}_1, S_i=k, A_i=0}} \bigg(\frac{\hat{f}_{0(2)}^{\prime}}{\hat{f}_{0(2)}}\bigg)^{\prime}\left(Y_i+{\delta}_n\right)\Bigg] = o_P(1). 
    \end{equation}

   By Assumption \ref{cond:second-moment}(iii), we have $n^{-1} \sum_{\substack{i \in \mathcal{L}_1, S_i=k, A_i=0}} (\hat{f}_{0(2)}^{\prime}/ \hat{f}_{0(2)})^{\prime}(Y_i+{\delta}_n)$ is bounded in probability (conditional on $(S^{(n)}, A^{(n)}, U^{(n)}, D^{(\mathcal{L}_2)})$, $\{Y_i\}_{i \in \mathcal{L}_1, S_i=k, A_i=0}$ are i.i.d.). Recall that $\pi_{n[k]} \cp \pi$, we have (\ref{e:str-continuity}) holds. Thus $B_4 = -1+o_P(1)$. 
   
	Note that $\sqrt{n}(\tilde{\tau}-\tau)=O_P(n^{-1/2})$, we have $ (\tilde{\tau}-\tau)(1+B_4) = o_P(n^{-1/2})$. By Slutsky's theorem,
	$$ \sqrt{n}(\hat{\tau}_{\str}-\tau) \stackrel{d}{\rightarrow} N(0, V_{Z}^2 + V_{H}^2).  $$

By the proof of Theorem~\ref{t2}, $\hat{V}_{Z}^2$ and $\hat{V}_{H}^2$ are consistent estimators of $V_{Z}^2$ and $V_{H}^2$ under Assumptions~\ref{a1}--\ref{a3}. Thus, $\hat \sigma^2_{\str} = \hat V_{Z}^2 + \hat V_{H}^2 \cp \sigma^2_{\str} = V_{Z}^2 + V_{H}^2$ under Assumptions~\ref{a1}--\ref{a3}.
     \end{proof}

\section{Proof of lemmas}
\label{sec:lemma}

\subsection{Proof of Lemma~\ref{l1}}
\begin{proof}
    By definition,
$$
\hat{I}(f_0)=\frac{1}{n_0}\bigg[\sum_{i \in \mathcal{L}_1, A_i=0}\bigg\{\frac{\hat{f}_{0(2)}^{\prime}}{\hat{f}_{0(2)}}(Y_i)\bigg\}^2 + \sum_{i \in \mathcal{L}_2, A_i=0}\bigg\{\frac{\hat{f}_{0(1)}^{\prime}}{\hat{f}_{0(1)}}(Y_i)\bigg\}^2\bigg].
$$

It holds that
$$\begin{aligned}
    &\frac{1}{n_0}\sum_{i \in \mathcal{L}_1, A_i=0}\bigg\{\frac{\hat{f}_{0(2)}^{\prime}}{\hat{f}_{0(2)}}(Y_i)\bigg\}^2 - \frac{I(f_0)}{2}\\
    =& \left[\frac{1}{n_0}\sum_{i \in \mathcal{L}_1, A_i=0}\bigg\{\frac{\hat{f}_{0(2)}^{\prime}}{\hat{f}_{0(2)}}(Y_i)\bigg\}^2 - \sum_{k=1}^K \frac{n_{[k]0(1)}}{n_0} \int \left(\frac{\hat{f}_{0(2)}^{\prime}}{\hat{f}_{0(2)}}\right)^2(y) d F_{[k]0}(y)\right] \\
    &+ \sum_{k=1}^K \left(\frac{n_{[k]0(1)}}{n_0}-\frac{p_{[k]}}{2}\right) \int \left(\frac{\hat{f}_{0(2)}^{\prime}}{\hat{f}_{0(2)}}\right)^2(y) d F_{[k]0}(y) \\
    &+ \left[\sum_{k=1}^K \frac{p_{[k]}}{2} \int \left\{ \bigg(\frac{\hat{f}_{0(2)}^{\prime}}{\hat{f}_{0(2)}}\bigg)^2 - \bigg(\frac{f_{0}^{\prime}}{f_{0}}\bigg)^2\right\}(y) d F_{[k]0}(y)\right]\\
    := & B_1+B_2+B_3.
\end{aligned}$$

For the first term $B_1$, since the splitting is independent of $S^{(n)}, A^{(n)}, Y^{(n)}(a)=\{Y_1(a),\ldots,Y_n(a)\}$, for $a=0,1$, we have
$$E\left[ \frac{1}{n_0} \sum_{i \in \mathcal{L}_1, A_i=0} \bigg\{\frac{\hat{f}_{0(2)}^{\prime}}{\hat{f}_{0(2)}}(Y_i)\bigg\}^2  \mid S^{(n)}, A^{(n)}, U^{(n)}, D^{(\mathcal{L}_2)}  \right] = \sumk \frac{n_{[k]0(1)}}{n_0} \int \left(\frac{\hat{f}_{0(2)}^{\prime}}{\hat{f}_{0(2)}}\right)^2(y) d F_{[k]0}(y).
$$
By the triangle inequality, we have
\begin{eqnarray}
    &&  \left|  \bigg\{ \int \left(\frac{\hat{f}_{0(2)}^{\prime}}{\hat{f}_{0(2)}}\right)^2(y) d F_{[k]0}(y) \bigg\}^{1/2} -  \bigg\{ \int \left(\frac{{f}_{0}^{\prime}}{{f}_{0}}\right)^2(y) d F_{[k]0}(y)  \bigg\}^{1/2} \right|^2 \nonumber \\
    &\leq &  \int \left(\frac{\hat{f}_{0(2)}^{\prime}}{\hat{f}_{0(2)}} - \frac{{f}_{0}^{\prime}}{{f}_{0}} \right)^2(y) d F_{[k]0}(y) \nonumber \\
    &\leq& \frac{1}{\min_{k=1,\ldots,K}\ps} \int \left(\frac{\hat{f}_{0(2)}^{\prime}}{\hat{f}_{0(2)}} - \frac{{f}_{0}^{\prime}}{{f}_{0}} \right)^2(y) d F_{0}(y) = \op, \nonumber
\end{eqnarray}
where the second inequality is due to
$
f_0(y) = \sumk \ps f_{[k]0}(y) \geq (\min_{k'=1,\ldots,K} p_{[k']}) \cdot f_{[k]0}(y), 
$ for all $k=1,\ldots,K,$ 
and the last equality is due to Assumption~\ref{cond:second-moment}(iii). Combined with $ \{ \int ({f}_{0}^{\prime}/{f}_{0})^2(y) d F_{[k]0}(y) \}^{1/2} = O_P(1)$, we have 
\begin{equation}\label{eqn:OP}
\int \left(\frac{\hat{f}_{0(2)}^{\prime}}{\hat{f}_{0(2)}}\right)^2(y) d F_{[k]0}(y) = \int \left(\frac{{f}_{0}^{\prime}}{{f}_{0}}\right)^2(y) d F_{[k]0}(y) + \op.
\end{equation}
Therefore, by the weak law of large numbers, we have $\{B_1 \mid  S^{(n)}, A^{(n)}, U^{(n)}, D^{(\mathcal{L}_2)} \}  = \op$, which  implies $B_1 = \op $.

For the second term $B_2$, by \eqref{eqn:OP} and $I(f_0) < \infty$, we have 
\begin{equation}
\int \left(\frac{\hat{f}_{0(2)}^{\prime}}{\hat{f}_{0(2)}}\right)^2(y) d F_{[k]0}(y) = O_P(1). \nonumber
\end{equation}
Since ${n_{[k]0(1)}}/{n_0}-{p_{[k]}}/{2} = \op$, then $B_2 = \op$.

For the third term $B_3$, by Assumption~\ref{cond:second-moment}(iii) and the fact that $F_0(y) = \sumk \ps F_{[k]0}(y)$, we have
\begin{equation}
    \label{e:int-square-dif}
    B_3 = 2^{-1} \int \left\{ \bigg(\frac{\hat{f}_{0(2)}^{\prime}}{\hat{f}_{0(2)}}\bigg)^2 - \bigg(\frac{f_{0}^{\prime}}{f_{0}}\bigg)^2\right\}(y) d F_{0}(y) = \op.
\end{equation}

Combining them together, we have
$$
\frac{1}{n_0}\sum_{i \in \mathcal{L}_1, A_i=0}\bigg\{\frac{\hat{f}_{0(2)}^{\prime}}{\hat{f}_{0(2)}}(Y_i)\bigg\}^2 = \frac{I(f_0)}{2} + \op.
$$
Similarly,
$$
\frac{1}{n_0}\sum_{i \in \mathcal{L}_2, A_i=0}\bigg\{\frac{\hat{f}_{0(1)}^{\prime}}{\hat{f}_{0(1)}}(Y_i)\bigg\}^2 = \frac{I(f_0)}{2} + \op.
$$
Therefore, $\hat I(f_0) = I(f_0) + \op$.

\end{proof}

\subsection{Proof of Lemma~\ref{l2}}

\begin{proof}
We will prove the result for the first claim ($i \in \mathcal{L}_1$) only, as the proof for the second claim is similar. By Lemma~\ref{l1}, we have $\hat I(f_0) = I(f_0) + \op$. Given that $0< I(f_0) < \infty$, it suffices to show that
\begin{equation}\label{eqn:suf}
\hat{I}(f_0) \cdot \frac{1}{\sqrt{n}}\sum_{i \in \mathcal{L}_1} \psi_{\hat{f}_{0(2)}}(A_i, Y_i ; \tau) - \hat{I}(f_0) \cdot \frac{1}{\sqrt{n}}\sum_{i \in \mathcal{L}_1} \psi_{f_0}(A_i, Y_i ; \tau) = \op.
\end{equation}

Note that
$$\begin{aligned}
    &\hat{I}(f_0) \cdot \frac{1}{\sqrt{n}}\sum_{i \in \mathcal{L}_1} \psi_{\hat{f}_{0(2)}}(A_i, Y_i ; \tau) - \hat{I}(f_0) \cdot \frac{1}{\sqrt{n}}\sum_{i \in \mathcal{L}_1} \psi_{f_0}(A_i, Y_i ; \tau) \\
    =& \left[\hat{I}(f_0) \cdot \frac{1}{\sqrt{n}}\sum_{i \in \mathcal{L}_1} \psi_{\hat{f}_{0(2)}}(A_i, Y_i ; \tau) - I(f_0) \cdot \frac{1}{\sqrt{n}}\sum_{i \in \mathcal{L}_1} \psi_{f_0}(A_i, Y_i ; \tau)\right] \\
    &+  \left[I(f_0) \cdot \frac{1}{\sqrt{n}}\sum_{i \in \mathcal{L}_1} \psi_{f_0}(A_i, Y_i ; \tau)-\hat{I}(f_0) \cdot \frac{1}{\sqrt{n}}\sum_{i \in \mathcal{L}_1} \psi_{f_0}(A_i, Y_i ; \tau)\right] \\
    :=& B_1 + B_2.
\end{aligned}$$
It suffices for \eqref{eqn:suf} to show that $B_1= \op$ and $B_2 = \op$.

For the first term $B_1$, by definition, we have
$$
\hat{I}(f_0) \sum_{i \in \mathcal{L}_1} \psi_{\hat{f}_{0(2)}}(A_i, Y_i ; \tau) = \sum_{i \in \mathcal{L}_1, A_i =1} \frac{1}{\pi} \cdot\left(\frac{\hat{f}_{0(2)}^{\prime}}{\hat{f}_{0(2)}}\right)\left(Y_i-\tau\right)-\sum_{i \in \mathcal{L}_1, A_i=0}\frac{1}{1-\pi} \cdot\left(\frac{\hat{f}_{0(2)}^{\prime}}{\hat{f}_{0(2)}}\right)\left(Y_i\right),
$$
$$
I(f_0) \sum_{i \in \mathcal{L}_1} \psi_{f_0}(A_i, Y_i ; \tau) =  \sum_{i \in \mathcal{L}_1, A_i =1} \frac{1}{\pi} \cdot\left(\frac{{f}_{0(2)}^{\prime}}{{f}_{0(2)}}\right)\left(Y_i-\tau\right)-\sum_{i \in \mathcal{L}_1, A_i=0}\frac{1}{1-\pi} \cdot\left(\frac{{f}_{0(2)}^{\prime}}{{f}_{0(2)}}\right)\left(Y_i\right).
$$

Thus, 
\begin{eqnarray}
    && \hat{I}(f_0) \sum_{i \in \mathcal{L}_1} \psi_{\hat{f}_{0(2)}}(A_i, Y_i ; \tau) - I(f_0) \sum_{i \in \mathcal{L}_1} \psi_{f_0}(A_i, Y_i ; \tau) \nonumber \\
   & = &  \sum_{i \in \mathcal{L}_1, A_i =1} \frac{1}{\pi} \cdot\left(\frac{\hat{f}_{0(2)}^{\prime}}{\hat{f}_{0(2)}} - \frac{{f}_{0(2)}^{\prime}}{{f}_{0(2)}} \right)\left(Y_i-\tau\right)-\sum_{i \in \mathcal{L}_1, A_i=0}\frac{1}{1-\pi} \cdot\left(\frac{\hat{f}_{0(2)}^{\prime}}{\hat{f}_{0(2)}} - \frac{{f}_{0(2)}^{\prime}}{{f}_{0(2)}} \right)\left(Y_i\right). \nonumber
\end{eqnarray}

Then, by Assumptions~\ref{a1}--\ref{a2}, we have
\begin{equation}
    \label{e:lemma2}
    \begin{aligned}
	&\Var\bigg(\bigg\{\frac{1}{\sqrt{n}} \sum_{i \in \mathcal{L}_1}\left(\hat{I}\left(f_{0}\right) \cdot \psi_{\hat{f}_{0(2)}}\left(A_i, Y_i ; \tau\right)-I\left(f_0\right) \cdot \psi_{f_0}\left(A_i, Y_i ; \tau\right)\right) \bigg\}\mid S^{(n)},A^{(n)}, U^{(n)},D^{(\mathcal{L}_2)}\bigg)\\
		=&\frac{1}{n}\Var\bigg(
		\bigg\{ \sum_{i \in \mathcal{L}_1, A_i =1} \frac{1}{\pi} \cdot\left(\frac{\hat{f}_{0(2)}^{\prime}}{\hat{f}_{0(2)}}-\frac{f_0^{\prime}}{f_0}\right)\left(Y_i-\tau\right)-\sum_{i \in \mathcal{L}_1, A_i=0}\frac{1}{1-\pi} \cdot\left(\frac{\hat{f}_{0(2)}^{\prime}}{\hat{f}_{0(2)}}-\frac{f_0^{\prime}}{f_0}\right)\left(Y_i\right)\bigg\} \\
		& \mid S^{(n)},A^{(n)},U^{(n)},D^{(\mathcal{L}_2)} \bigg)\\
    =&\frac{1}{n}\Var\bigg(
		\bigg\{ \sum_{k=1}^K\sum_{i \in \mathcal{L}_1,S_i=k, A_i =1} \frac{1}{\pi} \cdot\left(\frac{\hat{f}_{0(2)}^{\prime}}{\hat{f}_{0(2)}}-\frac{f_0^{\prime}}{f_0}\right)\left(Y_i-\tau\right) \\
& -\sum_{k=1}^K\sum_{i \in \mathcal{L}_1,S_i=k, A_i =0}\frac{1}{1-\pi} \cdot\left(\frac{\hat{f}_{0(2)}^{\prime}}{\hat{f}_{0(2)}}-\frac{f_0^{\prime}}{f_0}\right)\left(Y_i\right)\bigg\}^2  \mid S^{(n)},A^{(n)},U^{(n)},D^{(\mathcal{L}_2)} \bigg)\\
    \leq& \frac{2K}{n}\sum_{k=1}^K \Bigg[ \Var\bigg(\sum_{i \in \mathcal{L}_1,S_i=k, A_i =1} \frac{1}{\pi} \cdot\left(\frac{\hat{f}_{0(2)}^{\prime}}{\hat{f}_{0(2)}}-\frac{f_0^{\prime}}{f_0}\right)\left(Y_i-\tau\right) \mid S^{(n)},A^{(n)},U^{(n)},D^{(\mathcal{L}_2)} \bigg)\\
    &+ \Var\bigg(\sum_{i \in \mathcal{L}_1,S_i=k, A_i =0} \frac{1}{1-\pi} \cdot\left(\frac{\hat{f}_{0(2)}^{\prime}}{\hat{f}_{0(2)}}-\frac{f_0^{\prime}}{f_0}\right)\left(Y_i\right) \mid S^{(n)},A^{(n)},U^{(n)},D^{(\mathcal{L}_2)}\bigg)\Bigg].
\end{aligned}
\end{equation}
Moreover,
\begin{equation}
\label{E_singlek_lemma2}
    \begin{aligned}
    &\Var\bigg(\sum_{i \in \mathcal{L}_1,S_i=k, A_i =0} \frac{1}{1-\pi} \cdot\left(\frac{\hat{f}_{0(2)}^{\prime}}{\hat{f}_{0(2)}}-\frac{f_0^{\prime}}{f_0}\right)\left(Y_i\right) \mid S^{(n)},A^{(n)},U^{(n)},D^{(\mathcal{L}_2)} \bigg)\\
    =&\frac{1}{(1-\pi)^2} \sum_{i \in \mathcal{L}_1,S_i=k, A_i =0}\Var\bigg\{\left(\frac{\hat{f}_{0(2)}^{\prime}}{\hat{f}_{0(2)}}-\frac{f_0^{\prime}}{f_0}\right)\left(Y_i\right) \mid S^{(n)},A^{(n)},U^{(n)},D^{(\mathcal{L}_2)} \bigg\}\\
    \leq&\frac{1}{(1-\pi)^2} \sum_{i \in \mathcal{L}_1,S_i=k, A_i =0}E\bigg\{\left(\frac{\hat{f}_{0(2)}^{\prime}}{\hat{f}_{0(2)}}-\frac{f_0^{\prime}}{f_0}\right)^2\left(Y_i\right) \mid S^{(n)},A^{(n)},U^{(n)},D^{(\mathcal{L}_2)} \bigg\} \\
    =& \frac{n_{[k]0(1)}}{(1-\pi)^2}\int \left(\frac{\hat{f}_{0(2)}^{\prime}}{\hat{f}_{0(2)}}-\frac{f_0^{\prime}}{f_0}\right)^2\left(y\right) d F_{[k]0}(y) \\
    =&o_P(n),
\end{aligned}
\end{equation}
where the last equality is due to $n_{[k]0(1)}/n \cp \ps (1-\pi)/2$ and 
\begin{eqnarray}\label{eqn:FkOp}
    \int \left(\frac{\hat{f}_{0(2)}^{\prime}}{\hat{f}_{0(2)}}-\frac{f_0^{\prime}}{f_0}\right)^2\left(y\right) d F_{[k]0}(y) \leq \frac{1}{ \min_{k=1,\ldots,K} \ps}  \int \left(\frac{\hat{f}_{0(2)}^{\prime}}{\hat{f}_{0(2)}}-\frac{f_0^{\prime}}{f_0}\right)^2\left(y\right) d F_{0}(y) = \op.
\end{eqnarray}
Similarly,
\begin{equation}
\label{E_singlek_lemma2_2}
    \begin{aligned}
    &\Var\bigg(\sum_{i \in \mathcal{L}_1,S_i=k, A_i =1} \frac{1}{\pi} \cdot\left(\frac{\hat{f}_{0(2)}^{\prime}}{\hat{f}_{0(2)}}-\frac{f_0^{\prime}}{f_0}\right)\left(Y_i - \tau \right) \mid S^{(n)},A^{(n)},U^{(n)},D^{(\mathcal{L}_2)} \bigg) = o_P(n),
\end{aligned}
\end{equation}
By substituting \eqref{E_singlek_lemma2} and \eqref{E_singlek_lemma2_2} into \eqref{e:lemma2}, we have
\begin{equation}
    \label{e:lemma2-var}
    \begin{aligned}
	&\Var\bigg(\bigg\{\frac{1}{\sqrt{n}} \sum_{i \in \mathcal{L}_1}\left(\hat{I}\left(f_{0}\right) \cdot \psi_{\hat{f}_{0(2)}}\left(A_i, Y_i ; \tau\right)-I\left(f_0\right) \cdot \psi_{f_0}\left(A_i, Y_i ; \tau\right)\right) \bigg\}\mid S^{(n)},A^{(n)}, U^{(n)},D^{(\mathcal{L}_2)}\bigg)\\
 =& \op.
\end{aligned}
\end{equation}

Next, we show that the conditional bias also converges to zero in probability. Specifically, we will show that
\begin{eqnarray}
\label{e:lemma2-bias}
&& E\bigg(\bigg\{\sum_{i \in \mathcal{L}_1}\left(\hat{I}\left(f_{0}\right) \cdot \psi_{\hat{f}_{0(2)}}\left(A_i, Y_i ; \tau\right)-I\left(f_0\right) \cdot \psi_{f_0}\left(A_i, Y_i ; \tau\right)\right) \bigg\}\mid S^{(n)},A^{(n)}, U^{(n)},D^{(\mathcal{L}_2)}\bigg) \nonumber \\
& = & \sum_{k=1}^K \bigg(\frac{n_{[k]1(1)}}{n\pi}-\frac{n_{[k]0(1)}}{n(1-\pi)}\bigg) \cdot \int \left(\frac{\hat{f}_{0(2)}^{\prime}}{\hat{f}_{0(2)}}-\frac{f_0^{\prime}}{f_0}\right)\left(y\right) dF_{[k]0}(y) = o_P\bigg(\frac{1}{\sqrt{n}}\bigg).
\end{eqnarray}
Under Assumptions \ref{a1}, \ref{a2}, and \ref{a4}, by \citet[][Lemma C.2]{Bugni2019}, we have $n_{[k]1} / n - \ps\pi = O_P(n^{-1/2})$. Since $(n_{[k]1}-1)/2 \leq n_{[k]1(1)} \leq n_{[k]1} /2$, we have $n_{[k]1(1)} / (n\pi) - \ps/2 = O_P(n^{-1/2})$. Similarly, $n_{[k]0(1)} / \{n(1-\pi)\} - \ps/2 = O_P(n^{-1/2})$. Moreover, by \eqref{eqn:FkOp} and the Cauchy--Schwarz inequality, we have
$
\int (\hat{f}_{0(2)}^{\prime}/\hat{f}_{0(2)}-f_0^{\prime}/f_0)\left(y\right) dF_{[k]0}(y) = o_P(1).
$
Therefore, \eqref{e:lemma2-bias} holds.

Combining \eqref{e:lemma2-var} and \eqref{e:lemma2-bias}, we have 
\begin{eqnarray*}
    &&E\bigg(\bigg\{\frac{1}{\sqrt{n}} \sum_{i \in \mathcal{L}_1}\left(\hat{I}\left(f_{0}\right) \cdot \psi_{\hat{f}_{0(2)}}\left(A_i, Y_i ; \tau\right)-I\left(f_0\right) \cdot \psi_{f_0}\left(A_i, Y_i ; \tau\right)\right) \bigg\}^2\mid S^{(n)},A^{(n)}, aU^{(n)},D^{(\mathcal{L}_2)}\bigg)\\
    &=& \op.
\end{eqnarray*}
Therefore, conditional on $S^{(n)},A^{(n)}, U^{(n)},D^{(\mathcal{L}_2)}$, $B_1$ converges to zero in probability, which implies $B_1 = \op$ unconditionally.

For the second term $B_2$, in the proof of Theorem~\ref{t1}, we have shown that 
$
n^{-1/2} \sumn \psi_{f_0}(A_i, Y_i ; \tau)
$
is asymptotically normal with zero mean and finite variance. Since the splitting is independent of $S^{(n)}, A^{(n)}, Y^{(n)}(a)$, we have $n^{-1/2} \sum_{i \in \mathcal{L}_1} \psi_{f_0}(A_i, Y_i ; \tau) = O_P(1)$. By Lemma~\ref{l1}, we have $\hat I(f_0) = I(f_0) + \op$. 
Therefore, $B_2 = \op$.

\end{proof}

\subsection{Proof of Lemma~\ref{lem:mean-conv}}

\begin{proof}
    Recall that $\hat{Z}_i$ is defined as follows: $\hat{Z}_i = -\hat{I}(f_0)^{-1}  (\hat{f}_{0(3-j)}^{\prime}/\hat{f}_{0(3-j)})(Y_i),$ for $i \in \mathcal{L}_j, \ A_i = 0$, $j=1,2$, and $\hat{Z}_i = -\hat{I}(f_0)^{-1}  (\hat{f}_{0(3-j)}^{\prime}/\hat{f}_{0(3-j)})(Y_i-\tilde{\tau}),$ for $i \in \mathcal{L}_j,\ A_i = 1$, $j=1,2$. Recall that $\Bar{\hat{Z}}_{[k]a} = \sum_{S_i = k, A_i=a} \hat{Z}_i / n_{[k]a}$ is the stratum-specific sample mean of $\hat{Z}_i$ in stratum $k$ under treatment arm $a$ and $\Bar{\hat{Z}}_{a} = \sum_{i:A_i=a} \hat{Z}_i / n_a$ is the sample mean of  $\hat{Z}_i$ under treatment arm $a$, $a=0,1$, $k=1,\ldots,K$.

    (I) We show that $\Bar{\hat{Z}}_{[k]0} - \Bar{Z}_{[k]0} = \op$. By definition, we have
$$
\begin{aligned}
    &\Bar{\hat{Z}}_{[k]0} - \Bar{Z}_{[k]0} \\
    = & n_{[k]0}^{-1} \cdot \hat{I}(f_0)^{-1}\bigg\{\sum_{i \in \mathcal{L}_1, S_i=k, A_i=0} \frac{\hat{f}_{0(2)}^{\prime}}{\hat{f}_{0(2)}}(Y_i) + \sum_{i \in \mathcal{L}_2, S_i=k, A_i=0} \frac{\hat{f}_{0(1)}^{\prime}}{\hat{f}_{0(1)}}(Y_i)  \bigg\}\\
    &- n_{[k]0}^{-1} \cdot I(f_0)^{-1}\sum_{i:S_i=k, A_i=0} \frac{f_{0}^{\prime}}{f_{0}}(Y_i)\\
    = &  \frac{\nscone}{\nsc} \cdot  n_{[k]0(1)}^{-1}\bigg\{\hat{I}(f_0)^{-1} \sum_{i \in \mathcal{L}_1, S_i=k, A_i=0} \frac{\hat{f}_{0(2)}^{\prime}}{\hat{f}_{0(2)}}\left(Y_i\right)-I(f_0)^{-1} \sum_{i \in \mathcal{L}_1, S_i=k, A_i=0} \frac{f_0^{\prime}}{f_0}\left(Y_i\right) \bigg\} \\
    &+ \frac{n_{[k]0(2)}}{\nsc} \cdot n_{[k]0(2)}^{-1}\bigg\{\hat{I}(f_0)^{-1} \sum_{i \in \mathcal{L}_2, S_i=k, A_i=0} \frac{\hat{f}_{0(1)}^{\prime}}{\hat{f}_{0(1)}}\left(Y_i\right)-I(f_0)^{-1} \sum_{i \in \mathcal{L}_2, S_i=k, A_i=0} \frac{f_0^{\prime}}{f_0}\left(Y_i\right)\bigg\}.
\end{aligned}
$$  
        
        Since $n_{[k]0(1)}/n_{[k]0} \cp 1/2$ and $n_{[k]0(2)}/n_{[k]0} \cp 1/2$, it suffices for $\Bar{\hat{Z}}_{[k]0}-\Bar{Z}_{[k]0} = o_P(1)$ to show that
        $$n_{[k]0(1)}^{-1}\bigg\{\hat{I}(f_0)^{-1} \sum_{i \in \mathcal{L}_1, S_i=k, A_i=0} \frac{\hat{f}_{0(2)}^{\prime}}{\hat{f}_{0(2)}}\left(Y_i\right)-I(f_0)^{-1} \sum_{i \in \mathcal{L}_1, S_i=k, A_i=0} \frac{f_0^{\prime}}{f_0}\left(Y_i\right) \bigg\}= o_P(1),$$
        $$n_{[k]0(2)}^{-1}\bigg\{\hat{I}(f_0)^{-1} \sum_{i \in \mathcal{L}_2, S_i=k, A_i=0} \frac{\hat{f}_{0(1)}^{\prime}}{\hat{f}_{0(1)}}\left(Y_i\right)-I(f_0)^{-1} \sum_{i \in \mathcal{L}_2, S_i=k, A_i=0} \frac{f_0^{\prime}}{f_0}\left(Y_i\right)\bigg\} = o_P(1).$$

        We will prove the first claim only, as the proof of the second claim is similar. We make the following decomposition:
        \begin{eqnarray}
            & &n_{[k]0(1)}^{-1}\bigg\{\hat{I}(f_0)^{-1} \sum_{i \in \mathcal{L}_1, S_i=k, A_i=0} \frac{\hat{f}_{0(2)}^{\prime}}{\hat{f}_{0(2)}}\left(Y_i\right)-I(f_0)^{-1} \sum_{i \in \mathcal{L}_1, S_i=k, A_i=0} \frac{f_0^{\prime}}{f_0}\left(Y_i\right) \bigg\} \nonumber \\
            &&= n_{[k]0(1)}^{-1}\bigg\{\hat{I}(f_0)^{-1} \sum_{i \in \mathcal{L}_1, S_i=k, A_i=0} \frac{\hat{f}_{0(2)}^{\prime}}{\hat{f}_{0(2)}}\left(Y_i\right)-\hat{I}(f_0)^{-1} \sum_{i \in \mathcal{L}_1, S_i=k, A_i=0} \frac{f_0^{\prime}}{f_0}\left(Y_i\right) \bigg\} \nonumber \\
            &&\quad + n_{[k]0(1)}^{-1}\bigg\{\hat{I}(f_0)^{-1} \sum_{i \in \mathcal{L}_1, S_i=k, A_i=0} \frac{f_0^{\prime}}{f_0}\left(Y_i\right)-I(f_0)^{-1} \sum_{i \in \mathcal{L}_1, S_i=k, A_i=0} \frac{f_0^{\prime}}{f_0}\left(Y_i\right) \bigg\} \nonumber \\
    &&:= B_1 + B_2. \nonumber
        \end{eqnarray}

        Next, we show that $B_1=\op$ and $B_2 = \op$. By Lemma \ref{l1}, $\hat{I}(f_0) = I(f_0)+o_P(1)$. Thus, $\hat{I}(f_0)^{-1}=O_P(1)$. It suffices for $B_1=\op$ to show that
        \begin{eqnarray}
            n_{[k]0(1)}^{-1} E \bigg\{\sum_{i \in \mathcal{L}_1, S_i=k, A_i=0} \left| \frac{\hat{f}_{0(2)}^{\prime}}{\hat{f}_{0(2)}} - \frac{f_{0(2)}^{\prime}}{f_{0(2)}} \right| \left(Y_i\right) \mid S^{(n)},A^{(n)},U^{(n)},D^{(\mathcal{L}_2)} \bigg\} = o_P(1). \nonumber
        \end{eqnarray}
        It holds that
       \begin{eqnarray}
            && n_{[k]0(1)}^{-1} E \bigg\{\sum_{i \in \mathcal{L}_1, S_i=k, A_i=0} \left| \frac{\hat{f}_{0(2)}^{\prime}}{\hat{f}_{0(2)}} - \frac{f_{0(2)}^{\prime}}{f_{0(2)}} \right| \left(Y_i\right) \mid S^{(n)},A^{(n)},U^{(n)},D^{(\mathcal{L}_2)} \bigg\}  \nonumber \\
            &=& \int \left| \frac{\hat{f}_{0(2)}^{\prime}}{\hat{f}_{0(2)}}-\frac{f_0^{\prime}}{f_0}\right| \left(y\right) d F_{[k]0}(y) \nonumber \\
            & \leq & \left\{ \int \left( \frac{\hat{f}_{0(2)}^{\prime}}{\hat{f}_{0(2)}}-\frac{f_0^{\prime}}{f_0}\right)^2 \left(y\right) d F_{[k]0}(y) \right\}^{1/2} = \op, \nonumber
        \end{eqnarray}
where the last equality is due to \eqref{eqn:FkOp}.

For $B_2 = \op$, it suffices to show that
$$
n_{[k]0(1)}^{-1} \sum_{i \in \mathcal{L}_1, S_i=k, A_i=0} \frac{f_0^{\prime}}{f_0}\left(Y_i\right) = O_P(1).
$$
By applying the Cauchy--Schwarz inequality and using $f_{[k]0}(y) \leq \min_{k=1,\ldots,K} \ps^{-1} f_{0}(y)$, we have
$$
E \bigg|\frac{f_0^{\prime}}{f_0}\left(Y_i\right)  \bigg|  = \int \bigg| \frac{f_0^{\prime}}{f_0} \left(y\right) \bigg| d F_{[k]0}(y)  \leq \frac{1}{\min_{k=1,\ldots,K} \ps} \bigg[ \int  \bigg\{ \frac{f_0^{\prime}}{f_0} \left(y\right) \bigg\}^2 d F_{0}(y) \bigg]^{1/2} < \infty.
$$
Then, we have
$$
n_{[k]0(1)}^{-1} \sum_{i \in \mathcal{L}_1, S_i=k, A_i=0} \frac{f_0^{\prime}}{f_0}\left(Y_i\right) \mid S^{(n)},A^{(n)},U^{(n)} = E \bigg\{ \frac{f_0^{\prime}}{f_0}\left(Y_i\right)  \bigg\} + \op.
$$
Thus,
$$
n_{[k]0(1)}^{-1} \sum_{i \in \mathcal{L}_1, S_i=k, A_i=0} \frac{f_0^{\prime}}{f_0}\left(Y_i\right) = E \bigg\{ \frac{f_0^{\prime}}{f_0}\left(Y_i\right)  \bigg\} + \op.
$$

(II) We show that $\Bar{\hat{Z}}_{[k]1} - \Bar{Z}_{[k]1} = \op$. Similar to the proof of (I), we have
        $$\begin{aligned}
            &n_{[k]1}^{-1}\cdot \hat{I}(f_0)^{-1}\bigg\{\sum_{i \in \mathcal{L}_1, S_i=k, A_i=1} \frac{\hat{f}_{0(2)}^{\prime}}{\hat{f}_{0(2)}}(Y_i-\tau) + \sum_{i \in \mathcal{L}_2, S_i=k, A_i=1} \frac{\hat{f}_{0(1)}^{\prime}}{\hat{f}_{0(1)}}(Y_i-\tau) \bigg\}\\
            &-n_{[k]1}^{-1}\cdot I(f_0)^{-1}\sum_{S_i=k, A_i=1} \frac{\hat{f}_{0}}{f_{0}}(Y_i-\tau) = o_P(1).
        \end{aligned}$$
        The difference arises from our lack of knowledge about $\tau$, leading us to substitute $\tilde \tau$ for $\tau$ in the definition of $\Bar{\hat{Z}}{[k]1}$. Considering this substitution, $\Bar{\hat{Z}}{[k]1}-\Bar{Z}_{[k]1} = o_P(1)$ is implied by
        $$
            \frac{n_{[k]1(1)}}{n_{[k]1}} \cdot n_{[k]1(1)}^{-1}\bigg\{\sum_{i \in \mathcal{L}_1, S_i=k, A_i=1} \frac{\hat{f}_{0(2)}^{\prime}}{\hat{f}_{0(2)}}(Y_i-\tilde{\tau}) - \sum_{i \in \mathcal{L}_1, S_i=k, A_i=1} \frac{\hat{f}_{0(2)}^{\prime}}{\hat{f}_{0(2)}}(Y_i-\tau)\bigg\} = o_P(1).
            $$
            $$
            \frac{n_{[k]1(2)}}{n_{[k]1}} \cdot n_{[k]1(2)}^{-1}\bigg\{\sum_{i \in \mathcal{L}_1, S_i=k, A_i=1} \frac{\hat{f}_{0(1)}^{\prime}}{\hat{f}_{0(1)}}(Y_i-\tilde{\tau}) - \sum_{i \in \mathcal{L}_1, S_i=k, A_i=1} \frac{\hat{f}_{0(1)}^{\prime}}{\hat{f}_{0(1)}}(Y_i-\tau)\bigg\} = o_P(1).
            $$
        We will prove the first claim only, as the proof of the second claim is similar. Note that $n_{[k]1(1)}/n_{[k]1} \xrightarrow{P} 1/2$. By the first-order Taylor expansion, the first claim is implied by
        $$\frac{1}{n_{[k]1(1)}}\sum_{i \in \mathcal{L}_1, S_i=k, A_i=1} (\tau-\tilde{\tau}) \cdot \bigg(\frac{\hat{f}_{0(2)}^{\prime}}{\hat{f}_{0(2)}}\bigg)^{\prime}\big((Y_i-\tau)+(\tau-\Bar{\tau})\big)  = o_P(1), $$
        where $\Bar{\tau}$ is between $\tilde{\tau}$ and $\tau$.

        Recall that $\sqrt{n}(\tilde{\tau}-\tau) = O_P(1)$, thus $\tau-\Bar{\tau} = O_P(n^{-1/2}) = O_P(n_{[k]1(1)}^{-1/2})$, and it suffices to show that
        $$\frac{1}{n_{[k]1(1)}}\sum_{i \in \mathcal{L}_1, S_i=k, A_i=1} \bigg(\frac{\hat{f}_{0(2)}^{\prime}}{\hat{f}_{0(2)}}\bigg)^{\prime}\big(Y_i^{*}+(\tau-\Bar{\tau})\big)  = O_P(1), $$
        where $Y_i^{*} = Y_i-\tau \sim F_{0[k]}(y)$ for $i \in \mathcal{L}_1, S_i=k, A_i=1$. The above statement follows from Assumption \ref{cond:second-moment}(iii).

(III) We show that $\Bar{\hat{Z}}_{a}-\Bar{Z}_{a} = o_P(1)$. 
By definition, $\Bar{\hat{Z}}_{a} = \sum_{i:A_i=a} \hat{Z}_i / n_a = \sumk (n_{[k]a}/n_a) \Bar{\hat{Z}}_{[k]a} $. Then
$$
\Bar{\hat{Z}}_{a}-\Bar{Z}_{a} = \sumk (n_{[k]a}/n_a) (\Bar{\hat{Z}}_{[k]a} - \Bar{Z}_{[k]a}) = \op,
$$
where the last equality is due to $n_{[k]a}/n_a \cp \ps$ and $\Bar{\hat{Z}}_{[k]a} - \Bar{Z}_{[k]a} = \op$.

(IV) We show that $n_{[k]0}^{-1} \sum_{S_i=k, A_i=0} (\hat{Z}_i^2 - Z_i^2) = o_P(1)$, which is implied by
        $$n_{[k]0(1)}^{-1}\bigg\{\hat{I}(f_0)^{-2} \sum_{i \in \mathcal{L}_1, S_i=k, A_i=0} \bigg(\frac{\hat{f}_{0(2)}^{\prime}}{\hat{f}_{0(2)}}\bigg)^2\left(Y_i\right)-I(f_0)^{-2} \sum_{i \in \mathcal{L}_1, S_i=k, A_i=0} \bigg(\frac{f_0^{\prime}}{f_0}\bigg)^2\left(Y_i\right) \bigg\}= o_P(1),$$
        $$n_{[k]0(2)}^{-1}\bigg\{\hat{I}(f_0)^{-2} \sum_{i \in \mathcal{L}_2, S_i=k, A_i=0} \bigg(\frac{\hat{f}_{0(1)}^{\prime}}{\hat{f}_{0(1)}}\bigg)^2\left(Y_i\right)-I(f_0)^{-2} \sum_{i \in \mathcal{L}_2, S_i=k, A_i=0} \bigg(\frac{f_0^{\prime}}{f_0}\bigg)^2\left(Y_i\right) \bigg\}= o_P(1).$$
We will prove the first equation ($i \in \mathcal{L}_1$) only, as the proof for the second equation ($i \in \mathcal{L}_2$) is similar. 

Similar to (I), we make the following decomposition:
\begin{eqnarray}
        \label{e:difference-hatf-scone}
            & &n_{[k]0(1)}^{-1}\bigg\{\hat{I}(f_0)^{-2} \sum_{i \in \mathcal{L}_1, S_i=k, A_i=0} \bigg(\frac{\hat{f}_{0(2)}^{\prime}}{\hat{f}_{0(2)}}\bigg)^2\left(Y_i\right)-I(f_0)^{-2} \sum_{i \in \mathcal{L}_1, S_i=k, A_i=0} \bigg(\frac{f_0^{\prime}}{f_0}\bigg)^2\left(Y_i\right) \bigg\} \nonumber \\
            &&= n_{[k]0(1)}^{-1}\bigg\{\hat{I}(f_0)^{-2} \sum_{i \in \mathcal{L}_1, S_i=k, A_i=0} \bigg(\frac{\hat{f}_{0(2)}^{\prime}}{\hat{f}_{0(2)}}\bigg)^2\left(Y_i\right)-\hat{I}(f_0)^{-2} \sum_{i \in \mathcal{L}_1, S_i=k, A_i=0} \bigg(\frac{f_0^{\prime}}{f_0}\bigg)^2\left(Y_i\right) \bigg\} \nonumber \\
            &&\quad + n_{[k]0(1)}^{-1}\bigg\{\hat{I}(f_0)^{-2} \sum_{i \in \mathcal{L}_1, S_i=k, A_i=0} \bigg(\frac{f_0^{\prime}}{f_0}\bigg)^2\left(Y_i\right)-I(f_0)^{-2} \sum_{i \in \mathcal{L}_1, S_i=k, A_i=0} \bigg(\frac{f_0^{\prime}}{f_0}\bigg)^2\left(Y_i\right) \bigg\} \nonumber \\
    &&:= B_3 + B_4. \nonumber
        \end{eqnarray}
        
By (\ref{eqn:OP}), we have $n_{[k]0(1)}^{-1}[\sum_{i \in \mathcal{L}_1, S_i=k, A_i=0} \{(\hat{f}_{0(2)}^{\prime}/ \hat{f}_{0(2)})^2-(f_0^{\prime}/f_0)^2\}(Y_i)]= o_P(1)$. By Lemma \ref{l1}, $\hat{I}(f_0) = I(f_0) +\op$, which implies $\hat{I}(f_0)^{-2}=O_P(1)$. Therefore, $B_3 = o_P(1)$.

It suffices for $B_4 = \op$ to show that
$$
n_{[k]0(1)}^{-1} \sum_{i \in \mathcal{L}_1, S_i=k, A_i=0} \bigg(\frac{f_0^{\prime}}{f_0}\bigg)^2\left(Y_i\right) = O_P(1).
$$
Note that 
$$
E \bigg(\frac{f_0^{\prime}}{f_0}\bigg)^2\left(Y_i\right)  = \int \bigg(\frac{f_0^{\prime}}{f_0}\bigg)^2 \left(y\right)  d F_{[k]0}(y)  \leq \frac{1}{\min_{k=1,\ldots,K} \ps} \int \bigg(\frac{f_0^{\prime}}{f_0}\bigg)^2 \left(y\right)   d F_{0}(y)  < \infty.
$$
Then, we have
$$
n_{[k]0(1)}^{-1} \sum_{i \in \mathcal{L}_1, S_i=k, A_i=0} \bigg(\frac{f_0^{\prime}}{f_0}\bigg)^2\left(Y_i\right) \mid S^{(n)},A^{(n)},U^{(n)} = E \bigg(\frac{f_0^{\prime}}{f_0}\bigg)^2\left(Y_i\right) + \op.
$$
Thus,
$$
n_{[k]0(1)}^{-1} \sum_{i \in \mathcal{L}_1, S_i=k, A_i=0} \bigg(\frac{f_0^{\prime}}{f_0}\bigg)^2\left(Y_i\right) = E \bigg(\frac{f_0^{\prime}}{f_0}\bigg)^2\left(Y_i\right)  + \op = O_P(1).
$$

(V) We show that $n_{[k]1}^{-1} \sum_{S_i=k, A_i=1} (\hat{Z}_i^2 - Z_i^2) = o_P(1)$. Similar to the proof of (IV), we have
        $$n_{[k]1(1)}^{-1}\bigg\{\hat{I}(f_0)^{-2} \sum_{i \in \mathcal{L}_1, S_i=k, A_i=1} \bigg(\frac{\hat{f}_{0(2)}^{\prime}}{\hat{f}_{0(2)}}\bigg)^2\left(Y_i - \tau \right)-I(f_0)^{-2} \sum_{i \in \mathcal{L}_1, S_i=k, A_i=1} \bigg(\frac{f_0^{\prime}}{f_0}\bigg)^2\left(Y_i - \tau \right) \bigg\}= o_P(1),$$
        $$n_{[k]0(2)}^{-1}\bigg\{\hat{I}(f_0)^{-2} \sum_{i \in \mathcal{L}_2, S_i=k, A_i=1} \bigg(\frac{\hat{f}_{1(1)}^{\prime}}{\hat{f}_{0(1)}}\bigg)^2\left(Y_i-\tau \right)-I(f_0)^{-2} \sum_{i \in \mathcal{L}_2, S_i=k, A_i=1} \bigg(\frac{f_0^{\prime}}{f_0}\bigg)^2\left(Y_i - \tau \right) \bigg\}= o_P(1).$$
        The difference arises from our lack of knowledge about $\tau$, leading us to substitute $\tilde \tau$ for $\tau$ in the definition of $\Bar{\hat{Z}}{[k]1}$. Considering this substitution, $\Bar{\hat{Z}}_{[k]1}-\Bar{Z}_{[k]1} = o_P(1)$ is implied by
        $$
            n_{[k]1(1)}^{-1}\bigg\{\sum_{i \in \mathcal{L}_1, S_i=k, A_i=1} \bigg(\frac{\hat{f}_{0(2)}^{\prime}}{\hat{f}_{0(2)}}\bigg)^2(Y_i-\tilde{\tau}) - \sum_{i \in \mathcal{L}_1, S_i=k, A_i=1} \bigg(\frac{\hat{f}_{0(2)}^{\prime}}{\hat{f}_{0(2)}}\bigg)^2(Y_i-\tau)\bigg\} = o_P(1),
            $$
            $$
            n_{[k]1(2)}^{-1}\bigg\{\sum_{i \in \mathcal{L}_1, S_i=k, A_i=1} \bigg(\frac{\hat{f}_{0(1)}^{\prime}}{\hat{f}_{0(1)}}\bigg)^2(Y_i-\tilde{\tau}) - \sum_{i \in \mathcal{L}_1, S_i=k, A_i=1} \bigg(\frac{\hat{f}_{0(1)}^{\prime}}{\hat{f}_{0(1)}}\bigg)^2(Y_i-\tau)\bigg\} = o_P(1).
            $$
        As before we prove the first equation ($i \in \mathcal{L}_1$) only, as the proof for the second equation ($i \in \mathcal{L}_2$) is similar.

        By the first-order Taylor expansion, it suffices to show that
        \begin{equation}
        \label{e:lemma3-V}
                n_{[k]1(1)}^{-1} \cdot (\tilde{\tau}-\tau)\sum_{i \in \mathcal{L}_1, S_i=k, A_i=1}  \bigg(\frac{\hat{f}_{0(2)}^{\prime}}{\hat{f}_{0(2)}}\bigg)(Y_i-\Bar{\tau}) \cdot \bigg(\frac{\hat{f}_{0(2)}^{\prime}}{\hat{f}_{0(2)}}\bigg)^{\prime}(Y_i-\Bar{\tau}) =o_P(1),
        \end{equation}
        where $\Bar{\tau}$ is between $\tilde{\tau}$ and $\tau$.

        By Assumption \ref{cond:second-moment}(iii), we have
        $$n_{[k]1(1)}^{-1} \sum_{i \in \mathcal{L}_1, S_i=k, A_i=1} \bigg|\bigg(\frac{\hat{f}_{0(2)}^{\prime}}{\hat{f}_{0(2)}}\bigg)^{\prime}\bigg|(Y_i-\Bar{\tau}) =O_P(1).$$
        Note that for any $i \in \mathcal{L}_1$ with $S_i=k$ and $A_i=1$, $$\bigg|\frac{\hat{f}_{0(2)}^{\prime}}{\hat{f}_{0(2)}}\bigg|(Y_i-\Bar{\tau}) \leq \sup_{y \in \mathbb R} \bigg\{\bigg|\frac{\hat{f}_{0(2)}^{\prime}}{\hat{f}_{0(2)}}\bigg|(y)\bigg\} = o_P(n^{1/2}).$$
        Combining with $\tilde{\tau}-\tau = O_P(n^{-1/2})$, we obtain \eqref{e:lemma3-V}.
    
\end{proof}

\subsection{Proof of Lemma~\ref{l2:str}}
\begin{proof}
We adopt a similar approach to the proof of Lemma \ref{l2}. As before, we will only prove the result for $i \in \mathcal{L}_1$.

By Lemma~\ref{l1}, we have $\hat I(f_0) = I(f_0) + \op$. Since $0< I(f_0) < \infty$, then it suffices to show that
\begin{equation}\label{eqn:suf-lemma3}
\hat{I}(f_0) \cdot \frac{1}{\sqrt{n}}\sum_{k=1}^K\sum_{i \in \mathcal{L}_1, S_i=k} \psi_{\hat{f}_{0(2)}}^{(k)}(A_i, Y_i ; \tau) - \hat{I}(f_0) \cdot \frac{1}{\sqrt{n}}\sum_{k=1}^K\sum_{i \in \mathcal{L}_1, S_i=k} \psi_{f_0}^{(k)}(A_i, Y_i ; \tau) = \op.
\end{equation}

Firstly, we have
\begin{equation}
\label{e:lemma3}
\begin{aligned}
    &\hat{I}(f_0) \cdot \frac{1}{\sqrt{n}}\sum_{k=1}^K\sum_{i \in \mathcal{L}_1, S_i=k} \psi_{\hat{f}_{0(2)}}^{(k)}(A_i, Y_i ; \tau) - \hat{I}(f_0) \cdot \frac{1}{\sqrt{n}}\sum_{k=1}^K\sum_{i \in \mathcal{L}_1, S_i=k} \psi_{f_0}^{(k)}(A_i, Y_i ; \tau) \\
    =&  \left[\hat{I}(f_0) \cdot \frac{1}{\sqrt{n}}\sum_{k=1}^K\sum_{i \in \mathcal{L}_1, S_i=k}\psi_{\hat{f}_{0(2)}}^{(k)}(A_i, Y_i ; \tau) - I(f_0) \cdot \frac{1}{\sqrt{n}}\sum_{k=1}^K\sum_{i \in \mathcal{L}_1, S_i=k} \psi_{f_0}^{(k)}(A_i, Y_i ; \tau)\right] \\
    &+ \left[I(f_0) \cdot \frac{1}{\sqrt{n}}\sum_{k=1}^K\sum_{i \in \mathcal{L}_1, S_i=k} \psi_{f_0}^{(k)}(A_i, Y_i ; \tau)-\hat{I}(f_0) \cdot \frac{1}{\sqrt{n}}\sum_{k=1}^K\sum_{i \in \mathcal{L}_1, S_i=k} \psi_{f_0}^{(k)}(A_i, Y_i ; \tau)\right] \\
    :=&B_1+B_2.
\end{aligned}
\end{equation}
It suffices for \eqref{eqn:suf-lemma3} to show that $B_1 =\op$ and $B_2 = \op$.

In the proof of Theorem \ref{t4}, we have shown that $n^{-1/2}\sum_{k=1}^K\sum_{i: S_i=k} \psi_{f_0}^{(k)}(A_i, Y_i ; \tau)$ converges in distribution to a normal distribution with zero mean and finite variance. Since the splitting is independent of $S^{(n)}, A^{(n)}, Y^{(n)}(a)$, then $n^{-1/2}\sum_{k=1}^K\sum_{i \in \mathcal{L}_1, S_i=k} \psi_{f_0}^{(k)}(A_i, Y_i ; \tau) = O_P(1)$. By Lemma~\ref{l1}, we have $\hat I(f_0) = I(f_0) + \op$. Therefore, $B_2 = \op$.

Next, we show that $B_1 = \op$, which is implied by
\begin{eqnarray}\label{eqn:lemma3-var}
    &&\Var\bigg(\bigg\{\hat{I}(f_0) \cdot \frac{1}{\sqrt{n}}\sum_{k=1}^K\sum_{i \in \mathcal{L}_1, S_i=k}\psi_{\hat{f}_{0(2)}}^{(k)}(A_i, Y_i ; \tau) - I(f_0) \cdot \frac{1}{\sqrt{n}}\sum_{k=1}^K\sum_{i \in \mathcal{L}_1, S_i=k} \psi_{f_0}^{(k)}(A_i, Y_i ; \tau) \bigg\} \nonumber \\
    && \quad \mid S^{(n)},A^{(n)}, U^{(n)},D^{(\mathcal{L}_2)}\bigg) = \op
\end{eqnarray}
and 
\begin{eqnarray}\label{eqn:lemma3-bias}
    &&E\bigg(\bigg\{\hat{I}(f_0) \cdot \frac{1}{\sqrt{n}}\sum_{k=1}^K\sum_{i \in \mathcal{L}_1, S_i=k}\psi_{\hat{f}_{0(2)}}^{(k)}(A_i, Y_i ; \tau) - I(f_0) \cdot \frac{1}{\sqrt{n}}\sum_{k=1}^K\sum_{i \in \mathcal{L}_1, S_i=k} \psi_{f_0}^{(k)}(A_i, Y_i ; \tau) \bigg\} \nonumber \\
    &&\quad \mid S^{(n)},A^{(n)}, U^{(n)},D^{(\mathcal{L}_2)}\bigg) =  \op.
\end{eqnarray}

To prove \eqref{eqn:lemma3-var}, we have
\begin{equation}
    \nonumber
    \begin{aligned}
	&\Var\bigg(\bigg\{\hat{I}(f_0) \cdot \frac{1}{\sqrt{n}}\sum_{k=1}^K\sum_{i \in \mathcal{L}_1, S_i=k}\psi_{\hat{f}_{0(2)}}^{(k)}(A_i, Y_i ; \tau) - I(f_0) \cdot \frac{1}{\sqrt{n}}\sum_{k=1}^K\sum_{i \in \mathcal{L}_1, S_i=k} \psi_{f_0}^{(k)}(A_i, Y_i ; \tau) \bigg\} \\
    & \mid S^{(n)},A^{(n)}, U^{(n)},D^{(\mathcal{L}_2)}\bigg)\\
		=&\frac{1}{n}\Var\bigg(
		\bigg\{\sum_{k=1}^K \sum_{i \in \mathcal{L}_1,S_i=k, A_i =1} \frac{1}{\pi_{n[k]}} \cdot\left(\frac{\hat{f}_{0(2)}^{\prime}}{\hat{f}_{0(2)}}-\frac{f_0^{\prime}}{f_0}\right)\left(Y_i-\tau\right)\\
    &-\sum_{k=1}^K\sum_{i \in \mathcal{L}_1,S_i=k, A_i=0}\frac{1}{1-\pi_{n[k]}} \cdot\left(\frac{\hat{f}_{0(2)}^{\prime}}{\hat{f}_{0(2)}}-\frac{f_0^{\prime}}{f_0}\right)\left(Y_i\right)\bigg\} \mid S^{(n)},A^{(n)},U^{(n)},D^{(\mathcal{L}_2)} \bigg)\\
    \leq& \frac{2K}{n}\sum_{k=1}^K \Bigg[ \frac{1}{\pi_{n[k]}^2} \Var\bigg(\sum_{i \in \mathcal{L}_1,S_i=k, A_i =1} \cdot\left(\frac{\hat{f}_{0(2)}^{\prime}}{\hat{f}_{0(2)}}-\frac{f_0^{\prime}}{f_0}\right)\left(Y_i-\tau\right) \mid S^{(n)},A^{(n)},U^{(n)},D^{(\mathcal{L}_2)} \bigg)\\
    &+ \frac{1}{(1-\pi_{n[k]})^2}  \Var\bigg(\sum_{i \in \mathcal{L}_1,S_i=k, A_i =0} \cdot\left(\frac{\hat{f}_{0(2)}^{\prime}}{\hat{f}_{0(2)}}-\frac{f_0^{\prime}}{f_0}\right)\left(Y_i\right) \mid S^{(n)},A^{(n)},U^{(n)},D^{(\mathcal{L}_2)}\bigg)\Bigg].
\end{aligned}
\end{equation}

We have shown in \eqref{E_singlek_lemma2} and \eqref{E_singlek_lemma2_2} that
$$
 \Var\bigg(\sum_{i \in \mathcal{L}_1,S_i=k, A_i =0} \cdot\left(\frac{\hat{f}_{0(2)}^{\prime}}{\hat{f}_{0(2)}}-\frac{f_0^{\prime}}{f_0}\right)\left(Y_i\right) \mid S^{(n)},A^{(n)},U^{(n)},D^{(\mathcal{L}_2)}\bigg) = o_P(n),
$$
$$
 \Var\bigg(\sum_{i \in \mathcal{L}_1,S_i=k, A_i =1} \cdot\left(\frac{\hat{f}_{0(2)}^{\prime}}{\hat{f}_{0(2)}}-\frac{f_0^{\prime}}{f_0}\right)\left(Y_i - \tau \right) \mid S^{(n)},A^{(n)},U^{(n)},D^{(\mathcal{L}_2)}\bigg) = o_P(n).
$$
Since $\pi_{n[k]} \cp \pi$, then equation~\eqref{eqn:lemma3-var} holds.

To prove \eqref{eqn:lemma3-bias}, we have
\begin{eqnarray}
\label{e:lemma3-bias}
 &&E\bigg(\bigg\{\hat{I}(f_0) \sum_{k=1}^K\sum_{i \in \mathcal{L}_1, S_i=k}\psi_{\hat{f}_{0(2)}}^{(k)}(A_i, Y_i ; \tau) - I(f_0) \sum_{k=1}^K\sum_{i \in \mathcal{L}_1, S_i=k} \psi_{f_0}^{(k)}(A_i, Y_i ; \tau) \bigg\} \nonumber \\
    &&\quad \mid S^{(n)},A^{(n)}, U^{(n)},D^{(\mathcal{L}_2)}\bigg) \nonumber \\
& = & \sum_{k=1}^K \bigg(\frac{n_{[k]1(1)}}{\pi_{n[k]}}-\frac{n_{[k]0(1)}}{(1-\pi_{n[k]})}\bigg) \cdot \int \left(\frac{\hat{f}_{0(2)}^{\prime}}{\hat{f}_{0(2)}}-\frac{f_0^{\prime}}{f_0}\right)\left(y\right) dF_{[k]0}(y). \nonumber
\end{eqnarray}
Note that $n_{[k]1(1)}/ (n\pi_{n[k]}) = (n_{[k]1(1)}/n_{[k]1}) ( n_{[k]}/n ) = \ps /2 + O_P(n^{-1/2})$, and similarly, $n_{[k]0(1)}/\{n(1-\pi_{n[k]})\}=\ps /2 + O_P(n^{-1/2})$. Thus, $n_{[k]1(1)}/\pi_{n[k]} - n_{[k]0(1)}/(1-\pi_{n[k]}) = O_P(n^{1/2})$. We have shown by \eqref{eqn:FkOp} and the Cauchy--Schwarz inequality that, 
$
\int (\hat{f}_{0(2)}^{\prime}/\hat{f}_{0(2)}-f_0^{\prime}/f_0)\left(y\right) dF_{[k]0}(y) = o_P(1).
$
Therefore, equation \eqref{e:lemma3-bias} holds.
    
\end{proof}

\typeout{get arXiv to do 4 passes: Label(s) may have changed. Rerun}
\end{document}